\documentclass[letterpaper,12pt,titlepage,oneside,final]{report}
 
\newcommand{\href}[1]{#1} 

\usepackage{ifthen}
\newboolean{ElectronicVersion}
\setboolean{ElectronicVersion}{true} 

\usepackage[acronym,toc]{glossaries} 
\makeglossaries

\usepackage{amsmath,amssymb,amstext,mathrsfs,amsthm,bm} 
\usepackage[pdftex]{graphicx} 
\usepackage[bindingoffset=0.125in,centering,includeheadfoot,margin=1in]{geometry}

\ifthenelse{\boolean{ElectronicVersion}}{
    \usepackage[letterpaper=true,pdftex,bookmarks,pagebackref,
	plainpages=false, 
        pdfpagelabels=true, 
        pdftitle=Applications\ of\ Adiabatic\ Approximation\ to\ One-\ and\ Two-electron\ Phenomena\ in\ Strong\ Laser\ Fields, 
        pdfauthor=Denys\ Bondar	 
        ]{hyperref}}
     {\usepackage{url}} 

\DeclareFontFamily{OT1}{rsfs}{}
\DeclareFontShape{OT1}{rsfs}{m}{n}{<5> rsfs5 <7> rsfs7 <10> rsfs10}{}
\DeclareSymbolFont{mathrsfs}{OT1}{rsfs}{m}{n}
\DeclareSymbolFontAlphabet{\mathrsfs}{mathrsfs}

\newcommand{\h}{\hat{H}} 
\newcommand{\U}{\hat{U}}
\newcommand{\Hcl}{\mathrsfs{H}} 
\newcommand{\Pcl}{\boldsymbol{\mathrsfs{P}}} 

\renewcommand{\P}{\mathbf{p}}
\newcommand{\R}{\mathbf{r}}
\newcommand{\K}{\mathbf{k}}
\newcommand{\A}{\mathbf{A}}

\newcommand*{\bra}[1]{\left\langle{#1}\right|} 
\newcommand*{\ket}[1]{\left| #1 \right\rangle} 

\renewcommand{\Im}{{\rm Im}\,}
\renewcommand{\Re}{{\rm Re}\,}
\newcommand{\arccosh}{ \mathrm{arccosh}\,}
\newcommand{\arcsinh}{ \mathrm{arcsinh}\,}
\newcommand*{\supp}{{\rm supp}\,}

\newcommand{\kk}{{k_{\parallel}}}
\newcommand{\pp}{{p_{\parallel}}}

\renewcommand{\k}{{k_{\perp}}}
\newcommand{\p}{p_{\perp}}

\newtheorem{theorem}{Theorem}
\newtheorem{corollary}{Corollary}

\let\origdoublepage\cleardoublepage
\newcommand{\clearemptydoublepage}{%
  \clearpage{\pagestyle{empty}\origdoublepage}}
\let\cleardoublepage\clearemptydoublepage

\begin{document}

\pagestyle{empty}
\pagenumbering{roman}

\begin{titlepage}
        \begin{center}
        \vspace*{1.0cm}

        \Huge
        {\bf Applications of Adiabatic Approximation to One- and Two-electron Phenomena in Strong Laser Fields}

        \vspace*{1.0cm}

        \normalsize
        by \\

        \vspace*{1.0cm}

        \Large
        Denys Bondar \\

        \vspace*{3.0cm}

        \normalsize
        A thesis \\
        presented to the University of Waterloo \\ 
        in fulfillment of the \\
        thesis requirement for the degree of \\
        Doctor of Philosophy  \\
        in \\
        Physics \\

        \vspace*{2.0cm}

        Waterloo, Ontario, Canada, 2010 \\

        \vspace*{1.0cm}

        \copyright\ Denys Bondar 2010 \\
        \end{center}
\end{titlepage}

\pagestyle{plain}
\setcounter{page}{2}

\cleardoublepage 

  \noindent
I hereby declare that I am the sole author of this thesis. This is a true copy of the thesis, including any required final revisions, as accepted by my examiners.

  \bigskip
  
  \noindent
I understand that my thesis may be made electronically available to the public.

\cleardoublepage


\begin{center}\textbf{Abstract}\end{center}

The adiabatic approximation is a natural approach for the description of phenomena induced by low frequency laser radiation because the ratio of the laser frequency to the characteristic frequency of an atom or a molecule is a small parameter. Since the main aim of this work is the study of ionization phenomena, the version of the adiabatic approximation that can account for the transition from a bound state to the continuum must be employed. Despite much work in this topic, a universally accepted adiabatic approach of bound-free transitions is lacking. Hence, based on Savichev's modified adiabatic approximation [Sov. Phys. JETP {\bf 73}, 803 (1991)], we first of all derive the most convenient form of the adiabatic approximation for the problems at hand. Connections of the obtained result with the quasiclassical approximation and other previous investigations are discussed. Then, such an adiabatic approximation is applied to single-electron ionization and non-sequential double ionization of atoms in a strong low frequency laser field.

The momentum distribution of photoelectrons induced by single-electron ionization is obtained analytically without any assumptions on the momentum of the electrons. Previous known results are derived as special cases of this general momentum distribution. 

The correlated momentum distribution of two-electrons due to non-sequential double ionization of atoms is calculated semi-analytically. We focus on the deeply quantum regime -- the below intensity threshold regime, where the energy of the active electron driven by the laser field is insufficient to collisionally ionize the parent ion, and the assistance of the laser field is required to create a doubly charged ion. A special attention is paid to the role of Coulomb interactions in the process. The signatures of electron-electron repulsion, electron-core attraction, and electron-laser interaction are identified. The results are compared with available experimental data.

Two-electron correlated spectra of non-sequential double ionization below intensity threshold are known to exhibit back-to-back scattering of the electrons, viz., the anticorrelation of the electrons. Currently, the widely accepted interpretation of the anticorrelation is recollision-induced excitation of the ion plus subsequent field ionization of the second electron. We argue that there exists another mechanism, namely simultaneous electron emission, when the time of return of the rescattered electron is equal to the time of liberation of the bounded electron (the ion has no time for excitation), that can also explain the anticorrelation of the electrons in the deep below intensity threshold regime. 

Finally, we study single-electron molecular ionization. Based on the geometrical approach to tunnelling by P. D. Hislop and I. M. Sigal [Memoir. AMS 78, No. 399 (1989)], we introduce the concept of a leading tunnelling trajectory. It is then proven that leading tunnelling trajectories for single active electron models of molecular tunnelling ionization (i.e., theories where a molecular potential is modelled by a single-electron multi-centre potential) are linear in the case of short range interactions and ``almost'' linear in the case of long range interactions. The results are presented on both the formal and physically intuitive levels. Physical implications of the proven statements are discussed. 

\cleardoublepage


\begin{center}\textbf{Acknowledgements}\end{center}

First and foremost, I am thankful to my Mom, Ganna~V.~Bondar, and my Dad, Ivan~I.~Bondar, for their infinite support, without which nothing could have been done. 

I thank my art mentor, Oleksandr~V.~Sydoruk, who has inculcated me with the 	joy and love of creation of all types: art, music, science, {\it etc.} He also introduced me to the legacy of Leonardo da Vinci, whose immense multifarious talents have been the source of inspiration to seek knowledge. 

I am indebted to my friend, Robert R. Lompay, for being my physics mentor and for encouraging me to pursue doctoral studies.  His ideas regarding what physics is and how it ought to be done have had a tremendous impact on my personality. 

I am grateful to my supervisors, Misha~Yu.~Ivanov and Wing-Ki~Liu, for a wonderful opportunity not only to do research, but also to live in the best country in the world -- Canada! Moreover, they gave me the absolute freedom  to explore and supported me working on many other projects unrelated to this thesis. 

Misha's deep physical intuition has always  fascinated me. A special acknowledgment to Misha for teaching me to think physically as well as elucidating  the difference between physics and the tools employed in physics. He also has given me a cornerstone opportunity to work at the Theory and Computation Group of National Research Council of Canada in Ottawa, where I have met many wonderful people.  

Many thanks to Michael~Spanner for endless discussions of physics, encouragement, and substituting Misha when he left Ottawa; Ryan~Murray for discussions ranging from physics to finance and many valuable ideas; Serguei~Patchkovskii for his peculiar mentorship that has turned out to be very fruitful for me; Gennady~L.~Yudin for important and interesting collaborations and fervent discussions; Olga~Smirnova for masterly teaching the Lipmann-Schwinger equation and the S-matrix technique -- two theoretical pillars of strong field physics; Zi-Jian~Long for always being there whenever I needed help. 

I am also thankful to Michael, Serguei, and Ryan for training me in numerical methods.

Last but not least, I am obliged to my dear wife, Gopika~K.~Sreenilayam, for her motivation and patience when I was working during a countless number of evenings and weekends. 

This research was financially supported by the Ontario Graduate Scholarship program. 

\cleardoublepage


\begin{center}\textbf{Dedication}\end{center}

\begin{center} To my parents and my wife \end{center}

\cleardoublepage

\renewcommand{\contentsname}{Table of Contents}
\renewcommand{\acronymname}{List of Acronyms}

\tableofcontents
\cleardoublepage


\listoffigures
\addcontentsline{toc}{chapter}{List of Figures}
\cleardoublepage

\newacronym{SFA}{SFA}{strong field approximation}
\newacronym{NSDI}{NSDI}{nonsequential double ionization}
\newacronym{PPT}{PPT}{Perelomov-Popov-Terent'ev  \cite{Perelomov_1966, Perelomov_1967_A, Perelomov_1967_B, Perelomov_1968}}
\newacronym{SF-EVA}{SF-EVA}{strong-field eikonal-Volkov
approach \cite{Smirnova_2008}}
\newacronym{BIT}{BIT}{below intensity threshold}
\newacronym{SEE}{SEE}{simultaneous electron emission}
\newacronym{RESI}{RESI}{recollision-induced excitation of the ion plus subsequent field ionization of the second electron}

\printglossaries
\cleardoublepage


\pagenumbering{arabic}

\chapter{Introduction}\label{chapter1}

In this thesis we explore one- and two-electron ionization phenomena in a strong laser field. Since the ratio $\omega/\omega_{at} \sim 0.1$, where $\omega$ is, say, the frequency of Ti:sapphire lasers -- most commonly used tabletop laser systems in laboratories around the world \cite{Paschotta2008}, and  $\omega_{at}$ is the characteristic atomic frequency, the adiabatic approximation ought to be a perfectly suitable theoretical tool for the description of phenomena induced by such low-frequency laser fields.

Roughly speaking, the adiabatic approximation can be introduced as follows. Once the frequency of the external laser field is much lower than the characteristic atomic frequency, $\omega \ll \omega_{at}$, an approximate solution of the Schr\"{o}dinger equation can be found by means of averaging over atomic (internal) degrees of freedom. Therefore, the adiabatic approximation is a method of constructing an asymptotic expansion of the solution of the non-stationary Schr\"{o}dinger equation in terms of the small parameter $\omega/\omega_{at}$.

Born and Fock \cite{Born1928} founded the theory of the adiabatic approximation for a discrete spectrum by formulating the adiabatic theorem. Landau \cite{Landau1932, Landau1932a} estimated the probability of nonadiabatic transitions between discreet states. However, the leading-order asymptotic result for such a quantity was obtained by Dykhne \cite{Dykhne_1962}. Afterwards,  nonadiabatic transitions between discreet states were thoroughly analyzed by many authors \cite{Davis_1976, Solovev1976a, Schilling_2006} (for reviews see References \cite{ENikitin1984, Hiroki2002, Osherov2004}). Summarizing research in this area, one may conclude that nonadiabatic transitions in discreet spectra are quite well studied. 

However, transitions from a discreet state to  a continuum have been the subject of on-going investigations for many decades. Despite much work on this topic, a universally accepted approach is lacking. The direct generalization of the Dykhne method was developed by Chaplik \cite{Chaplik_1964, Chaplik1965a}. Exactly solvable models were reported by Demkov and Oserov \cite{Demkov1966, DEMKOV1968},  
Ostrovskii \cite{Ostrovskii1972}, Ostrovskii {\it et al.} \cite{Demkov2001, Volkov2004}, and Nikitin \cite{Nikitin1962a, Nikitin1962} (for review see, e.g., Reference \cite{Nikitin1970}). The advanced adiabatic approach was introduced and widely employed by Solov'ev \cite{Solovev1976, Solovev1989, Solovev2005}. Finally, Tolstikhin recently developed a promising version of the adiabatic approximation for the transitions to the continuum \cite{Tolstikhin2008a}  by employing the Siegert-state expansion for nonstationary quantum systems \cite{Tolstikhin2006a, Tolstikhin2006, Tolstikhin2008}. Yet, the approach presented in Reference \cite{Tolstikhin2008a} is limited to finite range potentials. 

The adiabatic approximation deserves special attention because it is one of a few known non-perturbative approaches in quantum mechanics. Strong field physics demands for such methods to describe nonlinear phenomena induced by a strong laser field. 

In Chapter \ref{chapter2}, we derive the amplitude of nonadiabatic transitions from a bound state to the continuum within the Savichev modified adiabatic approximation \cite{Savichev1991a}. Then, the properties of this amplitude are scrutinized, and a connection between the adiabatic and quasiclassical approximations is established. The obtained amplitude is the main theoretical formalism of this work, which is applied to strong field one- and two-electron processes in subsequent chapters.

The most common process that occurs when laser radiation is exerted on an atom is single-electron ionization. Although initial theoretical understanding of strong field ionization was put forth by Keldysh \cite{Keldysh_1965} as early as 1964, many questions remain unresolved.  As far as single-electron ionization in the presence of a linearly polarized laser field is concerned, there are two important topics. The first, namely the ionization rate as a function of instantaneous laser phase, was studied in depth by Yudin and Ivanov \cite{Yudin_2001B} (see also References \cite{Kienberger_2007,Uiberacker_2007}).  However, their result assumes zero initial momentum of the liberated electron. Effects due to nonzero initial momentum have yet to be included.  The second topic pertains to the single-electron spectra, that is, the ionization rate as a function of the final momentum of the electron.  Despite much work on this topic (see, for example, References \cite{Krainov_2003, Popov2004, Goreslavskii_2005} and references therein) a universal formula is absent and the discussion is still ongoing.  Among the most accurate results, Goreslavskii {\it et al.} \cite{Goreslavskii_2005} have obtained an expression for the complete single-electron ionization spectrum, but without consideration of the laser phase. In Chapter \ref{chapter3}, we derive a more general formula that includes the dependences on both the instantaneous laser phase and the final electron momentum.

In strong low-frequency laser fields, following one-electron ionization
of an atom or a molecule, the liberated electron can recollide
with the parent ion \cite{Kuchiev_1987, Corkum_1993}. The electron
acts as an ``atomic antenna'' \cite{Kuchiev_1987}, absorbing the
energy from the laser field between ionization and recollision and
depositing it into the parent ion. Inelastic scattering on the
parent ion results in further collisional excitation and/or
ionization. Liberation of the second electron during the
recollision -- the laser-induced e-2e process -- is known as \gls{NSDI}. 

The phenomenon of \gls{NSDI} was experimentally discovered by Suran and Zapesochny \cite{Suran1975} for alkaline-earth atoms
(for further experimental investigations of \gls{NSDI} for alkaline-earth atoms, see, e.g., References  \cite{Bondar1993, Bondar1998, Bondar2000, Liontos2004, Liontos2008, Liontos2010}). In this case, autoionizing double excitations below the second ionization threshold were shown to be extremely important. For a theoretical study of these effects, see, e.g., Reference \cite{Lambropoulos1988}. For noble gas atoms, nonsequential double ionization was first observed by L'Huillier {\it et al.} (see, e.g., References \cite{LHuillier1982, LHuillier1983}). The interest to the phenomenon of \gls{NSDI} grew rapidly after it was rediscovered in 1993-1994 \cite{Walker1993, Walker1994}. Recently,
correlated multiple ionization has also been observed
\cite{Rudenko_2004, Zrost_2006}. The renewed interest in \gls{NSDI} has been enhanced by the availability of new experimental techniques that allow one to perform accurate measurements of the angle- and energy-resolved spectra of the photoelectrons, in coincidence. Such measurements play a crucial role in elucidating the physical mechanisms of the \gls{NSDI}. 

From the theoretical perspective, direct {\it ab initio} simulations of
the photoelectron spectra corresponding to \gls{NSDI} in intense low-frequency laser
fields represent a major challenge. Only now such benchmark
simulations have become possible \cite{Taylor2007} for
the typical experimental conditions (the helium atom, laser
intensity $I\sim 10^{15}$ W/cm$^2$, laser wavelength $\lambda =
800$ nm). In addition to these calculations, tremendous insight into the physics of the problem
has been obtained from classical simulations performed in References
\cite{Panfili_2002, Ho_2005, Phay2005, Ho2006, Haan2006, Haan2008}. These papers have demonstrated a  variety of the
regimes of nonsequential double and triple ionization. Not only
do these simulations reproduce key features observed in the
experiment, they also give a clear view of the (classical)
interplay between the two electrons, the potentials of the laser
field, and of the ionic core. They also show how different types of
the correlated motion of the two electrons contribute to different
parts of the correlated two-electron spectra.

The physics of double ionization is different for
different intensity regimes, separated by the ratio of the energy
of the recolliding electron to the binding (or excitation) energy
of the second electron, bound in the ion.

According to classical considerations, the maximum energy which the recolliding electron can acquire from the laser field is $\sim 3.2 U_p$ \cite{Corkum_1993}, where
$U_p=(F/2\omega)^2$, $F$ is the laser field strength, and $\omega$
is the laser frequency (unless stated otherwise atomic units, $\hbar=m_e=|e|=1$, are used throughout the work). Hence, \gls{NSDI} can be divided into two types: 
if the intensity of the driving laser field is such that the recolliding electron gains enough kinetic energy to collisionally ionize the parent ion -- the above intensity threshold regime, and when such kinetic energy is insufficient to directly ionize the ion -- the \gls{BIT} regime. The former regime is thoroughly studied experimentally as well as theoretically (see, e.g., References \cite{Lein2000, Yudin_2001A, deMorissonFaria2003, deMorissonFaria2004, deMorissonFaria2004B, deMorissonFaria2005, Becker_2005, Liu_2006,  Rudenko_2007,
 Staudte_2007, deMorissonFaria2008b, deMorissonFaria2008a} and references therein).

\gls{NSDI} \gls{BIT}, being a most challenging regime, is currently of an active experimental interest \cite{Eremina2003, Rudenko_2004, Weckenbrock_2004, Zeidler_2005, Zrost_2006, Liu_2008, Liu2010}. In this regime, existing classical and quantum analysis
(see, e.g., References \cite{Haan2006, Ho2006, deMorissonFaria2006, Haan2008a, Emmanouilidou2009, Haan2010, Shaaran2010, Shaaran2010a, Ye2010}) demonstrates two possibilities of electron ejection
after the recollision. First, the two electrons can be ejected
with little time delay compared to the quarter-cycle of the
driving field. Second, the time delay between the ejection of the
first and the second electron can approach or exceed the
quarter-cycle of the driving field. In these two cases, the
electrons appear in different quadrants of the correlated
spectrum. If, following the recollision, the electrons are ejected
nearly simultaneously, their parallel momenta have equal signs,
and both electrons are driven by the laser field in the same
direction toward the detector. If, following the recollision, the
electrons are ejected with a substantial delay (quarter-cycle or
more), they end up going in the opposite directions exhibiting the phenomenon of anticorrelation of the electrons. Thus, these
two types of dynamics leave distinctly different traces in the
correlated spectra.

In Chapters \ref{chapter4}, we develop a fully quantum, analytical treatment of
\gls{NSDI} \gls{BIT}. It is important that our approach takes into account
all relevant interactions -- those with the laser field, the ion,
and between the electrons -- nonperturbatively. The case in which the two electrons are ejected simultaneously, i.e., the process of \gls{SEE},  is considered. We show that in this case the correlated
spectra bear clear signatures of the electron-electron and
electron-ion interactions after ionization,  including the
interplay of these interactions. These signatures are identified. In agreement with previous studies, the mechanism of \gls{SEE} manifests itself in the correlation of the electrons -- the electrons are moving in the same direction after \gls{NSDI}.

However, in Chapters \ref{chapter5}, we demonstrate that if the intensity of the laser field is lowered such that we enter the deep \gls{BIT} regime of \gls{NSDI}, \gls{SEE} can be responsible for the anticorrelation of the electrons. This novel mechanism is alternative to the widely accepted point of view that the anticorrelation of the electrons are caused by \gls{RESI}. Nevertheless, SEE and RESI are by no means mutually exclusive processes; they both contribute to the complex and diverse phenomenon of \gls{NSDI} \gls{BIT}.

Recent advances in experimental investigations of single-electron molecular ionization in a low frequency strong laser field \cite{Talebpour1996, Guo1998, DeWitt2001, Wells2002, Litvinyuk2003, Pavivcic2007, Kumarappan2008, Staudte2009} have created a demand for a theory of this phenomenon \cite{MuthBohm2000, Tong2002, Awasthi2006, Vanne2008, Fabrikant2009, Usachenko2009, Bin2010, Fabrikant2010, Murray2010, Spanner2010, Walters2010}. As far as low frequency laser radiation is concerned, one can ignore the time-dependence of the laser and consider the corresponding static picture, which is obtained as $\omega\to 0$. In this limit, single electron molecular ionization is realized by quantum tunnelling. This approximation is valid from qualitative and quantitative points of view, and it tremendously simplifies the theoretical analysis of the problem at hand. Such single active electron approaches to molecular ionization, where an electron is assumed to interact with multiple centres that model the molecule and a static field that models the laser, are among most popular. Analytical and semi-analytical versions of these methods, which are based on the quasiclassical approximation \cite{Tong2002, Fabrikant2009, Bin2010, Fabrikant2010, Murray2010} are indeed quite successful in interpreting and explaining available experimental data. However, these quasiclassical theories heavily relay on the assumption that the electron tunnels along a straight trajectory. Despite its wide use, the reliability of this conjecture has not been verified.

In Chapter \ref{chapter6}, we study the reliability of this hypothesis. Relying on the geometrical approach to many-dimensional tunnelling by Hislop and Sigal \cite{Sigal1988, Sigal1988a, Hislop1989a, Hislop1996}, which is a mathematically rigorous reformulation of the instanton method, we first introduce the notion of leading tunnelling trajectories. Then, we analyze their shapes in the context of single active electron molecular tunnelling. It will be rigorously proven that the assumption of ``almost'' linearity of leading tunnelling trajectories is satisfied in almost all the situations of practical interest. Such results justify the above mentioned models and  open new ways of further development of quasiclassical approaches to molecular ionization.

\glsresetall
 
\chapter{Adiabatic Approximation}\label{chapter2}

\section{General Discussion}

Our investigations are based on a seminal result that ought to be summarized foremost. Following the Solov'ev advanced adiabatic approach \cite{Solovev1976, Solovev1989, Solovev2005}, Savichev \cite{Savichev1991a} proved the following. If the adiabatic state $\ket{\psi_i(\varphi)}$ and the corresponding adiabatic term $E_i(\varphi)$,
 \begin{eqnarray}\label{Ch2_AdiabaticTermStateFullH}
 \h(\varphi)\ket{\psi_i (\varphi)} = E_i(\varphi)\ket{\psi_i(\varphi)},
 \end{eqnarray}
are known, then the solution $\ket{\Psi(t)}$ of the nonstationary Schr\"{o}dinger equation  
\begin{equation}\label{Ch2_FullSchEq}
 i\partial \ket{\Psi(t)}/\partial t = \h(\varphi)\ket{\Psi(t)},
 \end{equation}
 where $\varphi = \omega t$ is a phase of the laser field,
subjected to the initial condition
\begin{eqnarray}\label{Ch2_Savichev_initial_condition}
\ket{\Psi(t)} \xrightarrow[t\to -\infty]{} \ket{\psi_i(\varphi_i(E_i))}e^{ - i\int^t E_i(\omega\tau)d\tau} [ 1+ O(\omega)],
\end{eqnarray}
has the following form within the adiabatic approximation ($\omega \ll 1$)
\begin{eqnarray}\label{Ch2_SavichevEq}
 \ket{\Psi(t)} &=& \frac 1{2\pi\omega}\iint dEd\varphi'  \, \ket{\psi_i (\varphi_i(E))}\times \nonumber\\
&& \exp\left[ \frac i{\omega}\left(E\varphi' - \int^{\varphi'}E_i(\varphi)d\varphi \right)-iEt\right][ 1 + O(\omega)].
\end{eqnarray}
Here $\varphi_i(E)$ is the analytical continuation of the inverse function of $E_i(\varphi)$. Note that no assumptions\footnote{
Besides some tacit requirements such as the analyticity of both the adiabatic terms and the states.
}
on a form of the Hamiltonian $\h$ are required to derive Equation (\ref{Ch2_SavichevEq}) from the nonstationary Schr\"{o}dinger equation (\ref{Ch2_FullSchEq}).

The integral over $E$ in Equation (\ref{Ch2_SavichevEq}) can be interpreted as a generalization of the Born-Fock expansion to the case of a system with a continuos spectrum. The integral over $\varphi'$ in Equation (\ref{Ch2_SavichevEq}) can be calculated by the saddle point approximation without changing the accuracy of Equation (\ref{Ch2_SavichevEq}); however, the original integral representation is more advantageous and should be left unaltered.  

For the sake of completeness and clarity, we shall present the derivation of Equation (\ref{Ch2_SavichevEq}). Let us seek the solution of the Schr\"{o}dinger equation (\ref{Ch2_FullSchEq}) in the form
\begin{eqnarray}\label{Ch2_Anzatse_for_Psi}
	 \ket{\Psi(t)} = \int dE \, g(E) \exp(-iEt) \ket{\psi_i (\varphi_i(E))},
\end{eqnarray}
where the unknown function $g(E)$ can be represented as
\begin{eqnarray}\label{Ch2_Anzatse_for_g}
	g(E) = \frac 1{2\pi} \int dt \, \exp(iEt) G(t, E), \qquad G(t,E) = \langle \psi_i (\varphi_i(E)) \ket{\Psi(t)},
\end{eqnarray}
where the normalization condition $\langle \psi_i (\varphi) \ket{\psi_i (\varphi)} = 1$ is assumed for almost all $\varphi$.
To derive the equation for $G(t, E)$, we recall that  $\ket{\Psi(t)}$ satisfies the Schr\"{o}dinger equation (\ref{Ch2_FullSchEq}); hence,
\begin{eqnarray}
	 && \bra{\psi_i (\varphi_i(E))}  i\frac{\partial}{\partial t} \ket{\Psi(t)} = \bra{\psi_i (\varphi_i(E))} \hat{H}(\varphi) \ket{\Psi(t)}, \nonumber\\
	  && i\frac{\partial}{\partial t} G(t, E) = \bra{\psi_i (\varphi_i(E))}  i\frac{\partial}{\partial t} \ket{\Psi(t)}. \label{Ch2_Eq_for_G}
\end{eqnarray}
Having substituted Equations (\ref{Ch2_Anzatse_for_Psi}) and (\ref{Ch2_Anzatse_for_g}) into the system of equations (\ref{Ch2_Eq_for_G}), we get
\begin{eqnarray}
	&& \iint dE' dt' \, \exp[ iE'(t'-t) ] G(t', E') E' \langle \psi_i (\varphi_i(E)) \ket{ \psi_i (\varphi_i(E')) }  \nonumber\\
	&& = \iint dE' dt' \, \exp[ iE'(t'-t) ] G(t', E') \bra{\psi_i (\varphi_i(E))} \hat{H}(\varphi) \ket{ \psi_i (\varphi_i(E')) }, \nonumber\\
	&& i \frac{\partial}{\partial t} G(t, E) = \frac{1}{2\pi} \iint dE' dt' \, \exp[iE'(t'-t)] G(t', E') E' \langle \psi_i (\varphi_i(E)) \ket{ \psi_i (\varphi_i(E')) }. \nonumber
\end{eqnarray}
Introducing the variables $\varphi = \omega t$ and $\varphi' = \omega t'$, we have 
\begin{eqnarray}
	&& \iint dE' d\varphi' \, \exp\left[ \frac{i}{\omega} E'(\varphi'-\varphi) \right] G(\varphi', E') E' \langle \psi_i (\varphi_i(E)) \ket{ \psi_i (\varphi_i(E')) }  \nonumber\\
	&& = \iint dE' d\varphi' \, \exp\left[ \frac{i}{\omega} E'(\varphi'-\varphi) \right] G(\varphi', E') \bra{\psi_i (\varphi_i(E))} \hat{H}(\varphi) \ket{ \psi_i (\varphi_i(E')) }, \nonumber\\
	&& i \omega^2 \frac{\partial}{\partial \varphi} G(\varphi, E) 
		= \frac{1}{2\pi} \iint dE' d\varphi' \, \exp\left[\frac{i}{\omega} E'(\varphi'-\varphi)\right] G(\varphi', E') E' \langle \psi_i (\varphi_i(E)) \ket{ \psi_i (\varphi_i(E')) }. \nonumber
\end{eqnarray}
Substituting the following ansatz into the equations above 
\begin{eqnarray}\label{Ch2_saddlepoint_anzatse_for_G}
	G(\varphi, E) = \left[ X_0 (E, \varphi) + \omega X_1(E, \varphi) + \cdots \right] \exp\left[ -\frac{i}{\omega} \int^{\varphi} Q(\varphi') d\varphi' \right],
\end{eqnarray}
then performing integration over $E'$ and $t'$ by means of the saddle point approximation and collecting terms in front of the zeroth power of $\omega$, we obtain
\begin{eqnarray}
	&& \langle \, \psi_i (\varphi_i(E)) \, \ket{ \, \psi_i (\varphi_i(Q(\varphi))) \,} Q(\varphi) = \bra{\, \psi_i (\varphi_i(E)) \,} \hat{H}(\varphi) \ket{ \, \psi_i (\varphi_i(Q(\varphi))) \,}, \nonumber\\
	&& X_0 (E, \varphi) = X_0 (Q(\varphi), \varphi) \langle \, \psi_i (\varphi_i(E)) \, \ket{ \, \psi_i (\varphi_i(Q(\varphi))) \,}.
\end{eqnarray}
Whence, 
\begin{eqnarray}\label{Ch2_expressions_for_Q_and_X0}
	Q(\varphi) = E_i (\varphi), \quad X_0 (E, \varphi) = \langle \, \psi_i (\varphi_i(E)) \, \ket{ \, \psi_i (\varphi_i(E_i(\varphi))) \,} 
		= \langle \psi_i (\varphi_i(E)) \ket{ \psi_i (\varphi) }.
\end{eqnarray}
Substituting Equations (\ref{Ch2_saddlepoint_anzatse_for_G}) and (\ref{Ch2_expressions_for_Q_and_X0}) into Equation (\ref{Ch2_Anzatse_for_g}), we have
\begin{eqnarray}\label{Ch2_almost_final_g}
	g(E) = \frac{1}{2\pi \omega} \int d\varphi \, \langle \psi_i (\varphi_i(E)) \ket{ \psi_i (\varphi) }
			\exp\left[ \frac{i}{\omega} \left( E\varphi - \int^{\varphi} E_i(\varphi') d\varphi' \right) \right] \left[ 1 + O(\omega) \right].
\end{eqnarray}
According to the saddle point approximation, the only neighbourhoods of importance in the integral (\ref{Ch2_almost_final_g}) are those where the derivative of the exponent with respect to $\varphi$ vanishes. In these neighbourhoods the matrix element $\langle \psi_i (\varphi_i(E)) \ket{ \psi_i (\varphi) } = 1$; therefore, the integral (\ref{Ch2_almost_final_g}) is equivalent, up to a term of order of $\omega$, to the expression
\begin{eqnarray}
	g(E) = \frac{1}{2\pi \omega} \int d\varphi \,  \exp\left[ \frac{i}{\omega} \left( E\varphi - \int^{\varphi} E_i(\varphi') d\varphi' \right) \right] \left[ 1 + O(\omega) \right].
\end{eqnarray}
Recalling Equation (\ref{Ch2_Anzatse_for_Psi}), we conclude that Equation (\ref{Ch2_SavichevEq}) is finally derived. Note that a proper choice of complex integration paths for the integrals in Equation (\ref{Ch2_SavichevEq}) will ensure that the wave function (\ref{Ch2_SavichevEq}) indeed satisfies the initial condition (\ref{Ch2_Savichev_initial_condition}). (A more detailed version of the derivation of Equation (\ref{Ch2_SavichevEq}) presented above can be found in Reference \cite{Savichev1991a}.)

However, one must be cautious regarding Equation (\ref{Ch2_SavichevEq}). As far as rigorous asymptotic analysis is concerned, it is incorrect to assume that the wave function (\ref{Ch2_SavichevEq}) obeys the Schr\"{o}dinger equation (\ref{Ch2_FullSchEq}) even though Equation (\ref{Ch2_SavichevEq}) was obtained from Equation (\ref{Ch2_FullSchEq}). On the contrary, the validity of the solution (\ref{Ch2_SavichevEq}) must be verified independently.

Employing the equality 
$$
\h(\varphi_i(E)) \ket{\psi_i (\varphi_i(E))} = E_i(\varphi_i(E)) \ket{\psi_i (\varphi_i(E))} = E\ket{\psi_i (\varphi_i(E))},
$$
and substituting Equation (\ref{Ch2_SavichevEq}) into Equation (\ref{Ch2_FullSchEq}), one readily shows that
\begin{eqnarray}
i\frac{\partial}{\partial t}  \ket{\Psi(t)}  &=& \frac 1{2\pi\omega} \iint dEd\varphi'  \,\, \h(\varphi_i(E)) \,\, \ket{\psi_i (\varphi_i(E))}\times \nonumber\\
&& \exp\left[ \frac i{\omega}\left(E\varphi' - \int^{\varphi'}E_i(\varphi)d\varphi \right)-iEt\right][ 1 + O(\omega)], \\
\h(\varphi) \ket{\Psi(t)}  &=&  \frac 1{2\pi\omega} \iint dEd\varphi'  \,\, \h(\varphi) \,\, \ket{\psi_i (\varphi_i(E))}\times \nonumber\\
&& \exp\left[ \frac i{\omega}\left(E\varphi' - \int^{\varphi'}E_i(\varphi)d\varphi \right)-iEt\right][ 1 + O(\omega)].
\end{eqnarray}
Hence, in order to prove that  the wave function (\ref{Ch2_SavichevEq}) indeed satisfies Equation (\ref{Ch2_FullSchEq}), we ought to demonstrate that $\h(\varphi) \ket{\psi_i(\varphi')} \approx \h(\varphi')\ket{\psi_i(\varphi')}$.

Using the analyticity of $\h(\varphi)$ and $\ket{\psi_i(\varphi)}$, we write
\begin{eqnarray}
\h(\varphi+\varepsilon)\ket{\psi_i(\varphi)} = \sum_{n=0}^{\infty} \frac{\varepsilon^n}{n!} \left( \frac{d^n\h(\varphi)}{d\varphi^n}\right) \ket{\psi_i(\varphi)}. \nonumber
\end{eqnarray}
Since 
\begin{eqnarray}
\left( \frac{d^n\h(\varphi)}{d\varphi^n}\right) \ket{\psi_i(\varphi)} &=& \left[ \frac{d}{d\varphi} \left( \frac{d^{n-1}\h(\varphi)}{d\varphi^{n-1}}\right) \right] \ket{\psi_i(\varphi)} \nonumber\\
&=& \frac{d}{d\varphi} \left[ \left( \frac{d^{n-1}\h(\varphi)}{d\varphi^{n-1}}\right)\ket{\psi_i(\varphi)} \right] - \frac{d^{n-1}\h(\varphi)}{d\varphi^{n-1}} \left( \frac{d}{d\varphi} \ket{\psi_i(\varphi)} \right)\nonumber\\
&=& \left[ \frac{d}{d\varphi}, \frac{d^{n-1}\h(\varphi)}{d\varphi^{n-1}} \right]\ket{\psi_i(\varphi)}, \nonumber
\end{eqnarray}
we obtain 
\begin{eqnarray}
\h(\varphi+\varepsilon)\ket{\psi_i(\varphi)} &=& \sum_{n=0}^{\infty} \frac{1}{n!} 
\underbrace{\left. \left[ \varepsilon\frac{d}{d\varphi}, \left[\varepsilon\frac{d}{d\varphi}, \ldots, \right[  \right. \right. }_\text{n times}
\left. \left. \left. \varepsilon\frac{d}{d\varphi} , \h(\varphi)\right] \ldots \right]\right] \ket{\psi_i(\varphi)} \nonumber\\
&=& \exp\left( \varepsilon\frac{d}{d\varphi} \right) \h(\varphi) \exp\left( -\varepsilon\frac{d}{d\varphi} \right) \ket{\psi_i(\varphi)}.
\end{eqnarray}
Whence, we lamentably observe that the wave function (\ref{Ch2_SavichevEq}) does not obey the Schr\"{o}dinger equation (\ref{Ch2_FullSchEq}) in a general case. Nevertheless, if
\begin{eqnarray}\label{Ch2_CondValiditySavichev}
\left[ \frac{d}{d\varphi},  \h(\varphi) \right] = o(1),
\end{eqnarray}
then 
\begin{eqnarray}
\h(\varphi + \varepsilon) \ket{\psi_i(\varphi)} = \h(\varphi) \ket{\psi_i(\varphi)} + o(1),
\end{eqnarray}
i.e., Equation (\ref{Ch2_CondValiditySavichev}) is the assumption on the form of the Hamiltonian  that guarantees that the wave function
\begin{eqnarray}\label{Ch2_SavichevEq_Corrected}
 \ket{\Psi(t)} &=& \frac 1{2\pi\omega}\iint dEd\varphi'  \, \ket{\psi_i (\varphi_i(E))}\times \nonumber\\
&& \exp\left[ \frac i{\omega}\left(E\varphi' - \int^{\varphi'}E_i(\varphi)d\varphi \right)-iEt\right][ 1 + o(1)]
\end{eqnarray}
is indeed a solution of the Schr\"{o}dinger equation (\ref{Ch2_FullSchEq}). [Note a minor difference between Equations (\ref{Ch2_SavichevEq}) and (\ref{Ch2_SavichevEq_Corrected}).]

What are the implications of the condition (\ref{Ch2_CondValiditySavichev}) for strong field physics? Let us specify the Hamiltonian. A typical hamiltonian of a system interacting with an external laser field reads (in the length gauge) 
$ 
\h(\varphi) = \hat{\P}^2/2 + V(\R) + \R\cdot{\bf F} (\omega, \varphi)
$, 
where ${\bf F} (\omega, \varphi)$ denotes the laser pulse. Then, the following 
\begin{eqnarray}\label{Ch2_LimitConditioOnLaserField}
\lim_{\omega\to 0} {\bf F} (\omega, \varphi) = \lim_{\omega\to 0} d{\bf F} (\omega, \varphi)/d\varphi = {\bf 0},  
\end{eqnarray}
would satisfy Equation (\ref{Ch2_CondValiditySavichev}). A strictly periodic laser pulse, such as $ {\bf F} (\omega, \varphi) = {\bf F}_L \sin\varphi$, does not satisfy the condition (\ref{Ch2_LimitConditioOnLaserField}). However, a pulse with a Gaussian envelope, e.g., $ {\bf F} (\omega, \varphi) = {\bf F}_L \exp\left[-(\frac{\varphi}{\omega\tau})^2\right] \sin\varphi$, obeys it. Therefore, the condition (\ref{Ch2_CondValiditySavichev}) demands that the laser field used has to have an envelope, which is a realistic and, perhaps, even tacit requirement. 

Due to the connection (see References \cite{Savichev1991a}) between the Savichev adiabatic approach and the Solov'ev advanced adiabatic method, the condition (\ref{Ch2_LimitConditioOnLaserField}) also applies to the latter method. 

Finally, it is important to note that results obtained within different versions of the adiabatic approximation are in fact the same. This follows from the uniqueness of the asymptotic expansion (see, e.g., Reference  \cite{Olver1974}). Furthermore, the results are also gauge independent. Hence, the choice of the gauge as well as the choice of the specific adiabatic method is merely the issue of convenience.

\section{The Derivation of the Amplitude of Non-adiabatic Transitions}

Let $\ket{i}$ and $\ket{f}$ be stationary states (for specification see Equation (\ref{Ch2_IniConditions}) and the comment after), and we shall assume that the quantum system with the Hamiltonian $\h$ is in the state $\ket{i}$ at $t=-\infty$. The main aim of this section is to obtain the general form of the transition amplitude $\mathfrak{M}_{i\to f}$ that the given quantum system will be found in the state $\ket{f}$ at $t=+\infty$.

Before going further, we are to introduce notations. First, we arbitrarily partition the Hamiltonian $\h$:
\begin{equation}\label{Ch2_Partition}
\h(\varphi) \equiv \h_0(\varphi) + \hat{V}(\varphi).
\end{equation}
Second, we denote by $\ket{\Psi_{i,f}(t)}$  the solutions of Equation (\ref{Ch2_FullSchEq}) such that 
$\ket{\Psi_i(-\infty)} = \ket{i}$, $\ket{\Psi_f(+\infty)} = \ket{f}$; similarly, $\ket{\Phi_{i,f}(t)}$ are the solutions of  the nonstationary Schr\"{o}dinger equation
$$
 i\partial \ket{\Phi(t)}/\partial t = \h_0(\varphi)\ket{\Phi(t)},
 $$
 with the initial conditions: $\ket{\Phi_i(-\infty)} = \ket{i}$ and $\ket{\Phi_f(+\infty)} = \ket{f}$, correspondingly.
 
Having defined all necessary functions, we introduce two equivalent forms of the transition amplitude $\mathfrak{M}_{i\to f}$ by employing the corresponding version of the $S$-matrix (see, e.g., References \cite{Becker_2005, Smirnova2007a, Bauer2008a}):
the reversed time form (sometimes called the ``prior'' form) 
\begin{equation}\label{Ch2_M_PriorForm}
\mathfrak{M}_{i\to f}^{(r)} = -i \int_{-\infty}^{\infty}  \bra{ \Psi_f (t)}\hat{V}(\omega t)\ket{\Phi_i (t)}dt 
\end{equation}
and the direct time form (the ``post'' form)
\begin{equation}\label{Ch2_M_PostForm}
\mathfrak{M}_{i\to f}^{(d)} = -i \int_{-\infty}^{\infty}  \bra{ \Phi_f (t)}\hat{V}(\omega t)\ket{\Psi_i (t)}dt.
\end{equation}
It is noteworthy to recall the physical interpretation of Equations (\ref{Ch2_M_PriorForm}) and (\ref{Ch2_M_PostForm}). The terms $\bra{ \Psi_f (t)}\hat{V}(\omega t)\ket{\Phi_i (t)}$ and $\bra{ \Phi_f (t)}\hat{V}(\omega t)\ket{\Psi_i (t)}$ can be regarded as the amplitudes of quantum ``jumps,'' which occur at the time moment $t$. The integrals over $t$ convey that these jumps take place at {\it any} time.

Introducing the adiabatic state $\ket{\phi_f(\varphi)}$ and term $E_f(\varphi)$ of the Hamiltonian $\h_0$,
\begin{equation}\label{Ch2_H0AdiabaticState}
\h_0(\varphi) \ket{\phi_f(\varphi)} = E_f (\varphi) \ket{\phi_f(\varphi)}, 
\end{equation}
the wave function $\ket{\Phi_f(t)}$ can be readily presented in the form of Equation (\ref{Ch2_SavichevEq_Corrected}). In further investigations, we employ the post form [Equation (\ref{Ch2_M_PostForm})], and thus we shall assume that 
\begin{equation}\label{Ch2_IniConditions}
\ket{\psi_i(-\infty)} \equiv \ket{i}, \quad \ket{\phi_f(+\infty)} \equiv \ket{f}.
\end{equation}
In the case of the prior form [Equation (\ref{Ch2_M_PriorForm})], condition (\ref{Ch2_IniConditions}) has to be substituted by 
$\ket{\phi_i(-\infty)} \equiv \ket{i}$ and $\quad \ket{\psi_f(+\infty)} \equiv \ket{f}$, where $\ket{\phi_i(\varphi)}$ and $\ket{\psi_f(\varphi)}$ are adiabatic states of the Hamiltonians $\h_0$ and $\h$, correspondingly. 

Substituting the asymptotic representations [Equation (\ref{Ch2_SavichevEq_Corrected})] of the wave functions $\ket{ \Phi_f (t)}$ and $\ket{\Psi_i (t)}$ into Equation (\ref{Ch2_M_PostForm}), we obtain 
\begin{eqnarray}\label{Ch2_M_Before_SaddlePointInt}
\mathfrak{M}_{i\to f}^{(d)} = \frac{-i}{(2\pi)^2\omega^3} 
\int f({\bf z}) e^{iS({\bf z})/\omega} d^5 {\bf z}\left[ 1+ o(1)\right],
\end{eqnarray}
where ${\bf z} = (\mathcal{E}, \eta, \mathcal{E}', \eta', \varphi)$ is a five-dimensional vector, $d^5 {\bf z} = d\mathcal{E} d\eta d\mathcal{E}' d\eta' d\varphi$, 
$f({\bf z}) = \bra{\phi_f\left(\varphi_f(\mathcal{E}')\right)}\hat{V}(\varphi)\ket{\psi_i\left(\varphi_i(\mathcal{E})\right)}$, and
$S({\bf z}) = \mathcal{E}\left(\eta-\varphi\right) + \mathcal{E}'\left(\varphi-\eta'\right) - \int^{\eta} E_i(\xi)d\xi + \int^{\eta'} E_f(\xi)d\xi$. Bearing in mind that $1/\omega$ is a large parameter, the five-dimensional integral in Equation (\ref{Ch2_M_Before_SaddlePointInt}) can be calculated by means of the saddle-point approximation. Finally,   {\it the post form of the transition amplitude} within the adiabatic approximation reads
 \begin{eqnarray}\label{Ch2_AmplitudeFinalExpr}
\mathfrak{M}_{i\to f}^{(d)} &=& \sqrt{\frac{2\pi}{\omega}}\sum_{\varphi_{0}} 
\frac{ \bra{\phi_f(\varphi_{0})} \hat{V} (\varphi_{0}) \ket{\psi_i(\varphi_{0})}}
{  \left. \sqrt{ \frac{d}{d\varphi} \left[ E_f(\varphi)-E_i(\varphi)\right] } \right|_{\varphi = \varphi_{0}} } \nonumber\\
&& \times\exp\left\{\frac i{\omega}\int^{\varphi_{0}} \left[ E_f(\varphi) - E_i(\varphi) \right]d\varphi \right\} [ 1 + o(1)],
 \end{eqnarray}
 where $\sum_{\varphi_{0}}$ denotes the summation over simple saddle points $\varphi_{0}$, i.e., solutions of the equation 
 \begin{eqnarray}\label{Ch2_SaddlePointEq}
  E_f(\varphi_{0}) &=& E_i(\varphi_{0}), \\
  \frac{d}{d\varphi} E_f(\varphi_{0}) &\neq& \frac{d}{d\varphi}E_i(\varphi_{0}). \label{Ch2_SimpleSaddlePointEq}
 \end{eqnarray} 
 
The physical interpretation of the sum over $\varphi_{0}$ is as follows: quantum jumps occur only at {\it isolated} time moments  $t_{0}= \varphi_{0}/\omega$, when the jumps are most probable; hence, $t_{0}$ are called ``transition times.'' Note that the given interpretation deviates from the physical meaning of the time integral in Equation (\ref{Ch2_M_PostForm}).
 
 Some general remarks on Equation (\ref{Ch2_AmplitudeFinalExpr}) are to be made:
 
\begin{itemize}
\item[i.] $\varphi_{0}$ is usually a complex solution of Equation (\ref{Ch2_SaddlePointEq}); therefore, saddle points $\varphi_{0}$ with negative imaginary parts should be ignored because such points make exponentially large contributions to the amplitude, which leads to unphysical probabilities.
\item[ii.] If $\varphi_{0}^1, \, \varphi_{0}^2, \dots, \varphi_{0}^n$ are solutions of Equation (\ref{Ch2_SaddlePointEq}), such that 
$\Im\left(\varphi_{0}^1\right)>\Im\left(\varphi_{0}^2\right)> \ldots > \Im\left(\varphi_{0}^n\right)>0$, then all but the single term that corresponds to the saddle point $\varphi_{0}^n$ may be neglected in the sum over $\varphi_{0}$ in Equation (\ref{Ch2_AmplitudeFinalExpr}). One can do so since this saddle point has the largest contribution to the transition  amplitude.
\item[iii.] On the one hand, the explicit form of $E_f(\varphi)$ is solely determined by partitioning [Equation (\ref{Ch2_Partition})]; on the other hand, $E_i(\varphi)$ is unique for a given quantum system.
\item[iv.] The exponential factor of Equation (\ref{Ch2_AmplitudeFinalExpr}) is similar to the exponential factor in the Dykhne approach  \cite{Dykhne_1962, Chaplik_1964, Chaplik1965a, Davis_1976} (see also References \cite{Landau_1977, Delone_1985}) -- the methods for calculating the amplitude of bound-bound transitions within the adiabatic approximation. Hence, Equation (\ref{Ch2_AmplitudeFinalExpr}) may be considered as a generalization of the Dykhne formula for bound-free transitions.
\item[v.] By employing an appropriate version of the saddle-point method, one can in principle generalize Equation (\ref{Ch2_AmplitudeFinalExpr}) for the case when the condition (\ref{Ch2_SimpleSaddlePointEq}) is violated.
\end{itemize}

\section{A Connection between  the Quasiclassical and Adiabatic Approximations}

Now, the connection between the amplitude [Equation (\ref{Ch2_AmplitudeFinalExpr})] and the method of complex trajectories is to be manifested. According to the method of complex quantum trajectories (see, e.g., References \cite{Landau1932, Landau1932a, Landau_1977, Nikitin1993}, and the imaginary time method \cite{Popov2005}), to calculate the probability of the transition from the initial state to the final, one should first solve the corresponding  classical equations of motion and find the ``path'' of such a transition. However, this path is complex; in particular, the transition point $\R_{0}$ and transition time $t_{0}$ at which the transition occurs are complex. Parameters $\R_{0}$ and $t_{0}$ are determined by the classical conservation laws as shown below in this section. Next, one has to obtain the classical action $S_f(\R_f, t_f; \R_{0}, t_{0}) + S_i(\R_{0}, t_{0}; \R_i, t_i)$ for the motion of the system in the initial state from the initial position $\R_i$ at time $t_i$ to the transition point $\R_{0}$ at time $t_{0}$ and then in the final state from $\R_{0}$ at $t_{0}$ to the final position $\R_f$ at time $t_f$. Finally, the probability of the transition is given by 
\begin{equation}\label{Ch2_ImTime}
\Gamma \propto \exp\left\{ -2\Im\left[S_f(\R_f, t_f; \R_{0}, t_{0}) + S_i(\R_{0}, t_{0}; \R_i, t_i)\right]\right\}.
\end{equation}
Equations (\ref{Ch2_AmplitudeFinalExpr}) and (\ref{Ch2_ImTime}) must coincide in some region of parameters. The method of complex trajectories can be derived as the quasiclassical  approximation of the transition amplitude [Equation (\ref{Ch2_M_PriorForm}) or Equation (\ref{Ch2_M_PostForm})]; we outline this derivation below. Therefore, it would be of methodological interest to establish an explicit connection between Equations (\ref{Ch2_AmplitudeFinalExpr}) and (\ref{Ch2_ImTime}). 

Without loss of generality, assuming that $\hat{V}(\omega t)$ is a non-differential operator, we obtain the quasiclassical approximation to Equation (\ref{Ch2_M_PostForm}) 
\begin{eqnarray}\label{Ch2_M_WKB}
 \mathfrak{M}_{i\to f}^{(d)} &\approx& -i\int dt \int d^3 \R d^3 \R_f d^3 \R_i \,\bra{f} \R_f \rangle \hat{V}(\omega t, \R) \bra{\R_i} i \rangle \nonumber\\
&&\times F_f^* F_i \exp\left\{ i\left[S_f(\R_f, t_f; \R, t) + S_i(\R, t; \R_i, t_i) \right]\right\},
\end{eqnarray}
where $t_{f,i} = \pm\infty$, $F_f\exp(iS_f)$, and $F_i\exp(iS_i)$ are the quasiclassical versions of the propagators with the Hamiltonian $\h_0$ and $\h$, correspondingly.  We recall that the general form of the quasiclassical propagator is given by
\begin{eqnarray}\label{Ch2_GenaralQCPropagator}
\sum_{\alpha} F^{(\alpha)}(\R, t; \R', t')\exp\left[iS^{(\alpha)}(\R, t; \R', t')\right],
\end{eqnarray}
where the sum denotes the summation over classical paths that connect the initial $(\R', t')$ and final $(\R,t)$ points. Therefore, usage of this form of the quasiclassical propagator, $F\exp(iS)$, is justified if we assume that there is only one such path; indeed, this is the case in the majority of practical calculations, and thus we shall accept this assumption hereinafter.

In order to reach Equation (\ref{Ch2_ImTime}) from  Equation (\ref{Ch2_M_WKB}), one has to calculate the integrals over $\R$ and $t$ in Equation (\ref{Ch2_M_WKB}) by means of the saddle-point approximation. The equations for the saddle points $\R_{0}$ and $t_{0}$, i.e., the transition points, read
\begin{eqnarray}
\left. \partial \left[ S_i(\R, t; \R_i, t_i) -S_f(\R, t; \R_f, t_f)  \right]/\partial t\right|_{t=t_{0},\, \R=\R_{0}} = 0, \label{Ch2_Saddle_PointT}\\
\left. \nabla_{\R}\left[ S_i(\R, t; \R_i, t_i) -S_f(\R, t; \R_f, t_f)  \right]\right|_{t=t_{0},\, \R=\R_{0}} = {\bf 0}. \label{Ch2_Saddle_PointR}
\end{eqnarray}
Recalling the Hamilton-Jacobi equation
\begin{equation}\label{Ch2_HJEq}
\partial S_{i,f} (\R, t; \R_{i,f}, t_{i,f})/\partial t = 
-\Hcl_{i,f}(\R,\Pcl_{i,f}, t), 
\end{equation}
where $\Hcl_{i,f}$ are classical Hamiltonians and $\Pcl_{i,f}$ are classical canonical momenta
$$
\Pcl_{i,f}(\R, t)  = \nabla_{\R} S_{i,f}(\R, t; \R_{i,f}, t_{i,f}),
$$
we rewrite Equations (\ref{Ch2_Saddle_PointT}) and (\ref{Ch2_Saddle_PointR}) as the law of conservation of canonical momentum and the law of conservation of energy:
\begin{eqnarray}
 \Pcl_f(\R_{0}, t_{0}) &=& \Pcl_i(\R_{0}, t_{0}), \\
 \Hcl_{f}(\R_{0},\Pcl_{f}(\R_{0}, t_{0}), t_{0}) &=& \Hcl_{i}(\R_{0},\Pcl_{i}(\R_{0}, t_{0}), t_{0}).
\end{eqnarray}

Having introduced all the necessary quantities, we demonstrate the correspondence between Equations (\ref{Ch2_AmplitudeFinalExpr}) and (\ref{Ch2_ImTime}) within an exponential accuracy. Performing a simple transformation and using Equation (\ref{Ch2_HJEq}), we reach
\begin{eqnarray}
&& S_f(\R_f, t_f; \R_{0}, t_{0}) + S_i(\R_{0}, t_{0}; \R_i, t_i) \nonumber\\
&=& S_f(\R_f, t_f; \R_{0}, t_i)+S_i(\R_{0}, t_i; \R_i, t_i)  + \int_{t_i}^{t_{0}} \left[ \frac{\partial}{\partial\tau} S_i(\R_{0}, \tau; \R_i, t_i) -
\frac{\partial}{\partial\tau} S_f(\R_{0}, \tau; \R_f, t_f) \right]d\tau  \nonumber\\
&=& S_f(\R_f, t_f; \R_{0}, t_i)+S_i(\R_{0}, t_i; \R_i, t_i)  \nonumber\\
&&+ \int_{t_i}^{t_{0}} \left[ \Hcl_{f}(\R_{0},\Pcl_{f}(\R_{0}, \tau), \tau) - \Hcl_{i}(\R_{0},\Pcl_{i}(\R_{0}, \tau), \tau)\right] d\tau. \label{Ch2_SplusS_as_int_ofEnerg}
\end{eqnarray}
Usually in the case of multiphoton ionization, $t_{0}$ is complex and $\R_{0}$ is real; moreover, energies along the trajectories are always real -- even under the barrier. This suggests that the first two terms of the right hand side of Equation (\ref{Ch2_SplusS_as_int_ofEnerg})  affect only the phase and do not contribute to the probability.

We point out that sometimes it is useful to employ a mixed representation, such as
\begin{eqnarray}
&& S_f(\R_f, t_f; \R_{0}, t_{0}) + S_i(\R_{0}, t_{0}; \R_i, t_i) \nonumber\\
&=& S_f(\R_f, t_f; \R_{0}, t_i)  + \int_{t_i}^{t_{0}} \left[  \mathrsfs{L}_i (\tau, \R, \dot{\R})-
\frac{\partial}{\partial\tau} S_f(\R_{0}, \tau; \R_f, t_f) \right]d\tau  \nonumber\\
&=& S_f(\R_f, t_f; \R_{0}, t_i)  + \int_{t_i}^{t_{0}} \left[ \mathrsfs{L}_i (\tau, \R, \dot{\R}) + \Hcl_{f}(\R_{0},\Pcl_{f}(\R_{0}, \tau), \tau)\right]d\tau,
\end{eqnarray}
where $\mathrsfs{L}_i$ is the classical Lagrangian, $S_i = \int \mathrsfs{L}_i d\tau$.

Finally, since $\Hcl_f$ and $\Hcl_i$ are  the quasiclassical limits of $E_f$ and $E_i$ (this will be demonstrated below), we conclude that the exponential factors of Equations (\ref{Ch2_AmplitudeFinalExpr}) and (\ref{Ch2_ImTime}) indeed coincide within the quasiclassical approximation.

The wave function 
\begin{eqnarray}\label{Ch2_PsiQC}
\Psi_{qc}(\R, t) = \int  F_i(\R, t; \R', t_i) e^{\frac{i}{\hbar}S_i(\R, t; \R', t_i)} \phi_{in}(\R') d^3\R',
\end{eqnarray}
is the (leading-order term) quasiclassical solution of Equation (\ref{Ch2_FullSchEq}) with the initial condition $\Psi_{qc}(\R, t_i) = \phi_{in}(\R)$. Employing Equation (\ref{Ch2_HJEq}) and  bearing in mind that $\Hcl_{i}=\Hcl_{i}(\R,\Pcl_{i}, t)$ does not depend on $\R'$, we obtain
\begin{eqnarray}\label{Ch2_WKBAdiabaticTermState}
\h \Psi_{qc}  = i\hbar \partial_t \Psi_{qc}\left[ 1 + O(\hbar)\right] = \Hcl_{i} \Psi_{qc} \left[ 1 + O(\hbar)\right].
\end{eqnarray} 
Since we have freedom of choosing the initial condition $\phi_{in}(\R)$, there are in general infinitely many wave functions [Equation (\ref{Ch2_PsiQC})] that satisfy Equation (\ref{Ch2_WKBAdiabaticTermState}). Comparing Equations (\ref{Ch2_WKBAdiabaticTermState}) and (\ref{Ch2_AdiabaticTermStateFullH}), and  taking into account the latter, we formulate the following property of the adiabatic term and state of a given quantum system {\it within the quasiclassical limit: there exists only one adiabatic term, which is equal to the classical Hamiltonian, and any solution of the corresponding nonstationary Schr\"{o}dinger equation is also an adiabatic state that corresponds to this adiabatic term} (i.e., the adiabatic term is infinitely degenerate). Note that this property is completely ruled out once the general form of the quasiclassical propagator (\ref{Ch2_GenaralQCPropagator}) is considered.

The property stated above, nevertheless, merely accentuates the fundamental difference between the quasiclassical and adiabatic approximations.   As mentioned in Chapter \ref{chapter1}, the adiabatic approximation allows us to obtain the solution of the nonstationary Schr\"{o}dinger equation as an  asymptotic series in terms of the small parameter $\omega/\omega_{at}$; however, the quasiclassical approximation is a method of obtaining an asymptotic expansion of the solution with respect to the small parameter $\hbar$. These two series may be dissimilar in a general case. 

\section{Summary}

Since the adiabatic approximation shall be used in subsequent chapters. It is  convenient to conclude this section with a short summary of the main equations. 

Let the Hamiltonian of a system $\h(\omega t)$ be a slowly varying function of time $t$, i.e., $\omega\ll 1$, then the rate of a non-adiabatic transition is given by
\begin{eqnarray}\label{Ch2_PropNonAdiabat_Summary}
\Gamma \propto \exp\left( -\frac2{\omega} \Im\int^{\varphi_0} [E_f (\varphi)- E_i(\varphi)] d\varphi \right),
\end{eqnarray}
where $\varphi=\omega t$, $\varphi_0$ is the complex solution of the equation 
\begin{eqnarray}\label{Ch2_TransitionPoint_Summary}
E_i(\varphi_0) = E_f(\varphi_0)
\end{eqnarray} 
with the smallest positive imaginary part, and $E_{i,f}(\varphi)$ being the adiabatic terms of the system ``before'' and ``after'' the non-adiabatic transition. If the quasiclassical approximation to $\Gamma$ is sufficient, then instead of using the adiabatic terms, one can employ the total energies of the corresponding classical system.

\glsresetall

\chapter{Single-electron Ionization}\label{chapter3}

\section{Main Results}

The Keldysh theory \cite{Keldysh_1965} was reformulated in terms of the Dykhne adiabatic approximation in Reference \cite{Delone_1985}, for the first time. Later, this approach was employed in References \cite{Delone_1991, Yudin_2001B, Krainov_2003, Rastunkov_2007, Bondar2008}.

Similarly, we shall apply the adiabatic approach [Equation (\ref{Ch2_PropNonAdiabat_Summary})] to the problem of ionization of a single electron under the influence of a linearly polarized laser field with the frequency $\omega$ and the strength $\mathbf{F}$. The initial and final classical energies for such a process are given by 
\begin{equation}\label{Ch3_Ei_Ef}
E_i (\varphi) = -I_p, \qquad E_f (\varphi) = \frac 12 [\K + \A(\varphi)]^2,
\end{equation}
where $I_p$ is the ionization potential, $\K$ is the canonical momentum (measured on the detector), and $\A(t) = - (\mathbf{F}/\omega) \sin\varphi$.

According to Equation (\ref{Ch2_PropNonAdiabat_Summary}), the probability of one-electron ionization $\Gamma$ can be written as 
\begin{eqnarray}
\Gamma &\propto& \exp\left[ -2\Im S \left( \kk, \k, I_p\right)\right], \label{Ch3_SPAEq1}\\
S \left( \kk, \k, I_p\right) &=& \int^{t_0}\left( \frac 12 [\kk + A(t)]^2 + \frac 12 \k^2 + I_p \right) dt, \label{Ch3_SingleElectronIonizAction}
\end{eqnarray}
where $S$ is the action. Equation (\ref{Ch2_TransitionPoint_Summary}) can be rewritten in terms of $S$ as
\begin{equation}\label{Ch3_SPAEq2}
\frac{\partial}{\partial t_0} S \left( \kk, \k, I_p\right) = 0.
\end{equation}
Note that the analogy between the saddle point S-matrix calculations \cite{Ivanov_2005}, where transitions are calculated using stationary points of the action, and the adiabatic approach can be seen from Equations (\ref{Ch3_SPAEq1}) and (\ref{Ch3_SPAEq2}). The transition point is given by
\begin{equation}\label{Ch3_ExSolutT0}
\omega t_0 =   \mathrm{Arcsin}\left[ \gamma\left( \frac{ \kk}{\sqrt{2I_p}} + i\sqrt{1 +\frac{\k^2}{2I_p}} \right)\right],
\end{equation}
where $\gamma$ is the Keldysh parameter
$$
\gamma = \frac{\omega}F\sqrt{2I_p} = \sqrt{\frac{I_p}{2U_p}}
$$
and $U_p = \left( F/2\omega \right)^2$ is the ponderomotive potential. In order to extract the imaginary and real parts of this solution, the following equation \cite{Abramowitz_1972} can be used 
\begin{equation}\label{Ch3_AbramArcsin}
 \mathrm{Arcsin}\, (x+iy) = 2K\pi + \arcsin\beta + i\ln\left[\alpha + \sqrt{\alpha^2 -1}\right],
\end{equation}
where $K$ is an integer and
\begin{eqnarray}
\left\{ \alpha\atop \beta \right\} = \frac 12 \sqrt{(x+1)^2 + y^2} \pm \frac 12\sqrt{(x-1)^2 + y^2}.    \nonumber
\end{eqnarray}
Using Equation (\ref{Ch3_AbramArcsin}) in Equation (\ref{Ch3_SPAEq1}), we obtain
\begin{equation}\label{Ch3_GeneralRate}
\Gamma (\gamma, \kk, \k)  \propto \exp\left[- \frac{2I_p}{\omega}f(\gamma, \kk, \k) \right],
\end{equation}
where
\begin{eqnarray}
f(\gamma, \kk, \k)  &=& \left(1+\frac 1{2\gamma^2}+ \frac{k^2}{2I_p} \right)\arccosh\alpha - \sqrt{\alpha^2 -1}\left( \frac{\beta}{\gamma}\sqrt{\frac 2{I_p}}\kk  +
\frac{\alpha\left[1-2\beta^2\right]}{2\gamma^2} \right),  \nonumber\\
\left\{ \alpha\atop \beta \right\} &=& \frac{\gamma}2 \left( 
\sqrt{\frac{k^2}{2I_p} + \frac 2{\gamma}\frac{\kk}{\sqrt{2I_p}} 
+ \frac 1{\gamma^2} +1 }   \pm
\sqrt{\frac{k^2}{2I_p} - \frac 2{\gamma}\frac{\kk}{\sqrt{2I_p}} 
+ \frac 1{\gamma^2} +1 }
\right), \nonumber\\
k^2 &=& \kk^2 + \k^2. \nonumber
\end{eqnarray}
Note that $\alpha > 1$. It must be stressed here that {\it no assumptions on the momentum of the electron have been made.} However, Equation (\ref{Ch3_GeneralRate}) has an exponential accuracy because the influence of the Coulomb field of a nucleus cannot be accounted for by the \gls{SFA}. The correct exponential prefactor has been obtained within the \gls{PPT} approach.

Similarly to the Yudin-Ivanov formula \cite{Yudin_2001B}, Equation  (\ref{Ch3_GeneralRate}) is valid if the strength of the laser field $F$ depends on time, $F\to E_0 g(t)$, where the envelope  $g(t)$ of the pulse is assumed to be nearly constant during one-half of a laser cycle.

\section{Connections with Previous Results}

In this section, Equation (\ref{Ch3_GeneralRate}) is applied to  some special cases in order to establish connections with previously known results. 

In the case of zero final momentum ($k=0$), we have that $\alpha = \sqrt{1+ \gamma^2}$ and $\beta = 0$. 
In this limit we recover the original Keldysh formula \cite{Keldysh_1965}
$$
f(\gamma,0,0) = \left( 1 + \frac 1{2\gamma^2} \right) \arcsinh \gamma -
\frac {\sqrt{1+\gamma^2}}{2\gamma}.
$$

In the tunneling limit ($\gamma\ll 1$) the following formulas can be obtained. Expanding the function $f(\gamma, \kk, \k) $ in a Taylor series up to third order with respect to $\gamma$ and setting $\k=0$, we obtain
\begin{equation}\label{Ch3_CorkumTunnelEq}
\Gamma (\gamma, \kk, 0) \approx  \Gamma (\gamma, 0,0)
\exp\left[ -\frac{ \kk^2}{3\omega}\gamma^3 \right]. 
\end{equation}
Equation (\ref{Ch3_CorkumTunnelEq}) has been derived by a classical approach in Reference \cite{Corkum_1989} (see also Reference \cite{Delone_1991}). Discussions regarding the physical origin of Equation (\ref{Ch3_CorkumTunnelEq}) are presented in Reference \cite{Ivanov_2005}. Performing the same expansion and setting $\kk=0$, we obtain
\begin{equation}\label{Ch3_DeloneTunnelEq}
\Gamma (\gamma, 0,\k) \propto \exp\left[ -\frac{2\left(\k^2 + 2I_p\right)^{3/2}}{3F}\right].
\end{equation} 
This equation has been derived in Reference \cite{Delone_1991}. 
For small values of $\k$, Equation (\ref{Ch3_DeloneTunnelEq}) can be approximated by
\begin{equation}\label{Ch3_IvanovTunnelEq}
\Gamma (\gamma, 0,\k) \approx  \Gamma (\gamma, 0,0)
\exp\left[ -\frac{\sqrt{2I_p}\k^2}F \right].
\end{equation}

Let us fix $\k=0$ and continue working in the tunneling regime. For the case of high kinetic energy, such as $\kk^2/2 > 2U_p$ and $\sqrt{\kk^2/(4U_p)}-1 \gg \gamma$, we obtain 
$$
\alpha \approx \sqrt{ \frac{\kk^2}{4U_p}}, \qquad \beta \approx 1,
$$
and the ionization rate $\Gamma$ is given by 
\begin{equation}\label{Ch3_KrainovEq7}
\Gamma \propto \exp\left\{ -\frac{2U_p}{\omega}\left[ \left( \frac{\kk^2}{2U_p} + 1\right) \arccosh\sqrt{\frac{\kk^2}{4U_p}} - 
3\sqrt{\frac{\kk^2}{4U_p}\left(\frac{\kk^2}{4U_p} -1 \right)}
\right] \right\}.
\end{equation}
Equation (\ref{Ch3_KrainovEq7}) has been obtained in Reference \cite{Krainov_2003}.

Calculating the asymptotic expansion of the function $f(\gamma, \kk, 0) $ for $\kk^2/2 \gg 2U_p$, we obtain
\begin{equation}\label{Ch3_KrainovTunnelEq}
\Gamma \propto \left( \frac{F^2 \exp(3)}{4\omega^2 \kk^2} \right)^{\frac{\kk^2}{2\omega}} \qquad (\kk^2/2 \gg 2U_p, \,\k=0).
\end{equation}
Equation (\ref{Ch3_KrainovTunnelEq}) has been obtained for tunneling ionization in Reference \cite{Krainov_2003}. Here, we have proved that Equation (\ref{Ch3_KrainovTunnelEq}) is valid for arbitrary  $\gamma$. A similar formula can be derived for $\k\neq 0$
\begin{equation}\label{Ch3_PPerpAsympt}
\Gamma \propto \left( \frac{F^2 \exp(1)}{4\omega^2\k^2} \right)^{\frac{\k^2}{2\omega}} \qquad (\k^2/2 \gg 2U_p, \, \kk=0),
\end{equation}
which is also valid for arbitrary values of the Keldysh parameter $\gamma$.

Consider the asymptotic expansion of Equation (\ref{Ch3_GeneralRate}) for large values of $I_p$ ($I_p \gg k^2/2$). In this case $\alpha$ and $\beta$ can be approximated by
$$
\alpha \approx \sqrt{1+\gamma^2}, \qquad 
\beta \approx \frac {\gamma}{\sqrt{1+\gamma^2}} \frac{\kk}{\sqrt{2I_p}}.
$$
Using these equations, we obtain
\begin{equation}\label{Ch3_PPTEq}
f(\gamma, \kk, \k) \approx f(\gamma, 0,0) + 
\frac{k^2}{2I_p}\arcsinh\gamma - 
\frac{\gamma}{\sqrt{1+\gamma^2}}\frac{\kk^2}{2I_p}.
\end{equation}
Equation (\ref{Ch3_PPTEq}) has been reached within the \gls{PPT} approach (see also Reference \cite{Popov2004}). 

Goreslavskii {\it et al.} \cite{Goreslavskii_2005} have derived an expression for the spectral-angular distribution of single-electron ionization without any assumptions on the momentum of the electron. However, they have summed over saddlepoints, i.e., the contribution from previous laser cycles has been taken into account. On the contrary, we have not performed any summation because we are interested in the most recent contribution to ionization. Therefore, our result does account for the phase dependence of the ionization rate, unlike that of
Reference \cite{Goreslavskii_2005}.

To make the phase dependance explicit in Equation (\ref{Ch3_GeneralRate}), we apply the  substitution  
\begin{equation}\label{Ch3_PhaseDep}
\kk \to \kk - A(t).
\end{equation}
The analytical expression for the ionization rate as a function of a laser phase when $k = 0$ has been achieved by Yudin and Ivanov \cite{Yudin_2001B}. Thus, Equation~(\ref{Ch3_GeneralRate}) is seen to be a generalization of the Yudin-Ivanov formula.

Note that generally speaking, there is no unique and consistent way of defining the instantaneous ionization rates within quantum mechanics, and such a definition is a topic of an ongoing discussion (see, e.g., References \cite{Saenz_2007, Smirnova2008a} and references therein). However,
the instantaneous ionization rates are indeed rigorously defined within the quasiclassical approximation (e.g., the Yudin-Ivanov formula), and we have employed this approach. Alternatively, one can approximate the instantaneous ionization rates by the static ionization rates at each point in time using the instantaneous value of the laser field.

Lastly, we illustrate Equation (\ref{Ch3_GeneralRate}) for the case of a hydrogen atom. The single-electron ionization spectra in the multiphoton regime ($\gamma \gg 1$) and in the tunneling regime ($\gamma \ll 1$) are plotted in Figure \ref{Ch3_Fig1}(a) and Figure \ref{Ch3_Fig1}(b) respectively. One concludes that the smaller $\gamma$, the more elongated the single-electron spectrum. We can notice that the maxima of both spectra are at the origin. Nevertheless, a dip at the origin has been observed experimentally \cite{Rudenko2004a, Alnaser2006, Smeenk2009a} in the parallel-momentum
distribution for the nobel gases within the tunneling regime, and afterwards it has been investigated theoretically in paper \cite{Guo2008} and references therein. However, such a phenomenon is beyond Equation (\ref{Ch3_GeneralRate}). Recently, Formula (\ref{Ch3_GeneralRate}) was employed to interpreted experimental measurements of perpendicular momentum distributions of photoelectrons \cite{Arissian2010} -- substantial, but not total, agreement was observed (see Figure \ref{Ch3_Fig4}).

The phase dependence of ionization for different initial momenta, recovered by means of Equation (\ref{Ch3_PhaseDep}), is illustrated in Figure \ref{Ch3_Fig2} for selected positive momenta. The curves for negative momenta are mirror reflections (through the axis $\phi = 0$) of the corresponding positive curves. Figure \ref{Ch3_Fig3} shows that the cutoff of the single-electron spectrum in the tunneling regime (the dashed line) corresponds exactly to the kinetic energy $2U_p$, which is the maximum kinetic energy of a classical electron oscillating under the influence of a linearly polarized laser field.

\begin{figure}
\begin{center}
\includegraphics[scale=0.75]{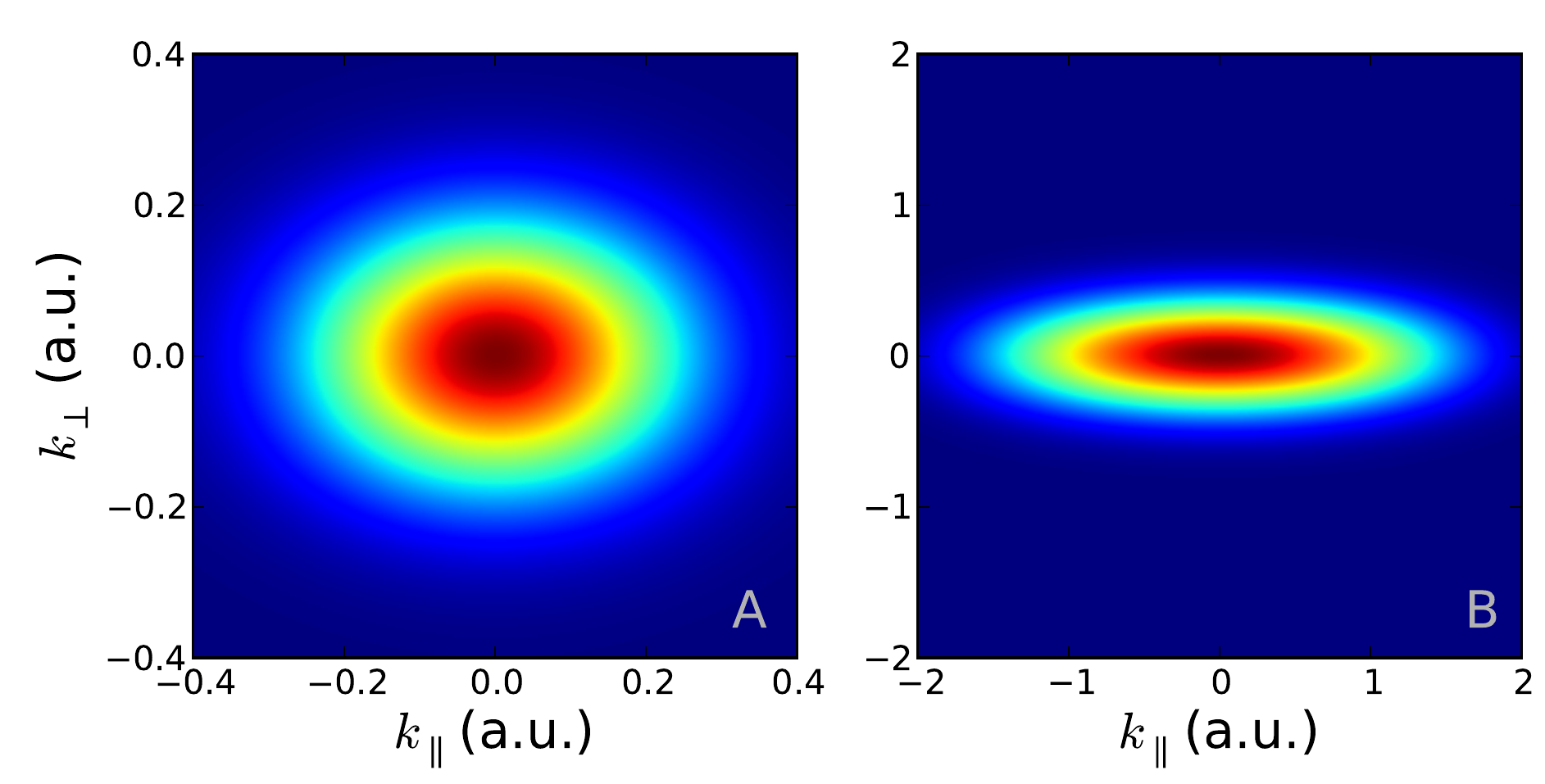}
\end{center}
\caption[The photoelectron spectrum of a hydrogen atom]{The photoelectron spectrum [Equation (\ref{Ch3_GeneralRate})] of a hydrogen atom at $800$ nm; (a) at $1\times 10^{13}$ W/cm$^2$ ($\gamma = 3.4$); (b) at $6\times 10^{14}$ W/cm$^2$ ($\gamma=0.4$).}\label{Ch3_Fig1}

\begin{center}
\includegraphics[scale=0.35]{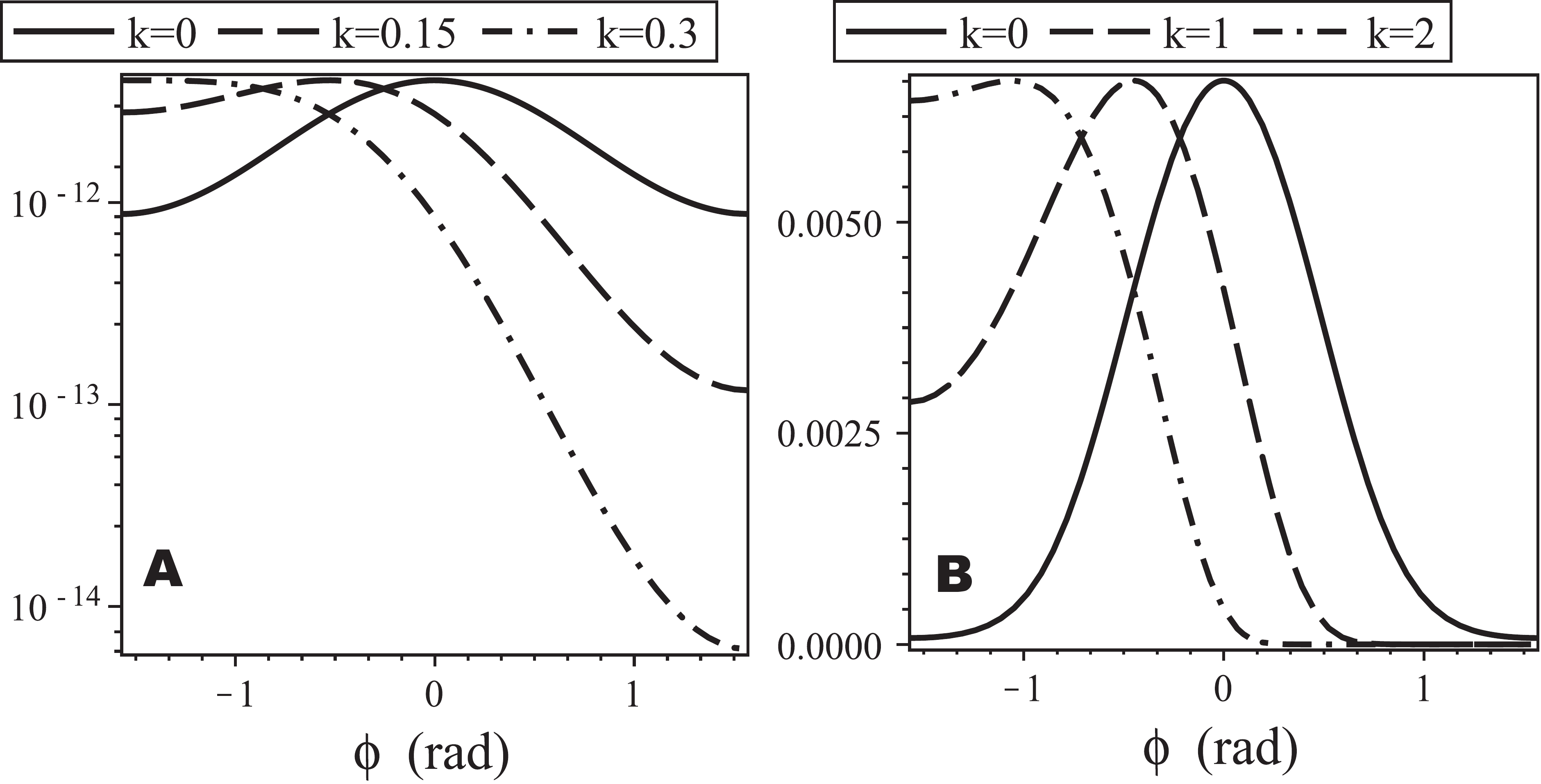}
\end{center}
\caption[The dependence of the photoelectron spectrum of a hydrogen atom on the instantaneous phase of a laser field]{(Reprinted from Reference \cite{Bondar2008}. Copyright (2008) by the American Physical Society.) The dependence of the photoelectron spectrum of a hydrogen atom on the instantaneous phase of a laser field: the plot of $\Gamma(\gamma, k+\frac F{\omega}\sin\phi, 0)$ [Equation (\ref{Ch3_GeneralRate})] for a hydrogen atom at $800$ nm; (a) at $1\times 10^{13}$ W/cm$^2$ ($\gamma = 3.4$); (b) at $6\times 10^{14}$ W/cm$^2$ ($\gamma=0.4$).}\label{Ch3_Fig2}
\end{figure} 

\begin{figure}
\begin{center}
\includegraphics[scale=0.37]{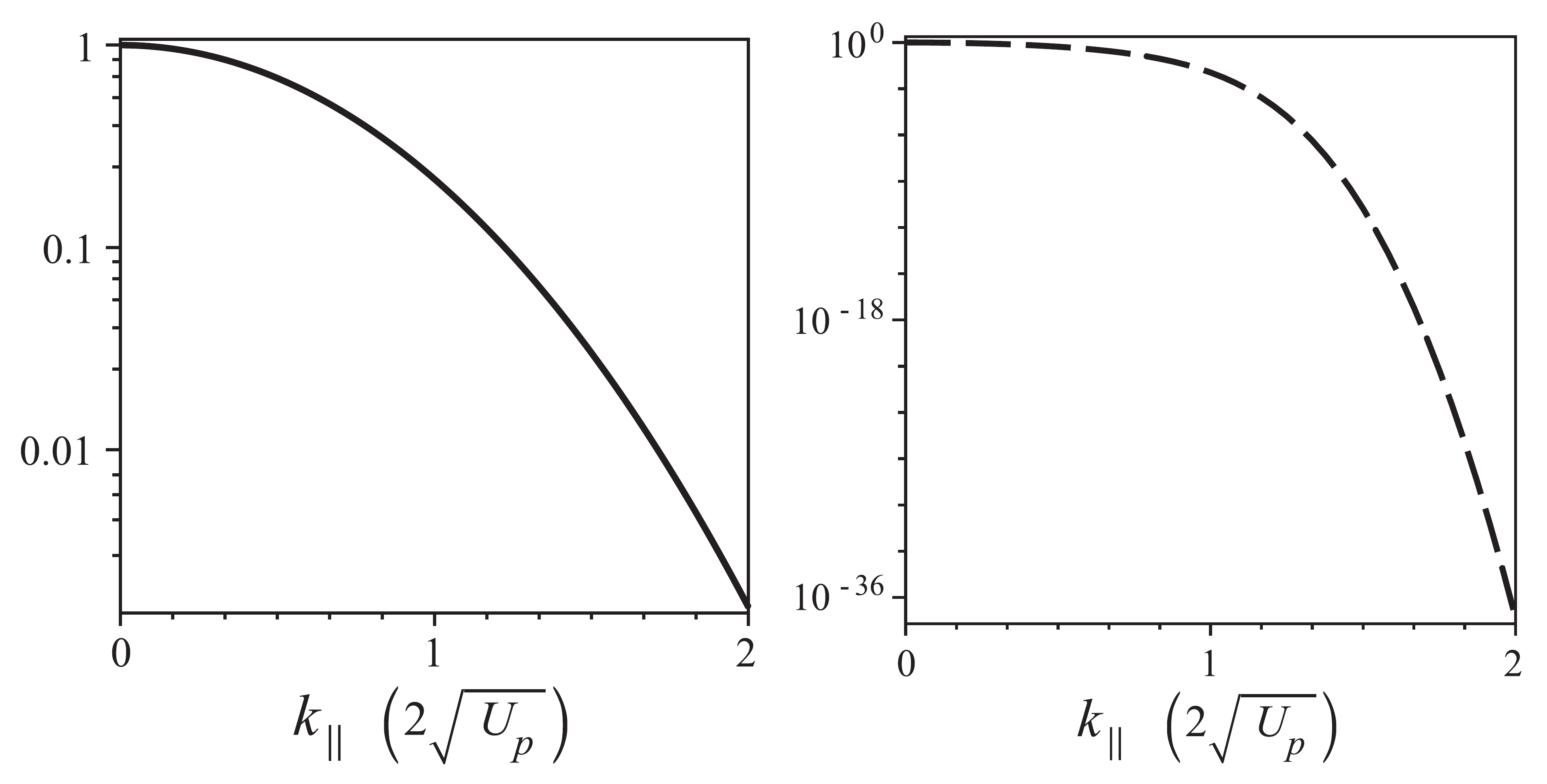}
\end{center}
\caption[The longitudinal momentum distribution of photoelectrons of a hydrogen atom]{(Reprinted from Reference \cite{Bondar2008}. Copyright (2008) by the American Physical Society.) The longitudinal momentum distribution of photoelectrons of a hydrogen atom: the plot of $\Gamma(\gamma, \kk, 0)/\Gamma(\gamma,0,0)$ [Equation (\ref{Ch3_GeneralRate})] for a hydrogen atom at $800$ nm; the solid line: $1\times 10^{13}$ W/cm$^2$ ($\gamma = 3.4$); the dashed line: $6\times 10^{14}$ W/cm$^2$ ($\gamma=0.4$).}\label{Ch3_Fig3}

\begin{center}
\includegraphics[scale=0.7]{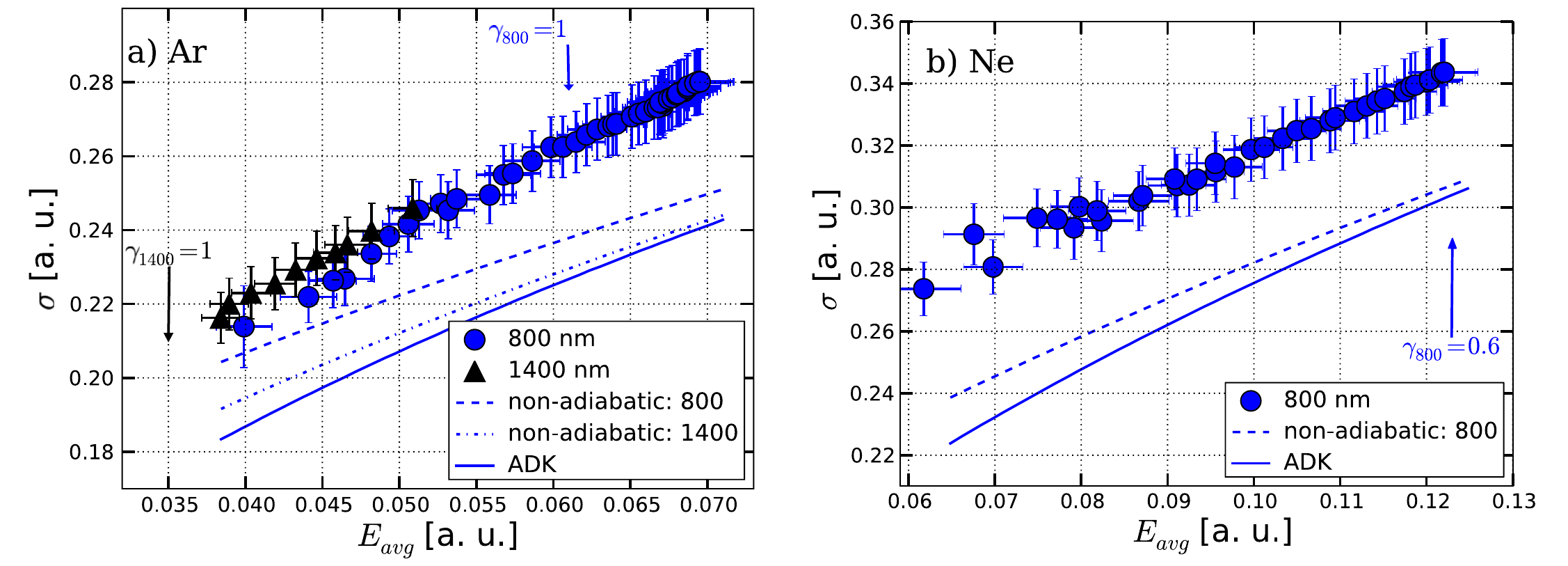}
\end{center}
\caption[The comparison of the theoretical momentum distribution with experimental data for Ar and Ne]{(Reprinted from Reference \cite{Arissian2010} with kind permission of C. Smeenk. Copyright (2010) by the American Physical Society.) (a) Width of the experimentally measured momentum distribution in Reference \cite{Arissian2010} as a function of the strength of the laser field for Ar at 800 and 1400 nm. The experimental data are compared with predictions of Equations (\ref{Ch3_GeneralRate}) and (\ref{Ch3_IvanovTunnelEq}). Results obtained by means of these equations are marked as  {\it ``non-adiabatic''} and {\it ``ADK''}, correspondingly. (b) Width of the experimentally measured momentum distribution as a function of the field for Ne.}\label{Ch3_Fig4}
\end{figure}
\glsresetall

\chapter{Nonsequential Double Ionization below Intensity Threshold: Contributions of the Electron-Electron and Electron-Ion Interactions}\label{chapter4}

\section{Formulation of the Problem}

Our model of \gls{NSDI} complements earlier theoretical
work  on calculating correlated two-electron spectra using
the strong-field S-matrix approach \cite{Becker_2005}. The key
theoretical advance of this work is the ability to include
nonperturbatively all relevant interactions for both active
electrons: with each other, with the ion, and with the laser
field. Electron-electron and electron-ion interactions are
included on an equal footing. Our model ignores multiple
recollisions and multiple excitations developing over several
laser cycles, such as those seen in the classical simulations
\cite{Ho_2005}. This simplification is particularly adequate for
the few-cycle laser pulses, as demonstrated in the experiment
\cite{Bhardwaj_2001}. According to this experiment, multiple
recollisions are noticeably suppressed already for 12 fs pulses
at $\lambda=800$ nm. For 6--7 fsec pulses at $\lambda=800$ nm, such
simplification is justified.

The process of \gls{NSDI} is shown by the Feynman diagram in Figure \ref{Ch4_Fig4}. The
system begins in the ground state $\ket{gg}$ at time $t_i$. At an
instant $t_b$, intense laser field promotes the first electron to
the continuum state $\ket{\K}$; the second electron remains in the
ground state of the ion $\ket{g^+}$. Recollison at $t_r$ frees
both electrons. The symmetric diagram where electrons 1 and 2
change their roles is not shown, but is included in the calculated
spectrum.

\begin{figure}
\begin{center}
\includegraphics[scale=0.3]{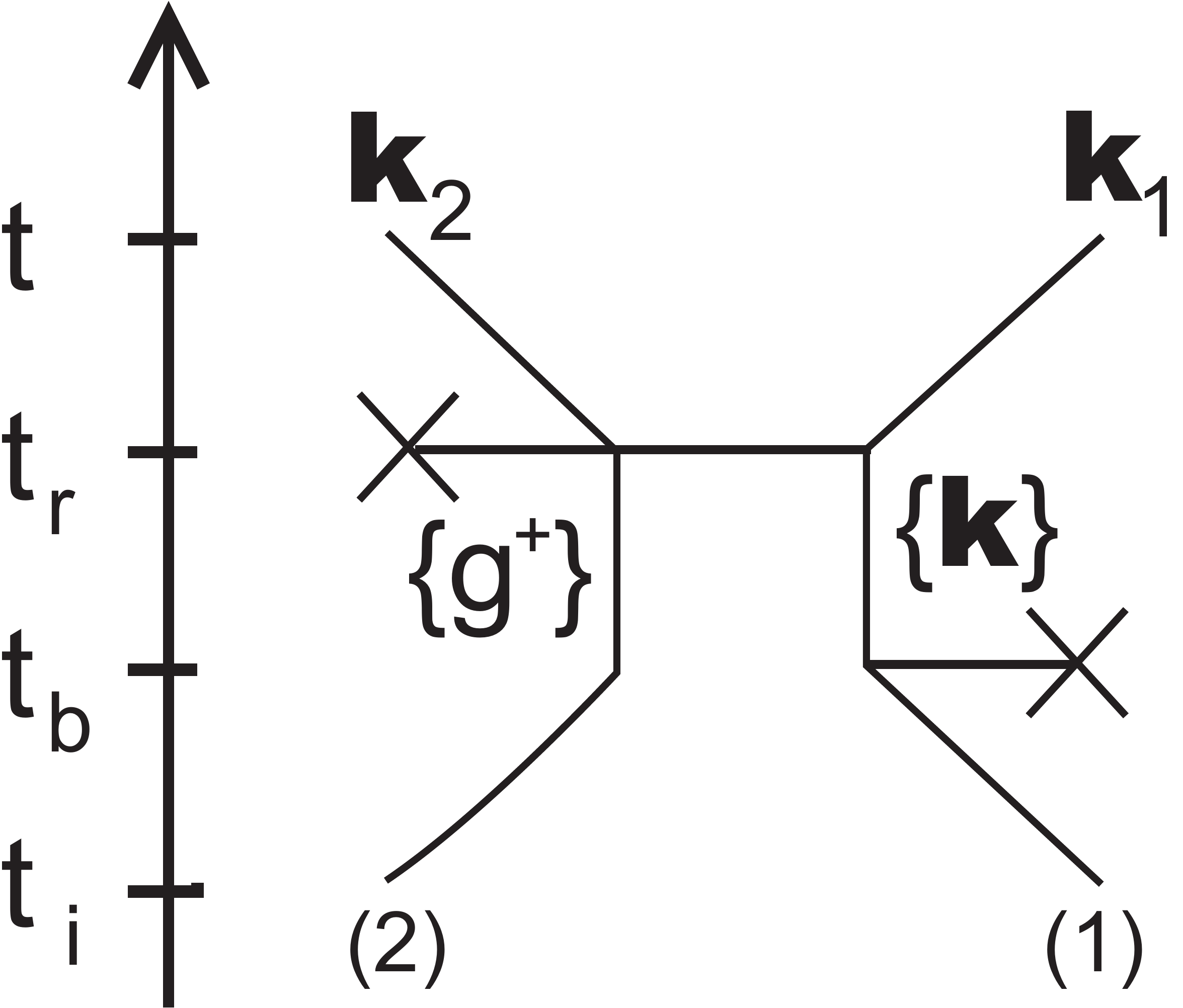}
\end{center}
\caption[The diagram of \glsentryname{NSDI} within the \glsentryname{SEE} mechanism.]{(Reprinted from Reference \cite{Bondar2009}. Copyright (2009) by the American Physical Society.) The diagram of \glsentryname{NSDI} within the \glsentryname{SEE} mechanism.}
\label{Ch4_Fig4}
\end{figure}

Let us apply the adiabatic approach [Equation (\ref{Ch2_PropNonAdiabat_Summary})] to the
two-electron process under consideration. The \gls{NSDI} within \gls{SEE} has two stages, namely
ionization of the first electron and the recollision. Hence,
strictly speaking, the adiabatic approximation has to be
applied to each of the two stages, since the total amplitude of
the process is the product of the ionization amplitude and the
recollision amplitude. However, it is the second (recollisioin)
amplitude that is responsible for the shape of the correlated spectra.
The first amplitude only gives the overall height of the spectra,
as it determines the overall probability of the recollision. Since
at this stage we are only interested in the shape of the
correlated spectra, we omit the ionization amplitude from this
discussion.

As a zero approximation, we define $E_{i}(\varphi)$ and $E_{f}(\varphi)$ for
the second part of \gls{NSDI} without the Coulomb interaction. Before
the recollision at the moment $\varphi_r = \omega t_r$ (Figure \ref{Ch4_Fig4}), one electron
is bound and another is free. The classical energy of the system
before $\varphi_r$ is
\begin{equation}\label{Ch4_EiCMI}
E_i (\varphi)= \frac 12 \left[ -\A(\varphi_b(\varphi))+\A(\varphi) \right]^2 -I_p^{(2)},
\end{equation}
where $I_p^{(2)}$ denotes the ionization potential of the ion and
$
 \A(t) = - ({\bf F}/\omega)\sin(\omega t)
$
 is the vector potential of a linearly polarized laser field.
The time of ``birth'' (ionization) for the first electron
$\varphi_b = \omega t_b(\varphi)$ is the standard function of the instant of recollision $\varphi_r$,
which is obtained from the saddle-point S-matrix calculations in Appendix \ref{Appendix_1}. In Equation (\ref{Ch4_EiCMI}), we have assumed that
the recolliding electron has been born at $\varphi_b(\varphi_r)$ with zero
velocity. After the recollision, both the electrons are free and the
energy of the system is
\begin{equation}\label{Ch4_EfCMI_}
E_f (\varphi) = \frac 12  [\K_1 + \A(\varphi)]^2 +  \frac 12  [\K_2 + \A(\varphi)]^2,
\end{equation}
where $\K_{1,2}$ are the asymptotic kinetic momenta at
$\varphi\to\infty$ of  the first and second electrons, correspondingly.

Now, substituting Equations (\ref{Ch4_EiCMI}) and (\ref{Ch4_EfCMI_}) into Equation
(\ref{Ch2_PropNonAdiabat_Summary}), we obtain the correlated spectrum standard for
the \gls{SFA},
\begin{eqnarray}\label{Ch4_SFASpectra}
\Gamma_{SFA}(\K_1, \K_2) &\propto& \exp\left(-\frac 2{\omega}\Im S_{SFA}(\K_1, \K_2)\right), \\
S_{SFA}(\K_1, \K_2)&=&  \int_{\Re\varphi_r^0}^{\varphi_r^0}
\left\{  \frac 12 \left[\K_1 + \A(\varphi)\right]^2 +  \frac 12 \left[\K_2 + \A(\varphi)\right]^2 \right.\nonumber\\
&& \left. \qquad -\frac 12 \left[ \A(\varphi)-\A(\Phi(\gamma; \varphi)) \right]^2 + I_p^{(2)} \right\} d\varphi,
\nonumber
\end{eqnarray}
where the phase of ``birth'' (ionization) $\Phi(\gamma; \varphi)$ corresponding to the recollison phase $\varphi$ and the transition point $\varphi_r^0$ are defined by Equations (\ref{Ch4_QuantPhaseBirth}) and (\ref{Ch4_DeltaEGamma}) in Appendix \ref{Appendix_1} and $\gamma$ is the Keldysh parameter for the ion [see Equation (\ref{Ch4_KeldyshParameter})].

The major stumbling block is to
account for the electron-electron and the electron-ion
interactions on the same footing, nonperturbatively. To include
these crucial corrections, we have to include the corresponding
Coulomb interactions into $E_{i,f}(\varphi)$. With the nucleus located
at the origin, the electron-electron and the electron-core
interaction energies are
\begin{equation}\label{Ch4_PotentialEnergies}
V_{ee} = 1/|\R_{12}(t)|, \qquad V_{en}^{(1,2)} = -2/|\R_{1,2}(t)|,
\end{equation}
correspondingly. Here $\R_{12}(t) = \R_1(t) - \R_2(t)$ and
$\R_{1,2}(t)$ are the trajectories of the two electrons.

However, we immediately see problems. The corrections depend on
the specific trajectory, and one needs to somehow decide what
this trajectory should be. Note that the classical trajectories
$\R_{1,2}(t)$ in the presence of the laser field and the Coulomb
field of the nucleus may even be chaotic. The solution to this
problem has already been discussed in the \gls{PPT} approach for single-electron ionization. In the
spirit of the eikonal approximation, these trajectories can be
taken in the laser field only \cite{Perelomov_1967_B, Perelomov_1968, Smirnova_2008}, so that they correspond to the saddle points of the
standard \gls{SFA} analysis. Not surprisingly, in the \gls{SFA} these
trajectories start at the origin,
$$
\R_{1,2}(t) = \int_{t_0}^{t} [ \K_{1,2} + \A(\tau)] d\tau.
$$

However, here we run into the second problem: the potentials
$V_{ee}$ and $V_{en}^{(1,2)}$ are singular. Consequently,  the
integral in Equation (\ref{Ch2_PropNonAdiabat_Summary}) is divergent and the result is
unphysical. Therefore, such implementation of the Coulomb
corrections requires additional care.

Appendix \ref{Appendix_2} describes an approach that deals with
these two problems, both defining the relevant trajectories and
removing the divergences of the integrals. Summarizing the results of Appendix \ref{Appendix_2}, we conclude that these problems are overcome by simply using the \gls{SFA} trajectories and soften versions of the Coulomb potentials (\ref{Ch4_PotentialEnergies}). Evidently, there are many ways to soften (i.e., to remove the singularity) of the Coulomb potentials; nevertheless, final results qualitatively are not affected by such freedom.

\section{The Correlated Two-electron Ionization within the \glsentryname{SFA}}\label{Ch4_s3}

In this section, we analyze the correlated spectrum of the \gls{NSDI} by
using the \gls{SFA}. In the next sections,
we will improve the \gls{SFA} result by employing the perturbation
theory in action with the \gls{SFA} result as the zero-order
approximation.

Ignoring the Coulomb corrections in Equation (\ref{Ch4_GenGamma}) and performing the saddle-point calculations described in Appendix \ref{Appendix_1},  we reach the usual \gls{SFA} expression for the correlated \gls{NSDI} spectra -- Equation (\ref{Ch4_SFASpectra}).

\begin{figure}
\begin{center}
\includegraphics[scale=0.6]{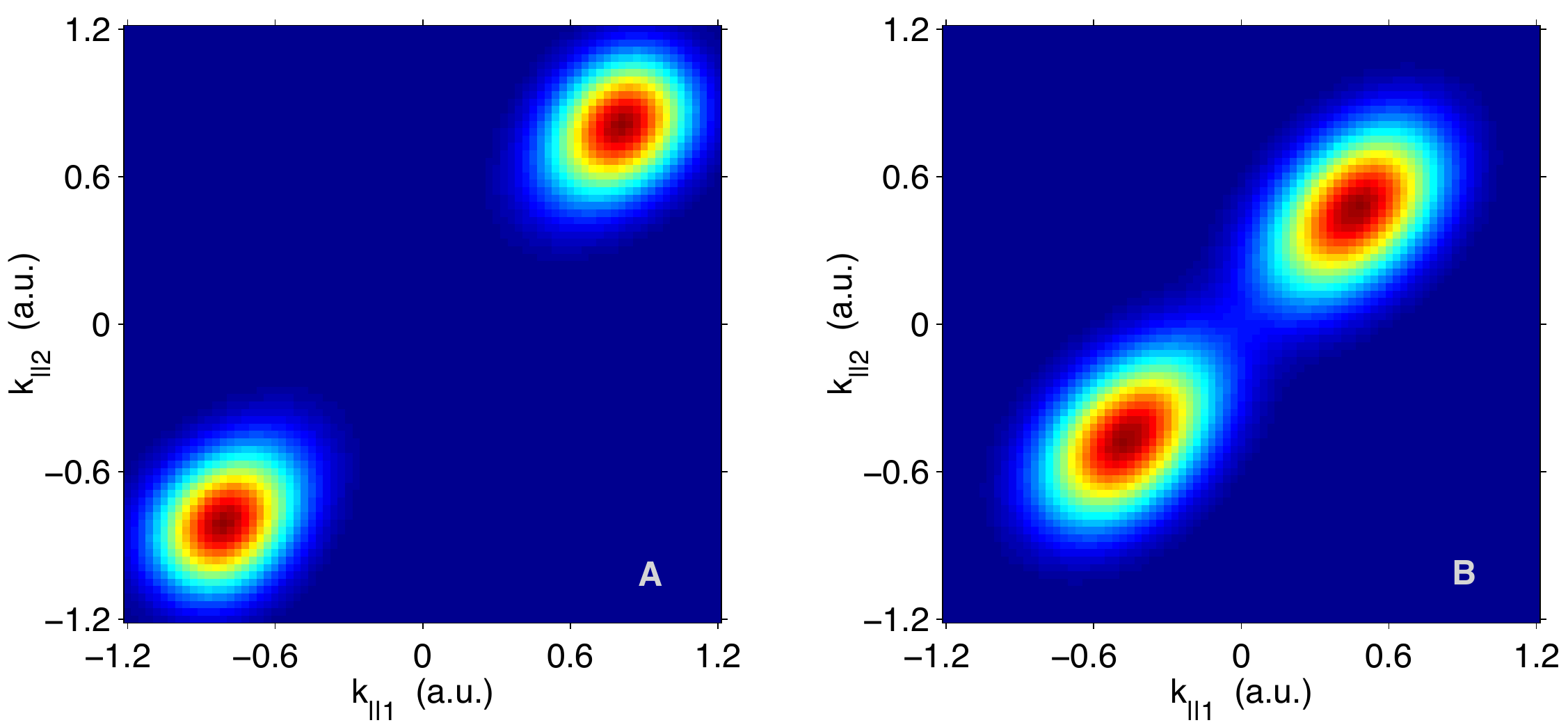}
\end{center}
\caption[Correlated two-electron spectra of Ar within the
\glsentryname{SFA}]{(Reprinted from Reference \cite{Bondar2009}. Copyright (2009) by the American Physical Society.) Correlated two-electron spectra (\ref{Ch4_SFASpectra}) of Ar (linear scale) within the
\glsentryname{SFA} at $7\times 10^{13}$ W/cm$^2$, 800 nm ($\k_1 = \k_2 = 0$)
(a) $\gamma=0$; (b) $\gamma = 1.373$. Maxima of figures correspond to probability densities: (a) $1.7\times10^{-6}$, (b) $2.9\times10^{-15}$.
}\label{Ch4_Fig1}
\end{figure}

To illustrate the \gls{SFA} results, we plot the two-electron correlated
spectrum for a system with $I_p$ of Ar in Figure \ref{Ch4_Fig1}. In Figure
\ref{Ch4_Fig1}(a) we set $\gamma=0$. Such an \gls{SFA} spectrum has a peak at
$\kk_1 = \kk_2 \approx -0.78$ a.u., which is the maximum of the
vector potential  $ -F/\omega \approx -0.78$ a.u. The last fact
has the following  interpretation: \gls{NSDI} is most efficient when the
velocity of the incident electron is maximal. This is achieved
near the zero of the laser field, ${\bf E}(\varphi) = {\bf
F}\cos\varphi$, and the maximum of $\A(\varphi)$. An electron
liberated at this time could acquire the final drift velocity
$\approx -F/\omega$. However, including the correct value of the
Keldysh parameter $\gamma$ not only substantially shifts the peak
position [Figure \ref{Ch4_Fig1}(b)], but also lowers the maximum by nine
orders of magnitude.

\section{Electron-Electron Interaction}\label{Ch4_s4}

In this section, we demonstrate the changes in the correlated spectrum due to
 the electron-electron repulsion.

Coulomb corrections to the single-electron \gls{SFA} theory were first
introduced by \gls{PPT} using the quasiclassical
(imaginary time) method (for reviews, see References \cite{Popov2004, Popov2005}). More recently, further improvements to this method
have been considered in References \cite{Popruzhenko2008a, Popruzhenko2008b, Popruzhenko2009}.  These improvements considered not only
subbarrier motion in imaginary time, but also the effects of the
Coulomb potential on the phase of the outgoing wave packet in
the classically allowed region. These improvements allowed the authors
of References \cite{Popruzhenko2008a, Popruzhenko2008b, Popruzhenko2009} to obtain
quantitatively accurate results not only for ionization yields,
but also for the above threshold ionization spectra of direct
electrons (i.e., not including recollision). An alternative, but
conceptually similar, approach is the \gls{SF-EVA}.
Unlike the two previous methods, the \gls{SF-EVA} allows a simple
treatment of the electron-electron and electron-ion interaction in
the two-electron continuum states.

 According to the \gls{SF-EVA}, the contribution of the interaction potentials is calculated along the \gls{SFA} trajectories,
$$
\R_{1,2}(\varphi) = \frac 1{\omega}\int_{\varphi_r^0}^{\varphi} \left[ \K_{1,2} + \A(\phi)\right] d\phi.
$$
Note that at the moment of recollision $\varphi_r^0$,  the
electrons are assumed to be at the origin, $ \R_{1,2}(\varphi_r^0)
={\bf 0}$. However, this does not cause any divergence since
according to Equation (\ref{Ch4_GenGamma}) we have to use the regularized
potential $V_{ee,lng}$.

From Equations (\ref{Ch4_PotentialsIntro}) and (\ref{Ch4_Potentials}), the
potential energy of electron-electron repulsion along these
trajectories is given by
\begin{eqnarray}\label{Ch4_Vee2}
V_{ee,lng}(\varphi) &=&  \frac 1{r_{12}(\varphi)} \left(1- \exp\left[ - \frac{r_{12}(\varphi)}{r_{12}^{(0)}}\right] \right), \nonumber\\
r_{12}(\varphi)  &=& |\R_1(\varphi) - \R_2 (\varphi) |
 = \sqrt{ \left[ (\kk_1-\kk_2)\frac{\varphi-\varphi_r^0}{\omega} \right]^2 + \left[ (\k_1-\k_2)\frac{\varphi-\varphi_r^0}{\omega}
\right]^2}.
\end{eqnarray}
As discussed in Appendix \ref{Appendix_2}, the parameter $r_{12}^{(0)}$ is set to the ionic radius, $r_{12}^{(0)}=1/I_p^{(2)}$.

The correlated spectrum, which accounts for the electro-electron interaction, has the form
\begin{eqnarray}\label{Ch4_SpectrSFAEE}
\Gamma_{ee}(\K_1, \K_2) &\propto&  \exp\left(-\frac 2{\omega} \Im\left[S_{SFA}(\K_1, \K_2) + S_{ee}(\K_1, \K_2) \right]\right), \\
S_{ee}(\K_1, \K_2) &=&  \int_{\Re\varphi_r^0}^{\varphi_r^0} V_{ee,lng}(\varphi)d\varphi. \nonumber
\end{eqnarray}

\begin{figure}
\begin{center}
\includegraphics[scale=0.6]{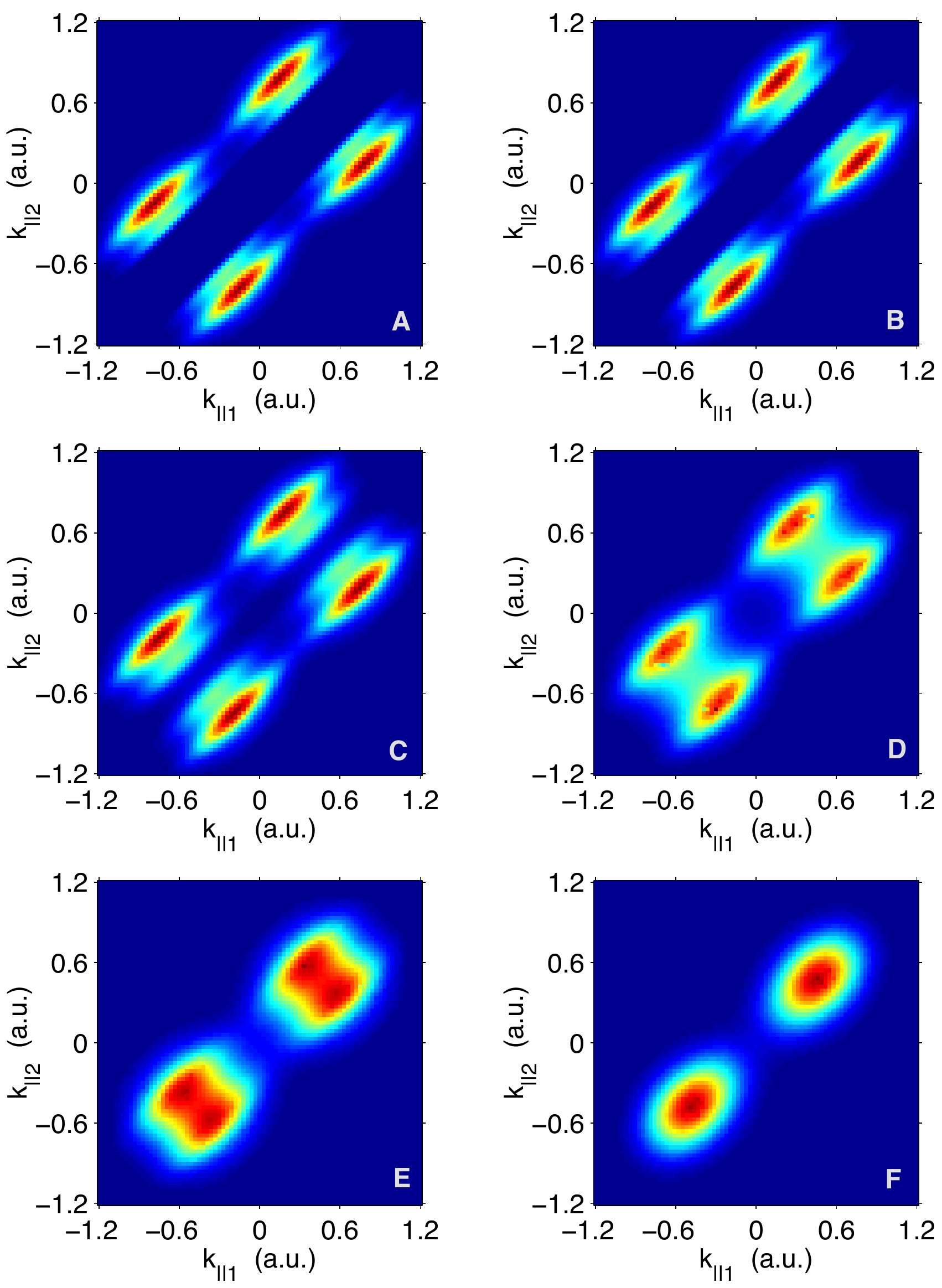}
\end{center}
\caption[Role of the electron-electron interaction in correlated two-electron
spectra of Ar]{(Reprinted from Reference \cite{Bondar2009}. Copyright (2009) by the American Physical Society.) Role of electron-electron interaction. Correlated
spectra  of Ar (linear scale) at $7 \times 10^{13}$ W/cm$^2$, 800
nm are calculated using Equation (\ref{Ch4_SpectrSFAEE}) with $r_{12}^{(0)}
= 0.985$ a.u. ($\gamma=1.373$). Electron-core interaction is not
included.  Spectra are shown for different values of $\k$ (in
a.u.) for both electrons: (a) $\k_1 = \k_2 = 0$; (b) $\k_1 = 0,
\k_2 = 0.2$; (c) $\k_1 = -0.1, \k_2 = 0.2$; (d) $\k_1=-0.2, \k_2 =
0.3$; (e) $\k_1 = -0.3, \k_2 = 0.3$; (f) $\k_1 = -0.5, \k_2 =
0.5$. Maxima of figures correspond to probability densities: (a)
$9.1\times 10^{-19}$, (b) $6.5\times 10^{-19}$, (c)  $8.36\times
10^{-19}$, (d) $8.4\times 10^{-19}$, (e) $9.1\times 10^{-19}$, (f)
$2.8\times 10^{-20}$. } \label{Ch4_Fig2}
\end{figure}

Figure \ref{Ch4_Fig2} shows the contribution of electron-electron
repulsion to the spectra of \gls{NSDI} for an atom with $I_p$ of
Ar (for experimental data see Reference \cite{Liu_2008} and Figure \ref{Ch4_Fig5}).

Comparing Figures \ref{Ch4_Fig1} and \ref{Ch4_Fig2}, we readily notice a
dramatic  influence of electron-electron interaction on the
correlated spectra. Electron-electron repulsion splits each \gls{SFA}
peak into two peaks because, due to the Coulomb interaction, two
electrons cannot occupy the same volume. Note that the larger the
difference between the perpendicular momenta of both the
electrons, the closer is the location of the peaks.

\section{Electron-Ion Interaction}\label{Ch4_s5}

Now we include the electron-ion attraction. The potential energy of  electron-ion
interaction for the case of two electrons and a single core, after partitioning (\ref{Ch4_PotentialsIntro}) and (\ref{Ch4_Potentials}), is
\begin{eqnarray}\label{Ch4_Vei2}
V_{en,lng}(\varphi) &=& -2\sum_{i=1}^2  \frac 1{r_{i}(\varphi)} \left(1- \exp\left[ - \frac{r_{i}(\varphi)}{r_0}\right] \right), \nonumber\\
r_{1,2}(\varphi) &=&
\sqrt{\left( \kk_{1,2}
\frac{\varphi-\varphi_r^0}{\omega} + \frac F{\omega^2}(
\cos\varphi - \cos\varphi_r^0 )\right)^2 + \left(
\k_{1,2}\frac{\varphi-\varphi_r^0}{\omega} \right)^2 }.
\end{eqnarray}
As far as the parameter $r_0$ is concerned, we set it equal to $r_{12}^{(0)}=1/I_p^{(2)}$.

Finally, the correlated spectrum of \gls{NSDI}, which takes into account
both the electron-electron and electron-ion interactions, reads
\begin{eqnarray}\label{Ch4_SpectrSFAEEEI}
\Gamma_{ee+en}(\K_1, \K_2) &\propto&
\exp\left(-\frac 2{\omega} \Im\left[S_{SFA}(\K_1, \K_2) + S_{ee}(\K_1, \K_2) + S_{en}(\K_1, \K_2)\right]\right), \\
S_{en}(\K_1, \K_2) &=& \int_{\Re\varphi_r^0}^{\varphi_r^0} V_{en,lng}(\varphi)d\varphi. \nonumber
\end{eqnarray}

\begin{figure}
\begin{center}
\includegraphics[scale=0.6]{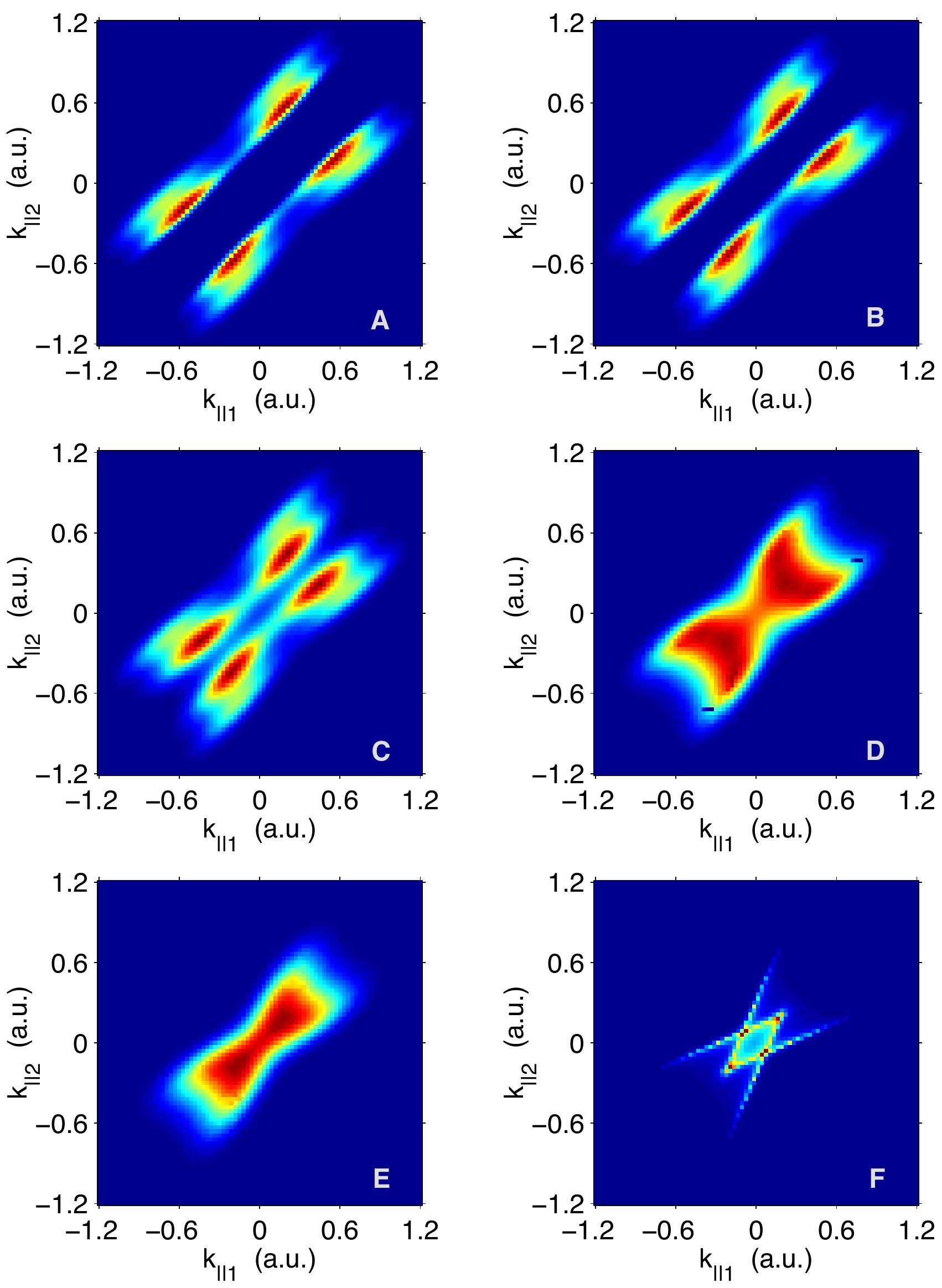}
\end{center}
\caption[Role of the electron-core and electron-electron interactions in correlated two-electron spectra
of Ar]{(Reprinted from Reference \cite{Bondar2009}. Copyright (2009) by the American Physical Society.) Role of the electron-core and electron-electron interactions. Correlated spectra
of Ar (linear scale) at $7 \times 10^{13}$ W/cm$^2$, 800 nm are
calculated using Equation (\ref{Ch4_SpectrSFAEEEI}) with $r_{12}^{(0)}=r_0
= 0.985$ a.u. ($\gamma=1.373$). Both electron-elecron and
electron-core interactions are  included. Spectra are shown for
different values of $\k$ (in a.u.) for both electrons: (a) $\k_1 =
\k_2 = 0$; (b) $\k_1 = 0, \k_2 = 0.2$; (c) $\k_1 = -0.1, \k_2 =
0.2$; (d) $\k_1=-0.2, \k_2 = 0.3$; (e) $\k_1 = -0.3, \k_2 = 0.3$;
(f) $\k_1 = -0.5, \k_2 = 0.5$. Maxima of figures correspond to
probability densities: (a) $2.8\times10^{-5}$, (b)
$3.3\times10^{-5}$, (c) $6.1\times10^{-5}$, (d)
$1.5\times10^{-4}$, (e) $4.4\times10^{-4}$, (f)
$7.7\times10^{-4}$. } \label{Ch4_Fig3}
\end{figure}

\begin{figure}
\begin{center}
\includegraphics[scale=0.6]{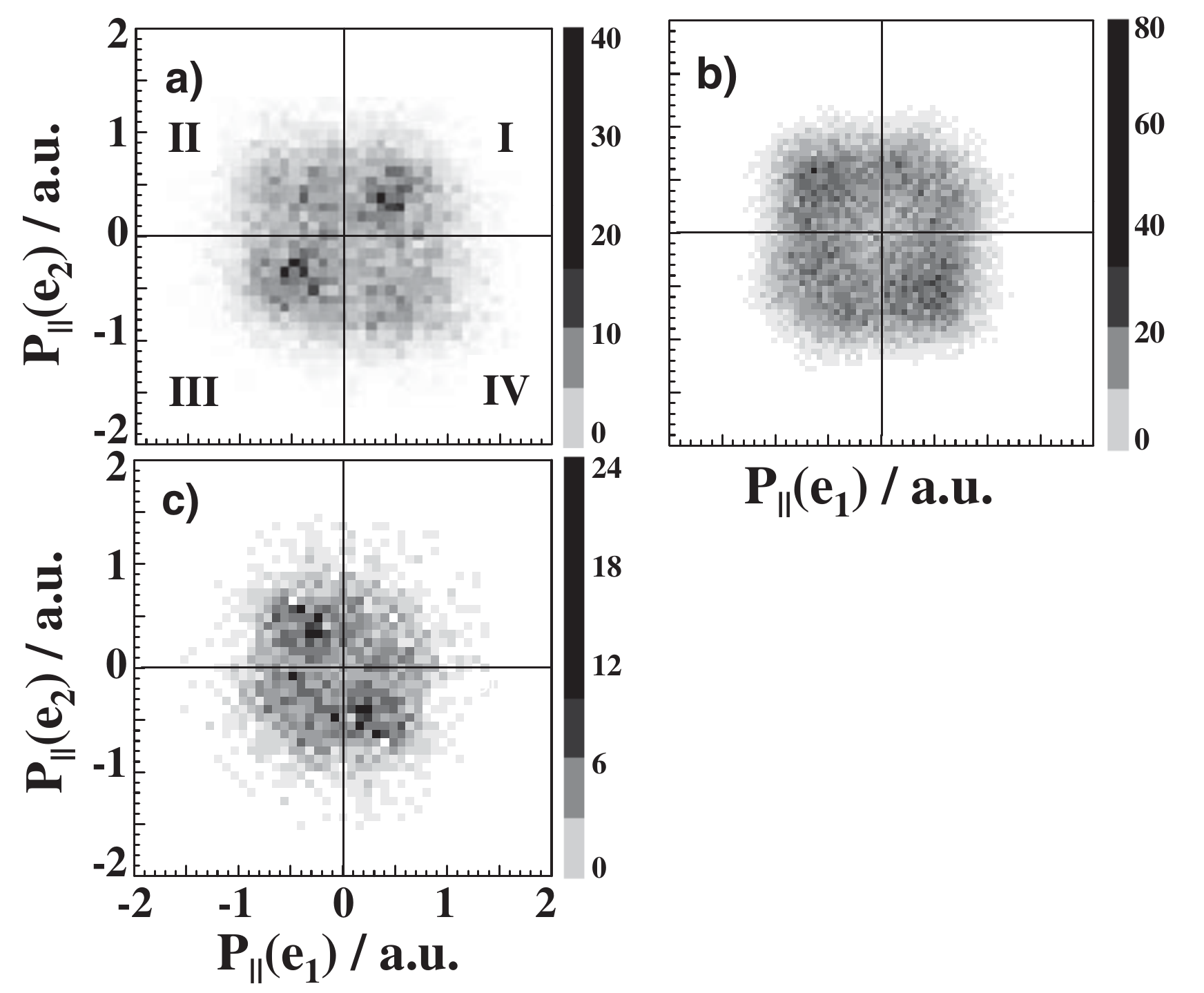}
\caption[Experimentally measured two-electron correlated spectra of Ar]{(Reprinted from Reference \cite{Liu_2008} with kind permission of Y. Liu. Copyright (2008) by the American Physical Society.) Experimentally measured correlated spectra of Ar at 800 nm (a) $9\times 10^{13}$ W/cm$^2$ \cite{Eremina2003}, (b) $7\times 10^{13}$ W/cm$^2$ \cite{Liu_2008}, and (c) $4 \times 10^{13}$ W/cm$^2$ \cite{Liu_2008}.}\label{Ch4_Fig5}
\end{center}

\begin{center}
\includegraphics[scale=0.6]{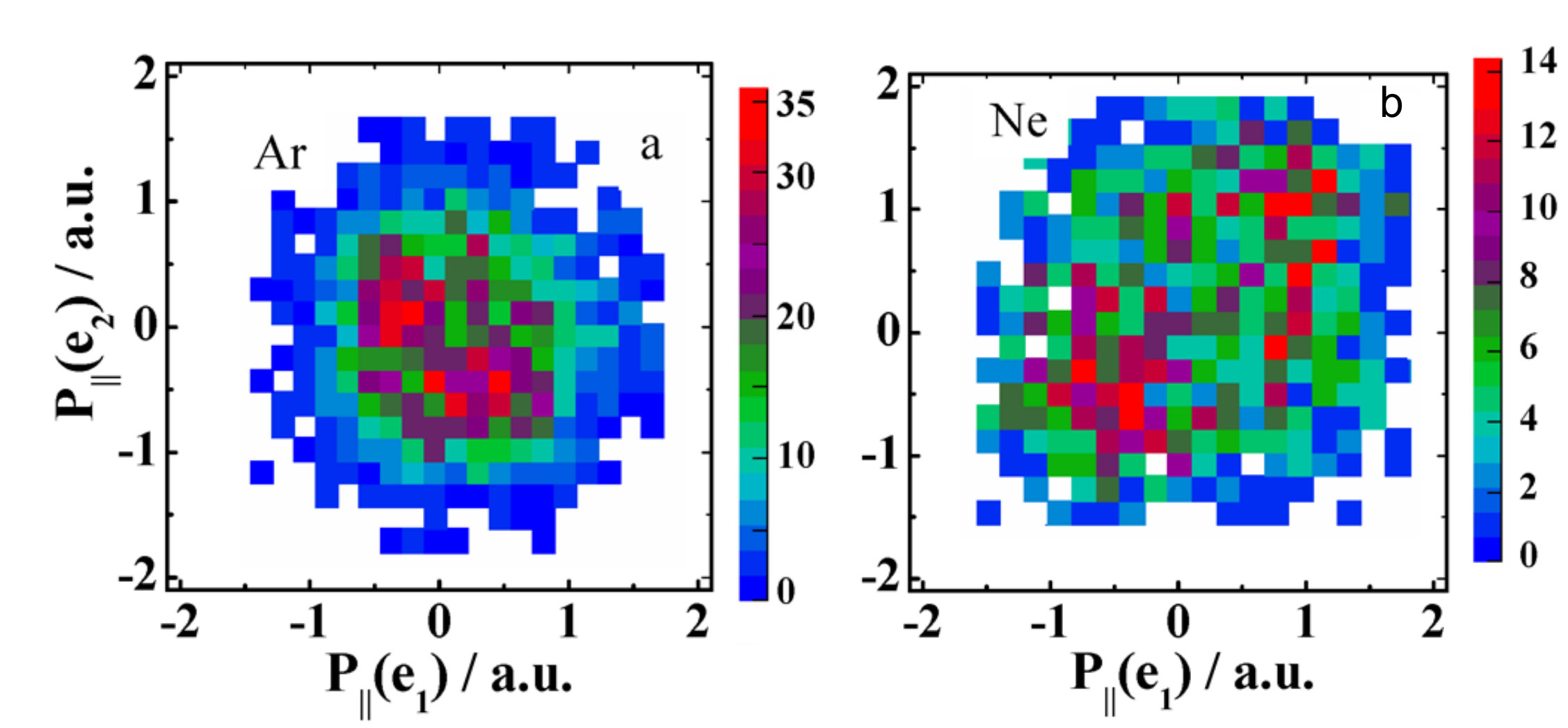}
\end{center}
\caption[Experimentally measured two-electron correlated spectra of Ar and Ne]{(Reprinted from Reference \cite{Liu2010} with kind permission of Y. Liu. Copyright (2010) by the American Physical Society.) Experimentally measured correlated spectra of (a) Ar at $3\times 10^{13}$ W/cm$^2$, 800 nm and (b) Ne at $1.5 \times 10^{14}$ W/cm$^2$, 800 nm in Reference \cite{Liu2010}.}\label{Ch4_Fig6}
\end{figure}

To illustrate the influence of the electron-ion attraction, we
have plotted the correlated spectra of Ar in Figure \ref{Ch4_Fig3} for
different perpendicular momenta. Comparing Figures \ref{Ch4_Fig2} and
\ref{Ch4_Fig3}, we conclude that the  larger the difference between
the perpendicular momenta of the two electrons, the larger is the
contribution of the electron-ion interaction. Furthermore,
accounting for electron-ion attraction increases the probability of
\gls{NSDI} by 15 orders of magnitude. This occurs because, as in
the case of single-electron ionization, electron-core interaction
significantly lowers an effective potential barrier. We can also
conclude that correlated spectra pictured in Figures \ref{Ch4_Fig3}(c),
\ref{Ch4_Fig3}(d), and \ref{Ch4_Fig3}(e) have the biggest contribution to
the total probability of \gls{NSDI}, which is an integral of the
probability density over momenta of both the electrons. Note that,
on the one hand, the maximum of the probability density shown in
Figure \ref{Ch4_Fig3}(f) is the largest among those presented in Figure
\ref{Ch4_Fig3}, and on the other hand, this maximum is localized in a few
pixels; therefore, the integral contribution of Figure \ref{Ch4_Fig3}(f)
to the total probability is smaller than Figure \ref{Ch4_Fig3}(e).
Additionally, as one would expect, further increasing of $\k$ leads
to a decrease of probability density.  The correlated spectra in Figure
\ref{Ch4_Fig3}  agree with the experimental data \cite{Liu_2008, Liu2010} in
quadrants one and three (see Figures \ref{Ch4_Fig5} and \ref{Ch4_Fig6}). The considered diagram (Figure \ref{Ch4_Fig4})
does not contribute to signals in quadrants two and four. Note
that taking into account a nonzero value of  $\gamma$ is vital to
achieve agreement with the experimental data.

From Equations (\ref{Ch4_Vee2}) and (\ref{Ch4_Vei2}), we can notice that  if
$r_{12}^{(0)} \to \infty$ and $r_0 \to \infty$, the Coulomb
corrections $V_{ee,lng}$ and $V_{en,lng}$ vanish, and the \gls{SFA}
result is recovered. Therefore, we conclude that the radii
$r_{12}^{(0)}$ and $r_0$ contain the information about the initial
position of electrons after they emerged in the continuum.
Obviously, the intra-electron distance should be on the order of
an ion radius. 

\section{Conclusions}\label{Ch4_s6}

The analytical quantum-mechanical theory
of \gls{NSDI} within the deeply quantum regime, when the energy of the active electron driven by the
laser field is insufficient to collisionally ionize the parent
ion, has been formulated based on the adiabatic approach. On the whole, the presented model
qualitatively agrees with available experimental data \cite{Eremina2003, Liu_2008, Liu2010}. We have defined the
quantum-mechanical phase of birth of the active electron
(\ref{Ch4_QuantPhaseBirth}), which accurately accounts for tunnelling
of the recolliding electron in the regime where both the phases
$\varphi_r$ and $\varphi_b$ are complex.  Moreover, it has been
demonstrated that ignoring such a contribution of tunnelling of the
active electron fails to agree with the experimental data.

Furthermore, our results show that any attempt to interpret \gls{NSDI}
spectra  in this regime in terms of a simple \gls{SFA}-based streaking
model would lead to wrong conclusions on the relative dynamics of
the two electrons.

The contributions of the electron-electron and electron-ion
interactions  have been analyzed. Both play an important and
distinct role in forming the shape of the correlated spectra.

The presented model is not able to reproduce the correlated spectra obtained experimentally \cite{Liu_2008, Liu2010} in
quadrants two and four (see Figures \ref{Ch4_Fig5} and \ref{Ch4_Fig6}). It is because the considered process of \gls{SEE}, when
two electrons detach {\it simultaneously} from the atom, does not
 contribute to that area. However, it is widely accepted that the anticorrelation of the electrons in those parts of the correlated spectra is formed due to
\gls{RESI} \cite{Feuerstein_2001, Jesus_2004, Haan2008a, Emmanouilidou2009, Shaaran2010, Shaaran2010a}, and it should be
noted  that this mechanism has been also observed in classical
simulations \cite{Haan2006, Haan2010, Ye2010}.

\glsresetall

\chapter{Nonsequential Double Ionization below Intensity Threshold: Anticorrelation of Electrons without Excitation of the Parent Ion}\label{chapter5}

Is excitation of the parent ion indeed necessary to explain the electron anticorrelation in \gls{NSDI} \gls{BIT}?

We show that this is not always the case. Our conclusions are based, first, on model {\it ab-initio} calculations showing that the anticorrelation of the electrons exists even if the ion has only a single bound state. Second, we present a simple analytical model based on the assumption that both the electrons are ejected simultaneously (the time of return of the first electron coincides with the time of liberation of the second electron, i.e., the ion has no time to be excited). An advantage of this model is that it allows for a simple analytical solution in closed form. In a certain range of parameters, the correlated two-electron spectrum obtained within this model exhibits the anticorrelation of the electrons. This novel mechanism of simultaneous electron emission can produce the anticorrelation of the electrons in the deep \gls{BIT} regime.

\section{An Ab Inito Evidence of the New Mechanism} 

Consider the model Hamiltonian for a system of two one-dimensional electrons 
\begin{eqnarray}\label{Ch5_Model_Hamiltonian}
\h(t) &=& -\frac 12 \left( \frac{\partial^2}{\partial x_1^2} + \frac{\partial^2}{\partial x_2^2} \right) + V(x_1)  + V(x_2)  - V(x_2 - x_1) + (x_1 + x_2) F_L(t),
\end{eqnarray}
where $V(x) =  -4\exp(-3x^2)$ is the prototype for the potential of the electron-core attraction, $-V(x_2 - x_1)$ is the prototype of the electron-electron repulsion, and $F_L(t) = F f(t) \sin (\omega t)$ ($F=0.05$ and $\omega = 0.6$) -- the laser pulse, where $f(t)$ is a trapezoid with one-cycle turn-on, six-cycle full strength, and one-cycle turn-off. The potential $V(x)$ is chosen such that the one-particle Hamiltonian, $\hat{H}_{ion} = -\partial^2/(2 \partial x^2) + V(x)$, supports only one bound state with the ionization potential $I_p^{(2)} = 2.11$. 

To make analysis more transparent, we substitute the original problem (\gls{NSDI}) by the corresponding problem of laser-assisted scattering, i.e., we simply discard the first step of \gls{NSDI} -- liberation of the first electron. In other words, instead of assuming that the two-electron system initially is in the ground state, we assume that the first electron is an incident wave packet and the second electron is in the single bound state of the ion. This modification is very useful, as it allows us to eliminate all other possible interactions and processes except the three major components -- the electron-electron repulsion, the electron-core attraction, and the electron-laser field interaction. 

We solve numerically, by means of the split-operator method, the time-dependent Schr\"{o}dinger equation,
\begin{eqnarray}\label{Ch5_SchrodingerEq_ModelWaveFunc}
i\partial \Psi (x_1, x_2; t) /\partial t = \h(t)\Psi (x_1, x_2; t). 
\end{eqnarray}
The initial unsymmetrized wave function takes the form
\begin{eqnarray}
\Psi (x_1, x_2; 0) = \exp\left[ { -(x_1- \mu)^2/(2\sigma_x)^2 } \right]\psi ( x_2 ),
\end{eqnarray}
where $\sigma_x = 2$, $\mu = 7\sigma_x$, and $\hat{H}_{ion} \psi ( x ) = -I_p^{(2)} \psi ( x )$. The parameters are selected such that,
$
 I_p^{(2)} - (F/\omega)^2/2 \approx 4 \omega
$,
i.e., the bound electron needs to absorb at least four photons to be liberated. The wave function then is symmetrized, i.e., we assume that the spins of the electrons are antiparallel.

The wave function of the model system in the coordinate and momentum representations is pictured in Figure \ref{Ch5_wavefunctions_Fig}.  The wave function in the absence of the laser produces only maxima on the axes $x_{1,2}$ corresponding to one electron bound and the other one free; hence, this part of the wave function, which also shows up when the laser file is turned on, should be ignored because it does not correspond to \gls{NSDI}. From Figure \ref{Ch5_wavefunctions_Fig}, we observe that the two electrons ``prefer'' to be anticorrelated rather than correlated even though we effectively soften the repulsion between them by selecting antiparallel spins.

\begin{figure}
\begin{center}
\includegraphics[scale=0.7]{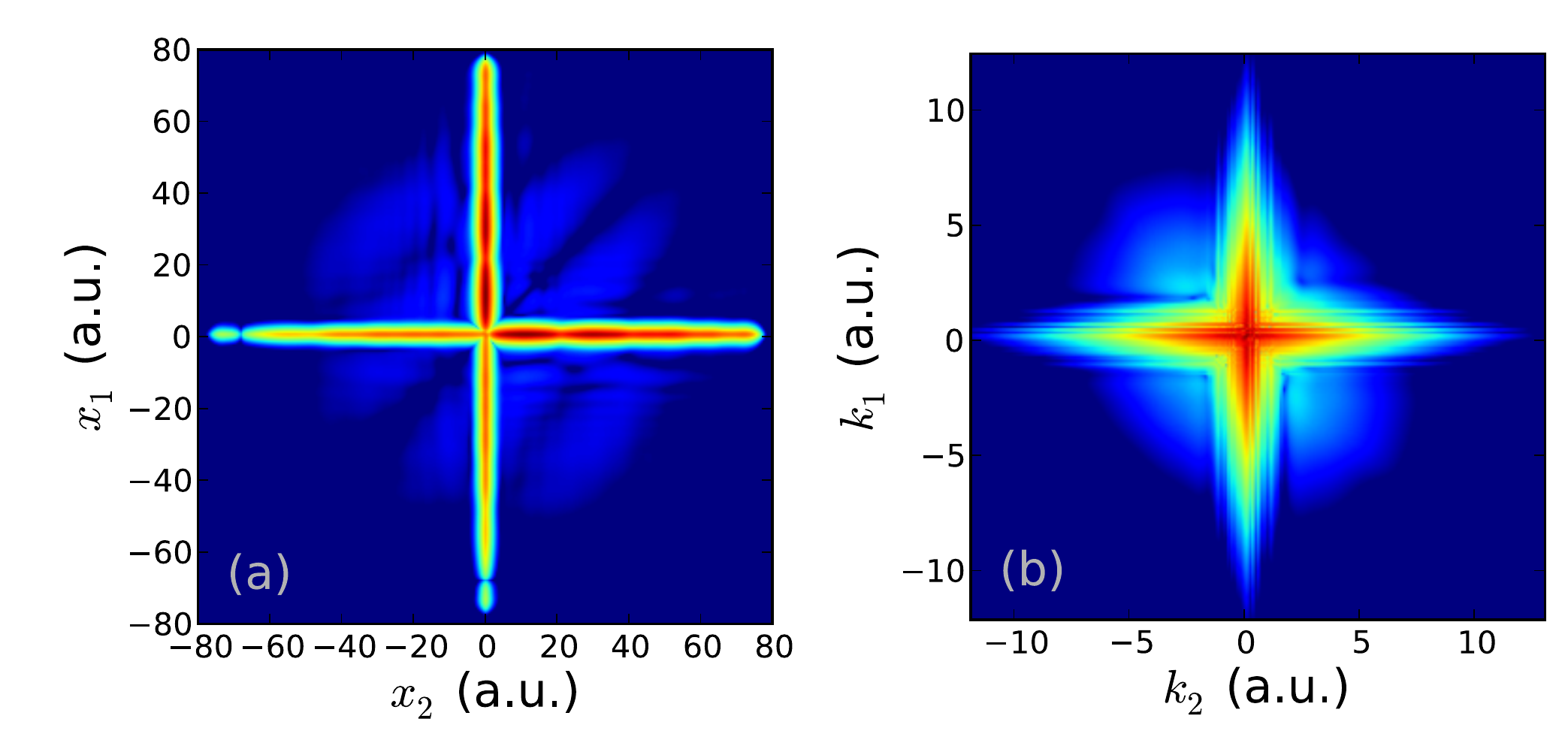}
\end{center}
\caption[Plots of the wave functions of the model system]{Plots of the wave functions of the model system [Equations (\ref{Ch5_Model_Hamiltonian}) and (\ref{Ch5_SchrodingerEq_ModelWaveFunc})]. (a) $\log \left |\Psi (x_1, x_2; t) \right|^2$ and (b) $\log \left |\widetilde{\Psi} (k_1, k_2; t) \right|^2$. Here, $t=73$ and $\widetilde{\Psi}$ is the Fourier transform of $\Psi$.}\label{Ch5_wavefunctions_Fig}
\end{figure}

\section{An Analytical Quasiclassical Expression of Correlated Spectra for the \glsentryname{SEE} Mechanism}

We shall calculate the correlated spectrum within the adiabatic approximation.

As in the previous section, we replace the problem of \gls{NSDI} \gls{BIT} by the problem of laser-assisted scattering of an electron by an ion. We study the \gls{SEE} process, when the moment of collision of the incident electron coincides with the moment of ionization of the ion. 

To include electron-electron repulsion, we follow Chapter \ref{chapter4} and apply the standard exponential perturbation theory -- \gls{SF-EVA}. Namely, first we find electron  energies and trajectories without the electron-electron repulsion. Then, we correct electron action and energies by adding the effect of the electron-electron repulsion, which are calculated along the zero-order trajectories.

Thus, in zero order, we define $E_{i,f}(\varphi)$ without the electron-electron repulsion. Since before ionization, one electron is free and the other is bound, the total classical energy of the system before the collision is
\begin{eqnarray}\label{Ch5_InitialEnergy}
E_i(\varphi) &=& \frac 12 \left[ \P + \A(\varphi) \right]^2 - I_p^{(2)}, 
\end{eqnarray}
where $\P$ is the canonical momentum (i.e., the kinetic momenta at $\varphi = \pm\infty$) of the incident electron, $I_p^{(2)}$ is the ionization potential of the ion, and 
$
\A(\varphi) = -({\bf F}/\omega) \sin\varphi
$
is the vector potential of a linearly polarized laser field. After collision both the electrons are free, and the classical energy of the system reads
\begin{eqnarray}\label{Ch5_FinalEnergy}
E_f(\varphi) &=& \frac 12 \left[ \K_1 + \A(\varphi) \right]^2 + \frac 12 \left[ \K_2 + \A(\varphi) \right]^2,
\end{eqnarray}
where $\K_{1,2}$ are canonical momenta of the first and second electrons.

Such two electron process is formally equivalent to single-electron strong field ionization of a quasiatom (within a pre-exponential accuracy). This statement is manifested by the following identity 
\begin{eqnarray}\label{Ch5_SingleElectronIonizationAndSEE}
 \frac 1{\omega} \int^{\varphi_0} \left[ E_f(\varphi) - E_i(\varphi) \right]d\varphi  &=& \frac 1{\omega} \int^{\varphi_0} \left\{ \frac 12 \left[ \mathrsfs{K} + A(\varphi) \right]^2 + \mathrsfs{I}_p \right\}d\varphi,
\end{eqnarray}
where the right hand side of Equation (\ref{Ch5_SingleElectronIonizationAndSEE}) is the action of single-electron ionization within the \gls{SFA} [Equation (\ref{Ch3_SingleElectronIonizAction})], $\mathrsfs{K} = \kk_1 + \kk_2 - \pp$ is the effective longitudinal momentum, and  $\mathrsfs{I}_p = I_p^{(2)} + \left(  \K_1^2 + \K_2^2 - \P^2 - \mathrsfs{K}^2 \right)/2$ 
is the effective ionization potential of the quasiatom. Therefore, derivation of the correlated spectra for \gls{SEE} is reduced to calculation of the momentum distribution of photoelectrons after single-electron ionization without any restrictions on $\mathrsfs{K}$. The most suitable solution of the last problem for our current discussion is given by Equation (\ref{Ch3_GeneralRate}). However, this equation is valid only for positive $\mathrsfs{I}_p$; hence, the case of $\mathrsfs{I}_p < 0$ being of interest for \gls{SEE}, must also be considered.

Substituting Equations (\ref{Ch5_InitialEnergy}) and (\ref{Ch5_FinalEnergy}) into Equation (\ref{Ch2_PropNonAdiabat_Summary}) and taking into account the previous comments, we obtain
\begin{eqnarray}
&& \Gamma (\K_1, \K_2) \propto \exp\left( -2 \left|\mathrsfs{I}_p\right| G /\omega + V_{ee} \right), \label{Ch5_QuasiClassicalCorrSpectra}
\end{eqnarray}
where
\begin{eqnarray}
\varphi_0 &=& \left\{
\begin{array}{ccc}
\arcsin x_- + i\arccosh x_+ &\mbox{if}& \mathrsfs{I}_p < 0 \\
 \arcsin y_- + i\arccosh y_+ &\mbox{if}&  \mathrsfs{I}_p > 0,
\end{array}
\right. \\
G &=& \left\{
\begin{array}{r}
\left( \eta^2 + \frac 1{2g^2} - 1 \right) \arccosh x_+ - \sqrt{x_+^2 - 1} \left( \frac{2\eta x_-}g + x_+ \frac{1-2x_-^2}{2g^2} \right)  \mbox{ if } \mathrsfs{I}_p < 0, \\
\left( \eta^2 + \frac 1{2g^2}  +1 \right) \arccosh y_+ - \sqrt{y_+^2 - 1}\left( \frac{2\eta y_-}g + y_+ \frac{1-2y_-^2}{2g^2}\right)  \mbox{ if }  \mathrsfs{I}_p > 0, 
\end{array}
\right. \\
x_{\pm} &=& \left| g(\eta-1) + 1\right|/2 \pm \left| g(\eta-1) - 1 \right|/2, \nonumber\\
y_{\pm} &=& \sqrt{ (g\eta+1)^2 + g^2}/2 \pm \sqrt{ (g\eta -1)^2 + g^2 }/2 ,\nonumber\\
g &=& \omega \sqrt{2\left|\mathrsfs{I}_p\right|} / F, \qquad \eta = \mathrsfs{K} / \sqrt{2\left|\mathrsfs{I}_p\right|}.\nonumber
\end{eqnarray}
In Equation (\ref{Ch5_QuasiClassicalCorrSpectra}), $V_{ee}$ denotes a crucial correction due to the electron-electron repulsion, 
\begin{eqnarray}
V_{ee} &=& -\frac 2{\omega} \Im \int_{\Re\varphi_0}^{\varphi_0} \frac{d\varphi}{\left| \R_1(\varphi) - \R_2 (\varphi) \right|}= \lim_{\varepsilon\to 0} \frac{-2}{\left| \K_1 - \K_2 \right|} \Im\int_{\Re\varphi_0}^{\varphi_0} \frac{d\varphi}{\sqrt{(\varphi-\varphi_0)^2 + \varepsilon^2}} \nonumber\\
&=& -\pi / \left| \K_1 - \K_2 \right| .
\end{eqnarray}

\begin{figure}
\begin{center}
\includegraphics[scale=0.8]{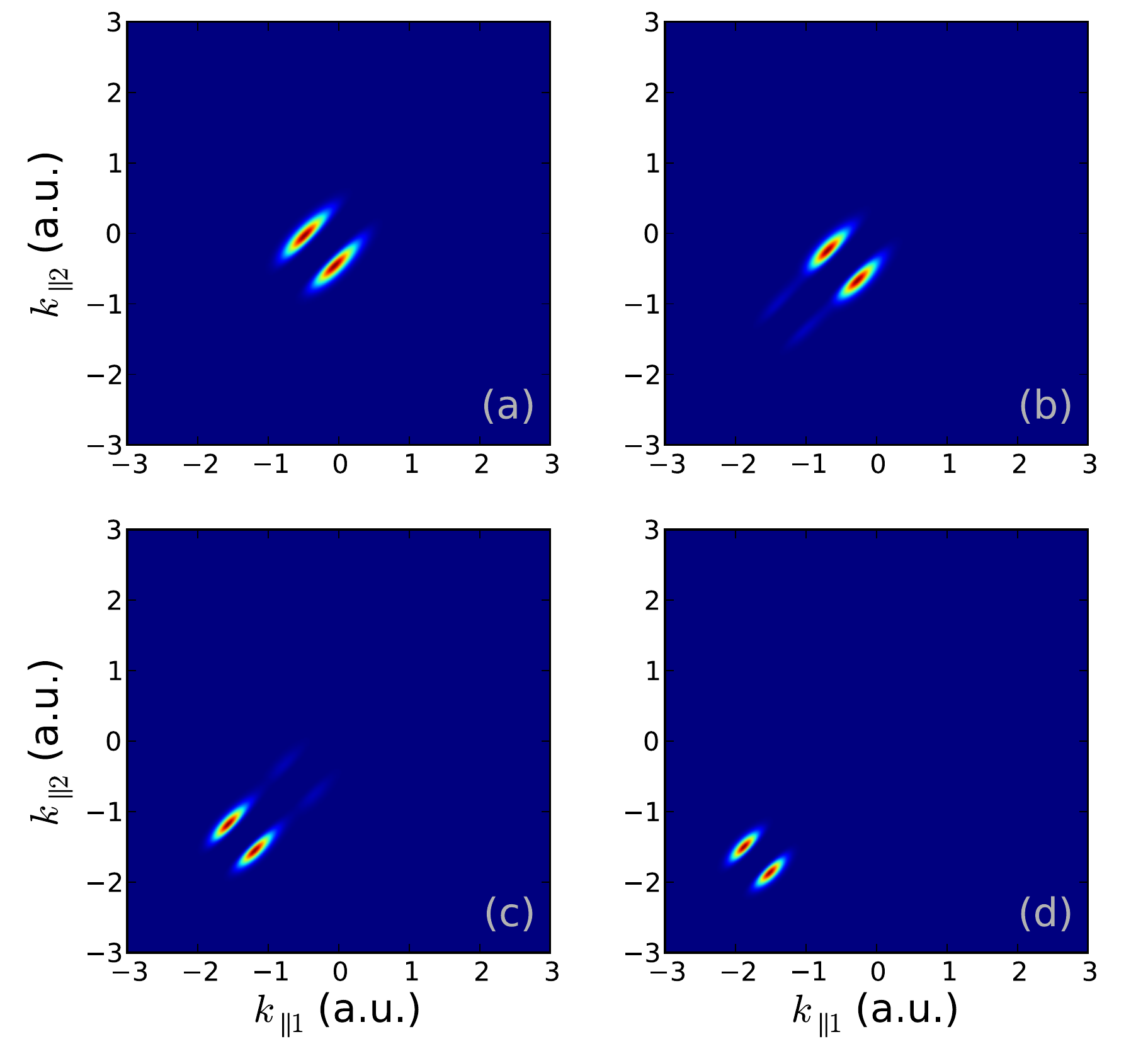}
\end{center}
\caption[Correlated two-electron spectra of Ar for different momenta of the incident electron]{Correlated two-electron spectra of Ar (linear scale) given by Equation (\ref{Ch5_QuasiClassicalCorrSpectra}) at  800 nm and $1\times 10^{13}$ W/cm$^2$ ($\p = \k_1 = \k_2 = 0$) for different momenta of the incident electron: (a) $\pp = 0.1$, (b) $\pp = 0.2$, (c) $\pp=0.25$, (d) $\pp=0.4$.}\label{Ch5_IncidentMomentumDependence}
\end{figure}

\begin{figure}
\begin{center}
\includegraphics[scale=0.7]{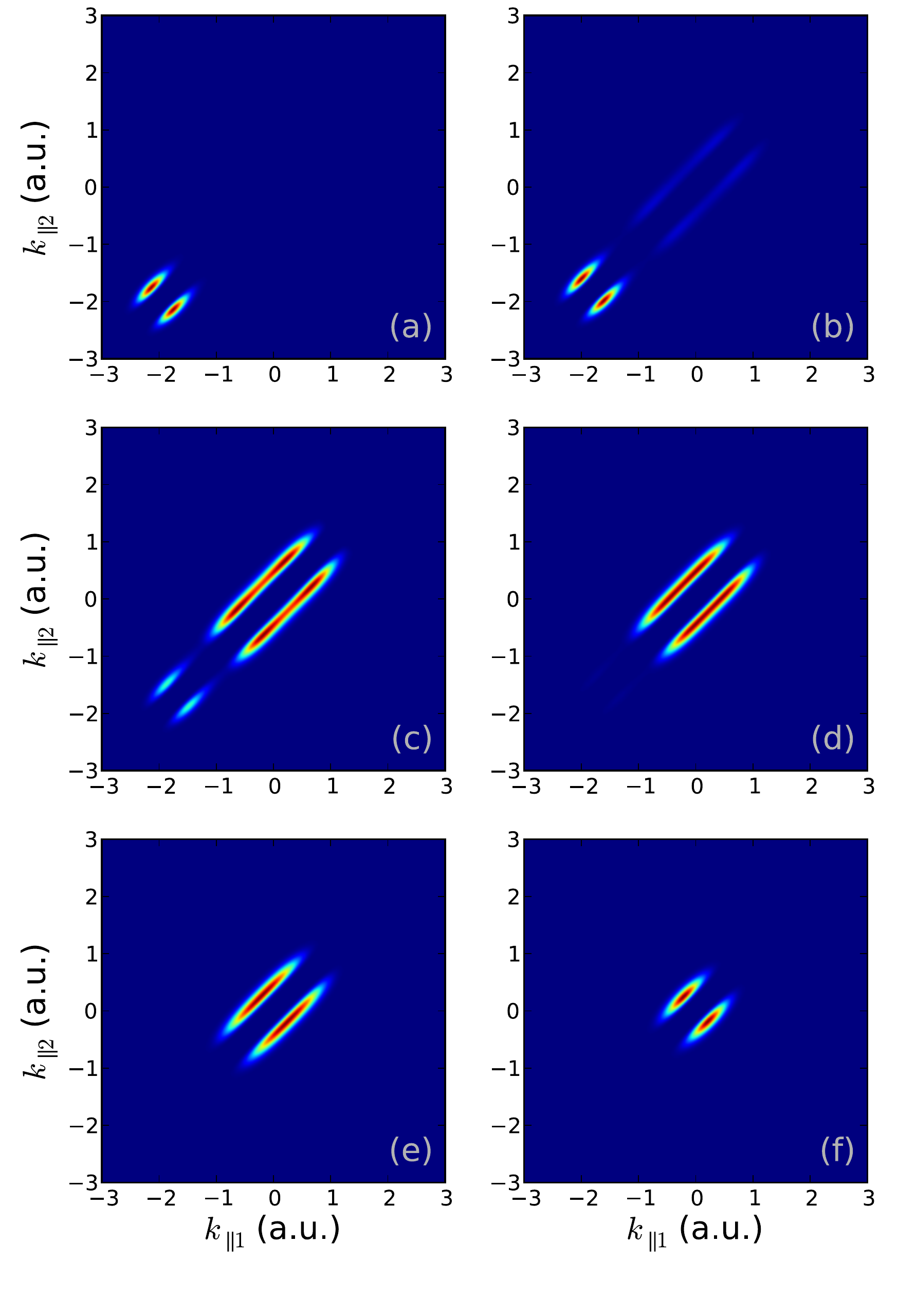}
\end{center}
\caption[Correlated two-electron spectra of Ar for different intensities of the laser field]{Correlated two-electron spectra of Ar (linear scale) given by Equation (\ref{Ch5_QuasiClassicalCorrSpectra}) at  800 nm ($\P = {\bf 0}$ and $\k_1 = \k_2 = 0$) for different intensities of the laser field: (a) $3\times 10^{13}$ W/cm$^2$, (b) $2.7\times 10^{13}$ W/cm$^2$, (c) $2.5\times 10^{13}$ W/cm$^2$, (d) $2.3\times 10^{13}$ W/cm$^2$, (e) $2\times 10^{13}$ W/cm$^2$, and (f) $1\times 10^{13}$ W/cm$^2$.}\label{Ch5_IntensityDependence}
\end{figure}

The model of \gls{NSDI} \gls{BIT} presented here [Equation (\ref{Ch5_QuasiClassicalCorrSpectra})] is similar to the one developed in Chapter \ref{chapter4}. However, there is an important difference. In the correlated spectra (\ref{Ch5_QuasiClassicalCorrSpectra}), the momentum of the incident electron ${\bf p}$ is a free parameter, whereas the correlated spectrum given by Equation (\ref{Ch4_SpectrSFAEEEI}) does not have this freedom -- the canonical momentum of the recolliding electron is fixed as a function of the phase of recollision [see Eq. (\ref{Ch4_QuantPhaseBirth})]. Therefore, Equation (\ref{Ch5_QuasiClassicalCorrSpectra}) allows us to establish the link between values of the canonical momentum of the rescattered (incident) electron and specific portions of the correlated spectra. The dependance of the correlated spectra (\ref{Ch5_IncidentMomentumDependence}) on the momentum of the incident electron is presented in Figure \ref{Ch5_IncidentMomentumDependence}.

However, more interesting question is the dependence of the correlated spectrum (\ref{Ch5_QuasiClassicalCorrSpectra}) on the intensity of the laser field, which is pictured in Figure \ref{Ch5_IntensityDependence}. The parameters used to plot Figure \ref{Ch5_IntensityDependence}(a) coincide with the parameters employed in the recent experiment \cite{Liu2010} (see also Figure \ref{Ch4_Fig6}). Yet, the experimentally measured correlated spectrum exhibits peaks in the second and fourth quadrants, which contradicts Figure \ref{Ch5_IntensityDependence}(a). The reason of such a disagreement is that the anticorrelation in these experimental data is due to the \gls{RESI} mechanism. As the intensity lowers, the peaks in the correlated spectra of \gls{SEE} shift to the second and fourth quadrants [see Figure \ref{Ch5_IntensityDependence}(f)], i.e., the \gls{SEE} process leads to the anticorrelation of the electrons in the deep \gls{BIT} regime. This can be explain intuitively in the following way: The lower the intensity, the smaller the canonical momentum of the returning (incident) electron. Since the canonical momentum of the system is approximately conserved, we obtain ${\bf 0} \approx \K_1 + \K_2$; hence, $\K_1 \approx -\K_2$. Having absorbed a necessary number of photons, both the electrons emerge in the continuum where they experience strong electron-electron repulsion that pushes them apart. Therefore, during \gls{SEE}, the electrons gain momenta because of the electron-electron repulsion. Indeed, if the electron-electron interaction is ``turned off''  in Equation (\ref{Ch5_QuasiClassicalCorrSpectra}) by setting $V_{ee} = 0$, then we obtain a single peak centred at the origin instead of the two peaks visible in Figure \ref{Ch5_IntensityDependence}(f). Recall that contrary to \gls{SEE}, the electrons gain kinetic energy by oscillating in the laser field in the course of \gls{RESI}.

\section{Conclusions} 

We have demonstrated that the mechanism of simultaneous electron emission, when the time of return of the rescattered electron is equal to the time of liberation of the bounded electron, can be responsible for the anticorrelation of the electrons during \gls{NSDI} in the deep \gls{BIT} regime [see Figure \ref{Ch5_IntensityDependence}(f)]. The \gls{SEE} process significantly differs from \gls{RESI} because it does not require the excitation of the ion to explain the anticorrelation of the electrons observed in the two-electron correlated spectra. \gls{SEE} and \gls{RESI} are by no means mutually exclusive; they both contribute. It will be interesting to study quantitatively the relative contribution of \gls{SEE} and \gls{RESI} to \gls{NSDI} in the deep \gls{BIT} regime.
\glsresetall

\chapter{Shapes of Leading Tunnelling Trajectories for Single-electron Molecular Ionization}\label{chapter6}

If the frequency of the laser field is very low, then it is a good approximation to reduce the original time dependent problem to the time independent one where the laser field is substituted by a static field. In such a picture, ionization processes are realized by quantum tunnelling. From the interpretational point of view, it is advantageous to use the language of quasiclassical trajectories to calculate tunnelling rates. Since the quasiclassical approximation in the original form cannot be readily utilized in the three-dimensional space, the assumption that all relevant quasiclassical trajectories are close to linear is very important, because it reduces the original three-dimensional problem to the effective one-dimensional problem. Therefore, it is not surprising that such an approximation became almost a tacit assumption in strong filed physics, especially in the case of single-electron molecular ionization. In fact, we have already used a very similar assumption in the time dependent cases when we calculated the Coulomb corrections -- The trajectories of electrons in the combined field of the core and the laser have been assumed to be merely \gls{SFA} trajectories perturbatively corrected by the Coulomb field.

How good is the hypothesis of linearity of tunnelling trajectories in strong fields?  To answer this question, we first need to pose it clearly.

\section{Mathematical Background}\label{Ch6_Sec2}

The instanton approach is one of the methods for description of tunnelling \cite{Holstein1996}. It can be introduced as a result of application of the saddle point approximation to the modification of the Feynman integral obtained after performing the transformation of time $t\to -i\tau$ to ``imaginary time'' $\tau$ (i.e., the Wick rotation). This technique has turned out to be tremendously fruitful in many branches of physics and chemistry (see, e.g., References \cite{Vainshtein1982, Leggett1984, Benderskii1994, Razavy2003a}). 

We shall reiterate the main steps in deriving the instanton approach. Let us consider a quantum system with the Hamiltonian  
\begin{eqnarray}\label{Ch6_InstSec_Hamiltonian}
\hat{H} = -\Delta/(2m) + U({\bf x}),
\end{eqnarray}
where $\Delta$ is the $n$-dimensional Laplacian and ${\bf x}$ is an $n$-dimensional vector. The Feynman integral representation of the propagator reads \cite{Feynman1965} 
\begin{eqnarray}
&& \bra{ {\bf x}_f } e^{-i\hat{H}t_0} \ket{{\bf x}_i} = N \int \mathrsfs{D}[{\bf x}(t)] e^{ i S[{\bf x}(t)] }, \label{Ch6_FeynmanRepPropag}\\
&& S[{\bf x}] = \int_0^{t_0} \mathrsfs{L}({\bf x}, \dot{\bf x}) dt, 
\quad \mathrsfs{L}({\bf x}, \dot{\bf x}) = \frac{\dot{\bf x}^2}{2m} - U({\bf x}), \nonumber
\end{eqnarray}
where the path integral sums up all the paths that obey boundary conditions ${\bf x}(0) = {\bf x}_i$ and ${\bf x}(t_0) = {\bf x}_f$, and $\dot{\bf x}(t) \equiv d {\bf x}(t) /dt$. After performing the Wick rotation, Equation (\ref{Ch6_FeynmanRepPropag}) becomes 
\begin{eqnarray}\label{Ch6_EuclidianFeynmanRepProp}
&& \bra{ {\bf x}_f } e^{-\hat{H}\tau_0} \ket{{\bf x}_i} = N \int \mathrsfs{D}[{\bf x}(\tau)] e^{ -\tilde{S}[{\bf x}(\tau)] }, \\
 && \tilde{S}[{\bf x}] = \int_0^{\tau_0} \left[ \frac{1}{2m} \left(\frac{d{\bf x}(\tau)}{d\tau}\right)^2 + U({\bf x}(\tau))  \right]d\tau, \nonumber
\end{eqnarray}
where $\tau_0 = it_0$ and $\tilde{S}$ is called {\it the Euclidian action}.  
Hence, one can say that the transition from Equation (\ref{Ch6_FeynmanRepPropag}) to Equation (\ref{Ch6_EuclidianFeynmanRepProp}) is achieved by the following formal substitutions
\begin{eqnarray}
 t \to -i\tau, \quad {\bf x}(t) \to {\bf x}(\tau), \quad \dot{\bf x}(t) \to id{\bf x}(\tau)/d\tau. \label{Ch6_WickRotationETC}
\end{eqnarray}
Comparing the actions $S$ and $\tilde{S}$, one concludes that the motion in imaginary time is equivalent to the motion in the inverted potential. In other words, the actions $S$ and $\tilde{S}$ are connected by the substitution
\begin{eqnarray}
U \to -U, \qquad (E \to -E). \label{Ch6_FlippingPotentialEnergy}
\end{eqnarray}
The final step in the instanton approach is the application of the saddle point approximation to the Euclidian Feynman integral in Equation (\ref{Ch6_EuclidianFeynmanRepProp}) assuming that $\tau_0 \to \infty$.

However, there is a long ongoing discussion \cite{Eldad1977, VanBaal1986, Dunne2000, Mueller-Kirsten2001} whether the instanton approach agrees with the quasiclassical approximation for tunnelling; some observations have been made that these two methods may disagree up to a pre-exponential factor. Furthermore,  as it has been pointed out in Reference \cite{Leggett1984}, the instanton approach in the formulation presented so far [substitutions (\ref{Ch6_WickRotationETC})] not only looks like a ``highly dubious manoeuvre,'' but also gives no prescription for getting a correct pre-exponential factor. Consequently,  a natural question aries how this method can be safely used and what the meaning of substitutions (\ref{Ch6_WickRotationETC}) and (\ref{Ch6_FlippingPotentialEnergy}) is.

The mathematical physics community has reinterpreted the instanton approach rigorously (see, e.g., References \cite{Carmona1981, Agmon1982, Helffer1988, Sigal1988a, Sigal1988, Hislop1989a, Herbst1993, Hislop1996} and references therein), and the corresponding rigorous analysis answers both questions. Moreover, this rigorous interpretation is extremely useful because it can be implemented as an effective numerical method, which will lead to a clear physical picture applicable to a broad class of problems. We shall review briefly the cited above works since on the one hand, they are unknown for physicists, and on the other hand, they may be challenging to read for non-specialists in mathematical physics. 

Historically, the first problem considered within such a framework was ``how fast does a bound state decay at infinity?'' \cite{Carmona1981, Agmon1982} (see also Section 3 of Reference \cite{Hislop1996}). Let us clearly pose the question. Consider the Hamiltonian (\ref{Ch6_InstSec_Hamiltonian}) as a self-adjoint operator on $\mathrsfs{L}_2(\mathbb{R}^n)$ -- the space of square-integrable functions. A bound state $\psi \in \mathrsfs{L}_2(\mathbb{R}^n)$ is a normalizable eigenfunction of such a Hamiltonian, $\hat{H}\psi = E\psi$. Since the normalization integral converges, the bound state $\psi = \psi({\bf x})$ must vanish as  $\| {\bf x} \| \to \infty$. Therefore, we want to determine how this decay is affected by the potential $V$. This question can be answered very elegantly if we confine ourself to {\it an upper bound} on the rate of decay. 

To obtain this upper bound, we need to introduce first some geometrical notions. Let $M$ be a real $n$-dimensional manifold (intuitively, $M$ is some $n$-dimensional surface). The tangent space at a point ${\bf x}\in M$, denoted by $T_{\bf x}(M)$, is a real linear vector space $\mathbb{R}^n$ that intuitively contains all the possible ``directions'' in which one can tangentially pass through ${\bf x}$. A metric is an assignment of an inner (scalar) product to $T_{\bf x}(M)$ for every ${\bf x} \in M$. 

Let ${\bf x} \in \mathbb{R}^n$ and ${\bm \xi}, {\bm \eta} \in T_{\bf x}(M)$. We define a (degenerate) metric by
\begin{eqnarray}\label{Ch6_AgmonMetricDeff}
\langle {\bm \xi}, {\bm \eta} \rangle_{\bm x} \equiv 2m(U({\bf x}) - E)_+ \langle {\bm \xi}, {\bm \eta} \rangle,
\end{eqnarray}
where $\langle {\bm \xi}, {\bm \eta} \rangle \equiv {\bm \xi}\cdot {\bm \eta} = \xi_1 \eta_1 + \ldots + \xi_n \eta_n$ is the Euclidean inner product and $f({\bf x})_+ \equiv \max\{ f({\bf x}), 0\}$. Following the convention used in mathematical literature, we shall call metric (\ref{Ch6_AgmonMetricDeff}) as {\it the Agmon metric}.

Having introduced the metric, we can equip the manifold $M$  with many geometrical notions such as distance, angle, volume, etc. The length of a differentiable path ${\bm \gamma} : [0,1] \to \mathbb{R}^n$ in the Agmon metric is defined by
\begin{eqnarray}\label{Ch6_AgmonLengthOfCurve}
L_A({\bm \gamma}) = \int_0^1 \| \dot{\bm \gamma}(t) \|_{{\bm \gamma}(t)} dt 
= \sqrt{2m} \int_0^1 [U({\bm \gamma}(t)) - E]_+^{1/2} \| \dot{\bm \gamma}(t) \| dt,
\end{eqnarray}
where $\| {\bm \xi} \| = \sqrt{\langle {\bm \xi}, {\bm \xi} \rangle}$ is the Euclidian (norm) length, and $\| {\bm \xi} \|_{\bf x} = \sqrt{\langle {\bm \xi}, {\bm \xi} \rangle_{\bf x}}$. The path of a minimal length is called a geodesic. Finally,  {\it the Agmon distance between points ${\bf x}, {\bf y} \in \mathbb{R}^n$}, denoted by $\rho_E({\bf x}, {\bf y})$, is the length of the shortest geodesic in the Agmon metric connecting ${\bf x}$ to ${\bf y}$. 

Before going further, we would like to clarify the physical meaning of the Agmon metric. Let us recall {\it the Jacobi theorem} from classical mechanics (see page 150 of Reference \cite{Gustafson2003} and page 247 of Reference \cite{Arnold1989}): The classical trajectories of the system with the potential $U({\bf x})$ and  a total energy $E$ are geodesics in {\it the Jacobi metric}
\begin{eqnarray}\label{Ch6_JacobiMetricDeff}
\langle\langle {\bm \xi}, {\bm \eta} \rangle\rangle_{\bf x} = 2m(E - U({\bf x}))_+ \langle {\bm \xi}, {\bm \eta} \rangle,
\end{eqnarray} 
on the set $\{ {\bf x} \in \mathbb{R}^n | U({\bf x}) \leqslant E\}$ --  the classical allowed region. The Agmon metric [Equation (\ref{Ch6_AgmonMetricDeff})] and the Jacobi metric [Equation (\ref{Ch6_JacobiMetricDeff})] are indeed connected through the substitution (\ref{Ch6_FlippingPotentialEnergy}). By virtue of this analogy, we conclude that the Agmon distance has to satisfy a time-independent Hamilton-Jacobi equation, also known as an eikonal equation,
\begin{eqnarray}\label{Ch6_Hamiljacobi_AgmonDistance}
| \nabla_{\bf x} \rho_E({\bf x}, {\bf y}) |^2 = 2m( U({\bf x}) - E )_+,
\end{eqnarray}
where $\nabla_{\bf x} f({\bf x}) \equiv (\partial f /\partial x_1, \ldots, \partial f/\partial x_n)$. In fact, the Agmon distance is the Euclidean version of the reduced action [here, the adjective ``Euclidian'' means the same as in Equation (\ref{Ch6_EuclidianFeynmanRepProp})]. In other words, the Agmon distance is the action of an instanton. 

Now we are in position to recall upper bounds on a bound eigenstate of the Hamiltonian (\ref{Ch6_InstSec_Hamiltonian}). First, 
under very mild assumptions on $U$ (merely, continuity, compactness of the classically allowed region, and absence of tunnelling, i.e., the spectrum of the Hamiltonian being only real), it has been proven \cite{Agmon1982} that for an arbitrary small $\epsilon > 0$, there exists a constant $0<c_{\epsilon}<\infty$, such that
\begin{eqnarray}\label{Ch6_MildUpperBoundOnBoundState}
\int e^{2(1-\epsilon)\rho_E({\bf x})} |\psi({\bf x})|^2 d^n {\bf x} \leqslant c_{\epsilon},
\end{eqnarray}
where $\rho_E({\bf x}) \equiv \rho_E({\bf x}, {\bf 0})$. Roughly speaking, result (\ref{Ch6_MildUpperBoundOnBoundState}) means that $\psi({\bf x}) = O\left( e^{-(1-\epsilon)\rho_E({\bf x})} \right)$. However, this result can be improved.  For any small $\epsilon > 0$,  there exists a constant $0<c_{\epsilon}<\infty$, such that the following inequality is valid under additional conditions of regularity of the potential $U$
\begin{eqnarray}\label{Ch6_UpperBoundOnBoundState}
| \psi({\bf x}) | \leqslant c_{\epsilon} e^{-(1-\epsilon)\rho_E({\bf x})}.
\end{eqnarray}

Analyzing Equation (\ref{Ch6_MildUpperBoundOnBoundState}) and Equation (\ref{Ch6_UpperBoundOnBoundState}), we conclude that the Agmon distance from the origin describes the exponential factor of the wave function. Further information can be found in References \cite{Agmon1982,  Herbst1993, Hislop1996} and references therein. We note that lower bounds on ground states can also be obtained by utilizing the Agmon approach \cite{Carmona1981}.  

We illustrate the power and utility of upper bound (\ref{Ch6_UpperBoundOnBoundState}) by deriving upper bounds for matrix elements and transition amplitudes in Appendix \ref{Appendix_3}. The former result is an estimate of the modulo square of the matrix element 
$ 
\bra{\psi_p} V \ket{\psi_q}, 
$ 
where $\psi_p$ and $\psi_q$ are bound eigenstates of the Hamiltonian (\ref{Ch6_InstSec_Hamiltonian}) that correspond to eigenvalues $E_p$ and $E_q$. It is demonstrated in Appendix \ref{Appendix_3} that for an arbitrary small $\epsilon > 0$, there exists a constant $0<c_{\epsilon}<\infty$, such that
\begin{eqnarray}\label{Ch6_Inequality_LandauQuasiclassicalMatrixElem}
\left| \bra{\psi_p} V \ket{\psi_q} \right|^2 \leqslant c_{\epsilon}\int V^2 ({\bf x}) e^{-2(1-\epsilon)\left[ \rho_{E_p}({\bf x}) + \rho_{E_q}({\bf x}) \right]} d^n {\bf x},
\end{eqnarray}
which could be interpreted as, 
\begin{eqnarray}\label{Ch6_LandauQuasiclassicalMatrixElemGenral}
 \left| \bra{\psi_p} V \ket{\psi_q} \right|^2   = O\left( \int V^2({\bf x}) e^{-2(1-\epsilon)\left[ \rho_{E_p}({\bf x}) + \rho_{E_q}({\bf x}) \right]} d^n {\bf x} \right).
\end{eqnarray}

Simplicity of the derivation of Equation (\ref{Ch6_LandauQuasiclassicalMatrixElemGenral}) does not imply its insignificance. On the contrary, Equation (\ref{Ch6_LandauQuasiclassicalMatrixElemGenral}) is a multidimensional generalization of the Landau method of calculating quasiclassical matrix elements \cite{Landau1932} (see also page 185 of Reference \cite{Landau_1977} and References \cite{Nikitin1991, Nikitin1993}). To the best of my knowledge, such a generalization has not been reported before. To prove the one-dimensional version of the Landau method using analytical techniques (as it is usually done), one deals with the Stokes phenomenon (see, e.g., Reference \cite{Meyer1989}); thus, the generalization to the multidimensional case without too restrictive assumptions is not obvious. The Agmon upper bounds lead not only to quite a trivial derivation, but also to an intuitive physical and geometrical picture. 

Now we explain briefly how these geometrical ideas are generalized to the problem of tunnelling (an interested reader should consult References \cite{Sigal1988, Sigal1988a, Hislop1989a, Hislop1996} and references therein for details and further developments). Let $E$ be an energy of a tunnelling particle. We denote the boundary of the classically  forbidden region by $S_E$. It is assumed that $S_E$ consists of two disjoint pieces  $S_E^-$ and $S_E^+$ (i.e., $S_E = S_E^- \cup S_E^+$ and $S_E^- \cap S_E^+ = \emptyset$) -- inside and outside turning surfaces, which are merely multidimensional analogs of turning points. Having introduced the concept of the Agmon distance, we naturally introduce two related notions: First, {\it the Agmon distance from the surface $S_E^-$ to a point ${\bf x}$, $\rho_E({\bf x}, S_E^-)$},  as the minimal Agmon distance between the point ${\bf x}$ and an arbitrary point ${\bf y} \in S_E^-$ [more rigorously, $
\rho_E({\bf x}, S_E^-) =  \inf_{{\bf y} \in S_E^-}  \rho_E({\bf x}, {\bf y})$]; second, {\it the Agmon distance between the turning surfaces $S_E^-$ and $S_E^+$, $\rho_E(S_E^-, S_E^+)$,} as the minimal Agmon distance between arbitrary two points ${\bf x} \in S_E^+$ and ${\bf y} \in S_E^-$ [ $\rho_E(S_E^-, S_E^+) = \inf_{{\bf x} \in S_E^+} \rho_E({\bf x}, S_E^-)$]. 

In a nutshell, and thus a bit abusing the formulation of the original result \cite{Hislop1989a}, we say that for an arbitrary small $\epsilon>0$, there exists a constant $c>0$, such that the tunnelling rate, $\Gamma$, (viz., the width of  a resonance) in the quasiclassical limit ($\hbar\to 0$) obeys 
\begin{eqnarray}\label{Ch6_HislopSigal_UpperBound}
\Gamma \leqslant c\exp[ - 2\beta_E (\tilde{\rho}_E - \epsilon) ],
\end{eqnarray}
where $0 < \tilde{\rho}_E < \infty$ and $\beta_E\tilde{\rho}_E$ being the leading asymptotic of $\rho_E(S_E^-, S_E^+)$ when $\hbar\to 0$. However, the following interpretation of upper bound (\ref{Ch6_HislopSigal_UpperBound}) is sufficient for our further applications
\begin{eqnarray}\label{Ch6_AgmonTunnellingRate}
\Gamma = O\left( e^{-2\rho_E(S_E^-, S_E^+) } \right),
\end{eqnarray}
i.e., twice the Agmon distance between the turning surfaces gives the leading exponential factor of the tunnelling rate within the quasiclassical approximation. 

Equation (\ref{Ch6_AgmonTunnellingRate}) is not only of analytical interest, but also is a starting point of an efficient numerical method for estimating tunnelling probabilities. The Agmon distance between two points, $\rho_E({\bf x}, {\bf y})$, can be computed by solving numerically Equation (\ref{Ch6_Hamiljacobi_AgmonDistance}) with the boundary condition 
\begin{eqnarray}
\rho_E( {\bf y}, {\bf y}) = 0
\end{eqnarray}
by means of the fast marching method \cite{Barth1998a, Kimmel1998, Sethian1999, Kimmel2004}. Moreover, having computed the solution, one can readily extract the minimal geodesic from a given initial point  ${\bf x}$ by back propagating along $\rho_E({\bf x}, {\bf y})$, where ${\bf y}$ is regarded as a fixed parameter; more explicitly, the minimal geodesic, ${\bf g} \equiv {\bf g}(t)$, is obtained as the solution of the following Cauchy problem \cite{Kimmel2004, Kimmel1998}
\begin{eqnarray}\label{Ch6_MinGeodesic}
\dot{{\bf g}} = -\nabla_{\bm \xi} \rho_E({\bm \xi}, {\bf y}), \qquad {\bf g}(0) = {\bf x}.
\end{eqnarray}
Such a geodesic can be interpreted as a ``tunnelling trajectory.''

A brief remark on types of the solutions of Equation (\ref{Ch6_Hamiljacobi_AgmonDistance}) ought to be made. Generally speaking, an eikonal equation admits a local solution under reasonable assumptions, but a global solution is not possible in a general case owing to the possibility of development of caustics (see, e.g., Reference \cite{Evans1998}). Nonetheless, when we talk about a solution of Equation (\ref{Ch6_Hamiljacobi_AgmonDistance}), we actually refer to a viscosity solution because not only it is a global solution, but also it has the meaning of distance \cite{Sethian1999, Kimmel2004} which we originally assigned to the function $\rho_E$.

In fact, the fast marching method is an ``upwind'' finite difference method that efficiently computes the viscosity solution of an eikonal equation. Note, hence, that the fast marching method as well as the other ideas presented and developed in the current work cannot be employed to study the influence of chaotic tunnelling trajectories (see Reference \cite{Levkov2009} and references therein). Some implementations of the fast marching method as well as the minimal geodesic tracing can be downloaded from References \cite{ChuLSMLIB, KroonAccurateFM, PeyreToolboxFM}.
  
The Agmon distance from the surface to a point, $\rho_E({\bf x}, S_E^-)$, must satisfy Equation (\ref{Ch6_Hamiljacobi_AgmonDistance}). Indeed, $\rho_E({\bf x}, S_E^-)$ is the solution of the boundary problem
\begin{eqnarray}\label{Ch6_AgmonDistanceFromSurf_HJ}
 |\nabla_{\bf x} \rho_E({\bf x}, S_E^-) |^2 = 2m(U({\bf x}) - E)_+, \qquad
 \rho_{E} ({\bf y}, S_E^-) = 0, \quad \forall {\bf y} \in S_E^-,
\end{eqnarray}
which can be solved by the fast marching method as well. Finally, the Agmon distance between the turning surfaces is computed as $\min_{{\bf x} \in S_E^+} \rho_E({\bf x}, S_E^-)$ after solving Equation (\ref{Ch6_AgmonDistanceFromSurf_HJ}). 

The points ${\bf b} \in S_E^-$ and ${\bf e} \in S_E^+$ such that
\begin{eqnarray}
\rho_E (S_E^-, S_E^+) = \rho_E ({\bf b}, {\bf e}),
\end{eqnarray}
 are of physical importance because they represent the points where the particle ``begins'' its motion under the barrier (${\bf b}$) and ``emerges'' from the barrier (${\bf e}$), correspondingly. Moreover, the minimal geodesic (\ref{Ch6_MinGeodesic}) that connects these points (${\bf g}(0) = {\bf b}$ and ${\bf g}(1) = {\bf e}$) is  a tunnelling trajectory which gives the largest tunnelling rates -- the {\it leading tunnelling trajectory}. Note, however, that these points as well as the trajectories may not be unique in a general case.

It is noteworthy that a power of the fast marching method in applications to tunnelling has already been recognized in chemistry within the context of the reaction path theory \cite{Dey2004, Dey2006a, Dey2006b, Dey2007}. Similarly to the current work, the main object of interest of those studies is the reaction path, which is the leading tunnelling trajectory in our terminology. Nevertheless, the motivation for the usage of the fast marching method, presented in References  \cite{Dey2004, Dey2006a, Dey2006b, Dey2007}, is tremendously  different from our geometrical point of view. 

\section{Main Results}\label{Ch6_Sec3}

In this section, we shall follow a two step program. First, we consider tunnelling in multiple finite range potentials, where we prove that leading tunnelling trajectories are linear (Theorem \ref{Ch6_theorem1}). Then, we reduce the case of multiple long range potentials to the previous one by employing the fact that a singular long range potential can be represented as a sum of a singular short range potential and a continuous long range tail [Equation (\ref{Ch6_PartitionLongRangeSingularPotential})]. Such a reduction allows us to prove that the leading tunnelling trajectories are ``almost'' linear (Theorem \ref{Ch6_theorem2}). We note that partitioning (\ref{Ch6_PartitionLongRangeSingularPotential}) was put forth by \gls{PPT}, and it is widely used for obtaining the Coulomb corrections in strong filed ionization  (see References \cite{Popov2004, Popov2005,  Popruzhenko2008b, Popruzhenko2008a, Smirnova2008a, Popruzhenko2009} and references therein).

Let us introduce some notations. Hereinafter, the dimension of the space is assumed to be $n\geqslant 2$.  The interaction of an electron with a static electric field of the strength $F$ is of the form $Fx_n$ ($F>0$). $\partial A$ denotes the boundary of the region $A$. The map, ${\bm \min_{x_n}} : \mathbb{R}^n \supset A \to \mathbb{R}^n$, selects a point ${\bf x} = {\bm \min_{x_n}} A \in A$ that has the smallest $x_n$ component among all the other points from $A$, assuming that $A$ has such a unique point. The projection $P{\bf x}$ of the point ${\bf x} = (x_1, x_2, \ldots, x_n)$ is defined as $P{\bf x} = (x_1, \ldots, x_{n-1}, E/F)$.

\begin{theorem}\label{Ch6_theorem1}
We study single electron tunnelling ($-\infty< E<0$, $F>0$) in the potential 
\begin{eqnarray}\label{Ch6_TotalPotentialComplectSupport}
U({\bf x}) = \sum_{j=1}^K V_j(\| {\bf x} - {\bf R}_j \|) + Fx_n.
\end{eqnarray}
Let us assume that 
\begin{enumerate}
\item $V_j : (0, R_j) \to (-\infty, 0)$ and $V_j : (R_j, +\infty) \to \{ 0 \}$, $R_j > 0$, $j=1,\ldots,K$, are differentiable on $(0, R_j)$ and strictly increasing functions, such that $V_j(0) = -\infty$ and $V_j$ may have a jump discontinuity at the point $R_j$. 
\item $\supp V_j  = \left\{ {\bf x} \in  \mathbb{R}^n \, | \, V_j (\| {\bf x} - {\bf R}_j \|) \neq 0 \right\}$ is the support of the potential $V_j (\| {\bf x} - {\bf R}_j \|)$, $\supp V_k \cap \supp V_j = \emptyset$, $\forall k \neq j$ and $\supp V_j \cap \left\{  {\bf x} \in  \mathbb{R}^n \, | \,  x_n \leqslant E/F \right\} = \emptyset$, $j = 1,\ldots,K$.
\item Introduce ${\bf q}_j = {\bm \min_{x_n}} \partial \supp V_j$, ${\bf p}_j =  {\bm \min_{x_n}} S_E^-(j)$, $S_E^-(j)$ is defined in Equation (\ref{Ch6_SEminusJ_Deff}). If there exists $N$, such that 
\begin{eqnarray}\label{Ch6_EuclidianAssumption}
\| {\bf p}_N - P{\bf R}_N \| < \| {\bf q}_j - P{\bf R}_j \|, \quad \forall j \neq N,
\end{eqnarray}
\end{enumerate}
Then, the leading tunnelling trajectory is unique and linear, and it starts at the point ${\bf p}_N$ and ends at $P{\bf R}_N$, $\rho_E( S_E^-, S_E^+) = \rho_E({\bf p}_N,  P{\bf R}_N)$.
\end{theorem}
\begin{proof}
 The boundary of the classically forbidden region is defined by the equation $U({\bf x}) = E$. Consider two cases: 
 
 First, if ${\bf x} \notin \bigcup_{j=1}^K \supp V_j $ then according to assumption 2, the above equation simply reads $Fx_n = E$, and thus its solution defines the outer turning surface 
 $$
 S_E^+ = \left\{  {\bf x} \in  \mathbb{R}^n \, | \,  x_n = E/F \right\}.
 $$ 
 One can see now that the projector operator $P$ projects a point onto $S_E^+$. 
 
 Second, if ${\bf x} \in \supp V_j$ and $V_j$ is continuous at the point $R_j$, then the equation reads $V_j (\| {\bf x} - {\bf R}_j \|) + Fx_n = E$. To proof that the set 
 \begin{eqnarray}\label{Ch6_SEminusJ_Deff}
 S_E^-(j) = \left\{ {\bf x} \in \supp V_j \, | \,  V_j (\| {\bf x} - {\bf R}_j \|) + Fx_n = E \right\}
 \end{eqnarray}
 is not empty, we construct the function $f_j( {\bf x} ) = V_j (\| {\bf x} - {\bf R}_j \|) + Fx_n -E$. Since $f_j({\bf R}_j) = -\infty$, we can find a set $A_j \subset \supp V_j$ located close to ${\bf R}_j$, such that $f_j({\bf x}) < 0$ for all ${\bf x} \in A_j$; correspondingly, since according to assumption 2, $x_n > E/F$, there exists the set $B_j \subset \supp V_j$ of points close to the boundary of $\supp V_j$ for which $f_j$ is positive. In fact, $A_j$ and $B_j$ can be constructed such that $\| {\bf x} - {\bf R}_j \| < \| {\bf y} - {\bf R}_j \|$, $\forall {\bf x} \in A_j$ and $\forall {\bf y} \in B_j$. Therefore, the intermediate value theorem guarantees that  $S_E^-(j) \neq \emptyset$ and it ``lies between'' $A_j$ and $B_j$. Furthermore, the inner turning surface is $S_E^- = \bigcup_{j=1}^K S_E^-(j)$, and $S_E^-(j) \cap S_E^-(k) = \emptyset$, $\forall j\neq k$. (Note that the strict monotonicity of $V_j(x)$ assures that the set $S_E^-(j)$ is connected.) Whence, 
\begin{eqnarray}\label{Ch6_ReductionManyCenterCaseToOneCenter}
\rho_E( S_E^-, S_E^+) = \min_j \left\{ \rho_E (S_E^-(j), S_E^+) \right\}.
\end{eqnarray}
Equation (\ref{Ch6_ReductionManyCenterCaseToOneCenter}) means the reduction of the many centre case  to the singe centre case under the assumptions made. Needles to mention that such a reduction tremendously simplifies the analysis. 

The same conclusions are valid if the jump of the function $V_j$ at $R_j$ is not too large, so that the equation $V_j (\| {\bf x} - {\bf R}_j \|) + Fx_n = E$ has solutions for $ {\bf x} \in \supp V_j$. However, if the jump is too large, i.e., this equation does not have solutions from the support of the potential, then it is natural to set $S_E^-(j) = \partial \supp V_j$.

Consider the single centre case -- single electron tunnelling in the potential $U_j ({\bf x}) = V_j (\| {\bf x} - {\bf R}_j \|) + Fx_n$. We shall show that this potential is axially symmetric. If ${\bf x}= (x_1, \ldots, x_n)$, then we introduce $\Pi{\bf x} \equiv (x_1, \ldots, x_{n-1})$. We can then symbolically write ${\bf x} = (\Pi{\bf x} , x_n)$. Using this new notation, we obtain 
\begin{eqnarray}\label{Ch6_Rewritten_Potential_Uj}
U_j ({\bf x}) = V_j \left( \sqrt{\| \Pi{\bf x} - \Pi{\bf R}_j\|^2 + \left( x_n - \left[{\bf R}_j\right]_n \right)^2} \right) + F x_n,
\end{eqnarray}
where $\left[ {\bf a} \right]_n$ denotes the $n^{\rm th}$ component of the vector ${\bf a}$. It is readily seen from Equation (\ref{Ch6_Rewritten_Potential_Uj}) that the potential $U_j ({\bf x})$ is invariant under transformations that do not change $x_n$ and arbitrary $(n-1)$ dimensional (proper and improper) rotations of the vector $\Pi{\bf x}$ about the point $\Pi{\bf R}_j$. The only invariant subspace of $\mathbb{R}^n$ under such transformations is the line $\{ (\Pi{\bf R}_j, x_n) \, | \, x_n \in \mathbb{R}\}$.

Since both regions $S_E^-(j)$ and $S_E^+$ are shape invariant under the axial symmetry transformations, we may expect that the shortest geodesic connecting these regions ought to be shape invariant as well. 
Thus, one readily concludes that the leading tunnelling trajectory should be linear and should connect the points ${\bf p}_j$ and  $P{\bf R}_j$
\begin{eqnarray}\label{Ch6_S_minus_j_equality}
\rho_E ({\bf p}_j, P{\bf R}_j ) = \rho_E (S_E^-(j), S_E^+),
\end{eqnarray}
since no other geodesic that connects $S_E^-(j)$ and $S_E^+$ is shape invariant with respect to the axial symmetry transformations. Below we shall present a formal version of this derivation.

Foremost, we demonstrate that the operation ${\bm \min_{x_n}}$ is defined on the set $S_j^-(j)$, viz., that there is a unique point of $S_j^-(j)$ that has the smallest component $x_n$. Employing the method of Lagrange multipliers and taking into account the symmetry of the potential, we construct the function 
\begin{eqnarray}
 \mathrsfs{L}_1 (x_n, c, \lambda) = x_n + \lambda \left[ V_j\left( \sqrt{ c^2 + \left(x_n - \left[{\bf R}_j\right]_n\right)^2 }\right)  + Fx_n -E\right].
\end{eqnarray}
The condition $\partial \mathrsfs{L}_1 /\partial c = 0$ leads to $c=0$. Therefore, 
${\bf p}_j =  {\bm \min_{x_n}} S_E^-(j) = (\Pi{\bf R}_j, y)$, where $y$ being the minimal solution of the equation 
\begin{eqnarray}\label{Ch6_EquationForY}
V_j\left(\left| y - \left[{\bf R}_j\right]_n\right|\right) + Fy = E. 
\end{eqnarray}
Moreover, $P{\bf p}_j \equiv P{\bf R}_j \equiv P{\bf q}_j$.

Equation (\ref{Ch6_EquationForY}) must have two distinct solutions $y_{1,2}$ ($y_1 < y_2$). $y_1$ ($y_2$) corresponds to the point from $S_E^-(j)$ with the minimum (maximum) $x_n$. Additionally, since $E -F y_1 > E-Fy_2$ $\Rightarrow$ $V_j\left(\left| y_1 - \left[{\bf R}_j\right]_n\right|\right) > V_j\left(\left| y_2 - \left[{\bf R}_j\right]_n\right|\right)$, we obtain 
\begin{eqnarray}\label{Ch6_InequalitiesForY}
\eta_j \equiv \left| y_1 - \left[{\bf R}_j\right]_n\right| > \left| y_2 - \left[{\bf R}_j\right]_n\right|.
\end{eqnarray}

To find the maximum of the function $\| {\bf x} - {\bf R}_j \|$ on the set $S_E^-(j)$ within the Lagrange multipliers method, we introduce the function
\begin{eqnarray}\label{Ch6_LagrangeMultipliers2}  
 \mathrsfs{L}_2 (x_n, c, \lambda) = \sqrt{c^2 + \left(x_n - \left[{\bf R}_j\right]_n\right)^2}  + \lambda \left[ V_j\left( \sqrt{ c^2 + \left(x_n - \left[{\bf R}_j\right]_n\right)^2 }\right)  + Fx_n -E\right]. 
\end{eqnarray}
Taking into account inequality (\ref{Ch6_InequalitiesForY}) and the fact that $\partial \mathrsfs{L}_2 /\partial c = 0$ $\Rightarrow$ $c=0$, we conclude that the maximum of the function $\| {\bf x} - {\bf R}_j \|$ on $S_E^-(j)$ is reached at the point ${\bf p}_j$.

Let $S_j( z )$ denote a sphere of the radius $z$ centred at ${\bf R}_j$, $S_j(z) = \{ {\bf x} \in \mathbb{R}^n \, | \, \| {\bf x} - {\bf R}_j \| = z \}$. Consider a sequence of spheres
$
\left\{ S_j\left(\eta_j + k[R_j-\eta_j]/W \right) \right\}_{k=0}^W, 
$
where $S_j (R_j) = \partial \supp V_j$ and $\eta_j$ was introduced in Equation (\ref{Ch6_InequalitiesForY}). Now pick a sequence of points, $\{ {\bm \gamma}( k/W ) \}_{k=0}^W$, such that, ${\bm \gamma}( k/W ) \in  S_j\left(\eta_j + k[R_j-\eta_j]/W \right)$, $k=0,
\ldots,W$. We assume that this sequence is a discretization of some  differentiable path ${\bm \gamma} : [0,1] \to \mathbb{R}^n$. According to Equation (\ref{Ch6_AgmonLengthOfCurve}), the sums,
\begin{eqnarray}\label{Ch6_SigmaWDeff}
\Sigma_W({\bm \gamma}) &=& \sqrt{2m} \sum_{k=0}^W \sqrt{ U_j( {\bm \gamma}(k/W)) - E} \,\, \| {\bm \gamma}( [k+1]/W ) - {\bm \gamma}(k/W) \|,
\end{eqnarray}
where we set ${\bm \gamma}( 1+1/W ) \equiv {\bm \gamma}( 1 )$, obeys the property $\lim_{W\to\infty} \Sigma_W ({\bm \gamma}) = L_A({\bm \gamma})$. Introduce a path:
\begin{eqnarray}\label{Ch6_PathGDeff}
{\bf g}(t) = {\bf p}_j+ t\left[ {\bf q}_j - {\bf p}_j \right].
\end{eqnarray}
Since $\forall k$, ${\bf g}(k/W) \in  S_j\left(\eta_j + k[R_j-\eta_j]/W \right)$, $ \left[{\bm \gamma}(k/W)\right]_n  \geqslant \left[{\bf g}(k/W)\right]_n$ and
$V_j (\| {\bf g}(k/W) - {\bf R}_j \| ) = V_j (\| {\bm \gamma}(k/W) - {\bf R}_j \|) $ $\Rightarrow$ $U_j ({\bm \gamma}(k/W)) \geqslant U_j({\bf g}(k/W))$. Moreover, $ \| {\bm \gamma}( [k+1]/W ) - {\bm \gamma}(k/W) \| \geqslant  \| {\bf g}( [k+1]/W ) - {\bf g}(k/W) \|$. Therefore, 
\begin{eqnarray}
\Sigma_W({\bm \gamma}) \geqslant \Sigma_W({\bf g}) \Rightarrow 
L_A ({\bm \gamma}) \geqslant L_A({\bf g}).
\end{eqnarray}
Since $\Sigma_W({\bm \gamma}) = \Sigma_W({\bf g}) \Leftrightarrow {\bm \gamma}(k/W) = {\bf g}(k/W)$, $k = 0,\ldots,W-1$, $\forall W$, we conclude that path (\ref{Ch6_PathGDeff}) is indeed the shortest geodesic that connects the regions $S_E^-(j)$ and $\partial\supp V_j$. By the same token, the geodesic connecting $\partial\supp V_j$ and $S_E^+$ must be a straight line that starts at ${\bf q}_j$ and ends at $P{\bf q}_j$ because the potential between these two regions is merely $V({\bf x}) =  Fx_n$.

To finalize the proof, we shall {\it ``backward propagate''} the leading tunnelling trajectory starting from the outer turning surface $S_E^+$. Let $\tilde{\rho}({\bf x}, {\bf y})$ denote the Agmon distance between two points for the potential $V({\bf x}) =  Fx_n$. Then, it is easy to demonstrate that 
\begin{eqnarray}\label{Ch6_AgmonDIstanceInConstantField}
\tilde{\rho}_E({\bf x}, P{\bf x}) = (2/3)\sqrt{2mF} \| {\bf x} - P{\bf x} \|^{3/2}.
\end{eqnarray}
The plane $T(c) = \{ {\bf x} \in \mathbb{R}^n \, | \, x_n = c \}$ is a surface of constant Agmon distance [Equation (\ref{Ch6_AgmonDIstanceInConstantField})], such that $\tilde{\rho}_E(T_{E/F}, S_E^+) = 0$ and $\tilde{\rho}_E(T(c), S_E^+)$ is a strictly increasing function of $c$. Since $\| {\bf q}_N - P{\bf R}_N \| = \| {\bf p}_N - P{\bf R}_N \| - \| {\bf p}_N - {\bf q}_N \| < \| {\bf q}_j - P{\bf R}_j \|$,  $\forall j\neq N$, condition (\ref{Ch6_EuclidianAssumption}) guarantees that increasing $c$ the plane $T(c)$ will ``hit'' the boundary of $\supp V_N$ at the point ${\bf q}_N$. (Note that $\tilde{\rho}_E\left(T(c), S_E^+\right) \equiv \rho_E\left(T(c), S_E^+\right)$, $E/F < c  < \left[ {\bf q}_N \right]_n$.) Moreover,  the following follows from Equation (\ref{Ch6_EuclidianAssumption})
$$
\{ {\bf x} \in \mathbb{R}^n \, | \,\left[ {\bf q}_N \right]_n \leqslant x_n \leqslant \left[ {\bf p}_N \right]_n \}\cap\supp V_j = \emptyset, \quad \forall j \neq N,
$$
which means that the $N^{\rm th}$ centre is isolated from all the other. Therefore, the shortest geodesic must connect the point ${\bf q}_N$ to the point ${\bf p}_N$.
\end{proof}

\begin{corollary}\label{Ch6_corollary1}
Consider a single electron tunnelling in the potential (\ref{Ch6_TotalPotentialComplectSupport}), such that assumptions 1 and 2 of Theorem \ref{Ch6_theorem1} are satisfied, then the leading trajectory is linear (but may not be unique).
\end{corollary}
\begin{proof} 
This corollary follows from the straightforward generalization of the idea of backward propagation. 
\end{proof}

\begin{theorem}\label{Ch6_theorem2}
We shall study single electron tunnelling ($-\infty< E<0$, $F>0$) in the potential 
\begin{eqnarray}\label{Ch6_TotalPotentialLongRangeCase}
U({\bf x}) = \sum_{j=1}^K \mathrsfs{V}_j(\| {\bf x} - {\bf R}_j \|) + Fx_n.
\end{eqnarray}
Assume that
\begin{enumerate}
\item $\mathrsfs{V}_j : (0, +\infty) \to (-\infty, 0)$ are differentiable on $(0, +\infty)$ and strictly increasing functions, such that $\mathrsfs{V}_j(0)=-\infty$ and $\mathrsfs{V}_j(+\infty)=0$.
\item The boundary of the classically forbidden region consists of two disjoin pieces -- the internal turning surface $S_E^-$ and  the outer one $S_E^+$. Furthermore, $S_E^- = \bigcup_{j=1}^K S_E^-(j)$, $S_E^-(j) \cap S_E^-(k) = \emptyset$, $\forall j\neq k$, where each $S_E^-(j)$ encircles ${\bf R}_j$\footnote{
The verb ``encircle'' should be understood in the following sense. A piece of the inner turning surface, $S_E^-(j)=\partial CA(j)$, is a boundary of the classically allowed region, $CA(j)$, associated with centre $j$, such that ${\bf R}_j \in CA(j)$.
}. 
\item $B(j) \cap B(k) = \emptyset$, $\forall j\neq k$, and $B(j) \cap S_E^+ = \emptyset$, $\forall j$, where $B(j) = \left\{ {\bf x} \in  \mathbb{R}^n \, |\,  \| {\bf x} - {\bf R}_j \| \leqslant r_j \right\}$ being the ball of radius $r_j$ centered at ${\bf R}_j$. Here $r_j = \max\left\{ \| {\bf x} - {\bf R}_j \| \, | \, {\bf x} \in S_E^-(j) \right\}$ is the ``radius'' of $S_E^-(j)$\footnote{
The parameter $r_j$ can be calculated by means of the method of Lagrange multipliers as it was shown in the proof of Theorem \ref{Ch6_theorem1} [see Equation (\ref{Ch6_LagrangeMultipliers2})].
}.
\end{enumerate}
Then, the leading tunnelling trajectory (may not be unique) is linear up to a term of $O(\lambda)$ as $\lambda\to 0$, where $\lambda = \max_j \left\{ |\mathrsfs{V}_j(\Delta_j) |\right\}$ and 
\begin{eqnarray}\label{Ch6_DeltaDeff}
\Delta_j = \min\left( \frac{r_j}2 + \frac 12 \min_{k, \, k \neq j} \left\{ \| {\bf R}_j - {\bf R}_k \| - r_k \right\}, d_j \right).
\end{eqnarray}
Here, $d_j = \min\left\{ \| {\bf x} - {\bf R}_j \| \, | \, {\bf x} \in S_E^+ \right\}$ is the Euclidean distance from ${\bf R}_j$ to $S_E^+$.
\end{theorem}

\begin{proof}
We introduce two auxiliary functions
\begin{eqnarray}
V_{sh}^{(j)}(x) = \left\{ 
	\begin{array}{lll}
	\mathrsfs{V}_j(x) & : & 0 \leqslant x < \Delta_j, \\
	0 & : & x \geqslant \Delta_j,
	\end{array} 
\right. &\quad& 
V_{lg}^{(j)}(x) = \left\{ 
	\begin{array}{lll}
	0 & : & 0 \leqslant x < \Delta_j, \\
	\mathrsfs{V}_j(x)  & : & x \geqslant \Delta_j.
	\end{array} 
\right. \nonumber
\end{eqnarray}
One evidently notices that
\begin{eqnarray}\label{Ch6_PartitionLongRangeSingularPotential}
\mathrsfs{V}_j(x) = V_{lg}^{(j)}(x) + V_{sh}^{(j)}(x),
\end{eqnarray}
where $V_{sh}^{(j)}(x)$ is a singular short range potential and $V_{lg}^{(j)}(x)$ being a long range tail. The purpose of such a partition is to make $V_{sh}^{(j)}(x)$ satisfy assumption 1 of Theorem \ref{Ch6_theorem1} and produce $V_{lg}^{(j)}(x)$ that obeys the following upper bound:
$$ 
|V_{lg}^{(j)}(x)| \leqslant  |\mathrsfs{V}_j(\Delta_j) | \leqslant \lambda, \qquad \forall x.
$$ 

We analyze the length of a curve in the Agmon metric [Equation (\ref{Ch6_AgmonLengthOfCurve})]. Since 
\begin{eqnarray}
\sqrt{ U({\bf x}) - E } = \sqrt{ \sum_{j=1}^K V_{sh}^{(j)}( \| {\bf x} - {\bf R}_j \| ) + Fx_n -E + O(\lambda)} \nonumber\\
= \sqrt{ \sum_{j=1}^K V_{sh}^{(j)}( \| {\bf x} - {\bf R}_j \| ) + Fx_n -E} + O(\lambda), \nonumber
\end{eqnarray}
under the assumption that $\lambda\to 0$, we have reduced the initial situation to the case of single electron tunnelling in the potential 
\begin{eqnarray}
U_{sh}({\bf x}) = \sum_{j=1}^K V_{sh}^{(j)}(\| {\bf x} - {\bf R}_j \|) + Fx_n.
\end{eqnarray}

Let us now utilize assumption 3 to show that
\begin{eqnarray}\label{Ch6_DeltaJBiggerRJ}
\Delta_j > r_j.
\end{eqnarray}
Indeed, on the one hand, $B(j)\cap S_E^+ = \emptyset$ $\Rightarrow$ $d_j > r_j$; on the other hand, $B(j) \cap B(k) = \emptyset$, $\forall j\neq k$, $\Rightarrow$ $\| {\bf R}_j - {\bf R}_k \| - r_k > r_j$. 

Furthermore, we shall demonstrate that the definition of  $\Delta_j$ [Equation (\ref{Ch6_DeltaDeff})] assures that assumption 2 of  Theorem \ref{Ch6_theorem1} for the functions $V_{sh}^{(j)}(x)$ holds. According to Equation (\ref{Ch6_DeltaDeff}), 
$$
\Delta_j  \leqslant \left(  \| {\bf R}_j - {\bf R}_k \| - r_k + r_j \right)/2, \qquad 
j \neq k;
$$
hence, $\Delta_j + \Delta_k \leqslant \| {\bf R}_j - {\bf R}_k \|$ $\Rightarrow$ $\supp V_{sh}^{(j)} \cap \supp V_{sh}^{(k)} = \emptyset$. From Equation (\ref{Ch6_DeltaDeff}), we also obtain $\Delta_j \leqslant d_j$ $\Rightarrow$ $\supp V_{sh}^{(j)} \cap S_E^+ = \emptyset$; thus, the outer turning surface for the potential $U_{sh}({\bf x})$ should be $\left\{  {\bf x} \in  \mathbb{R}^n \, | \,  x_n = E/F \right\}$.

Finally, we have proven the theorem because the potential $U_{sh}({\bf x})$ satisfies all the assumptions of Corollary \ref{Ch6_corollary1}. 
\end{proof}

Physical clarifications of Theorems \ref{Ch6_theorem1} and \ref{Ch6_theorem2} are due. Assumption 1 of Theorem \ref{Ch6_theorem1} physically implies that $V_j$ are attractive, singular, spherically symmetric short range potentials.  Assumption 2 of the same theorem requires that the potentials do not merge, i.e., their ranges do not overlap. This condition connotes that the classically allowed regions associated with the centres ${\bf R}_j$ [their boundaries are $S_E^-(j)$] do not overlap as well. The latter statement is proven in Theorem \ref{Ch6_theorem1}. The statement of Corollary \ref{Ch6_corollary1} can be rephrased as follows: leading tunnelling trajectories for a system of non-overlapping, attractive, singular, short range potentials are linear. However, if the additional condition (\ref{Ch6_EuclidianAssumption}) is satisfied then Theorem \ref{Ch6_theorem1} not only guarantees the uniqueness of the leading tunnelling trajectory, but also provides the initial and final points of the trajectory. Assumption 1 of Theorem \ref{Ch6_theorem2}  means that $\mathrsfs{V}_j$ are attractive, singular, spherically symmetric long range  potentials that vanish at infinity.  Assumptions 2 and 3 of Theorem \ref{Ch6_theorem2} require the same non-overlapping condition for the classically allowed internal regions mentioned above. Physically, Theorem \ref{Ch6_theorem2} says that leading tunnelling trajectories for a system of several such potentials are ``almost'' linear, and a deviation from being strictly linear is caused by vanishing long tails of the potentials; thus, the larger the distance between the centres, the smaller the deviation. 

\section{The Application of Spherically Symmetric Potential Wells to Single-Electron Molecular Tunnelling}\label{Ch6_Sec4}

The simplest type of model molecular potentials that allows for full analytical treatment is of type (\ref{Ch6_TotalPotentialComplectSupport}) where
\begin{eqnarray}\label{Ch6_PotentialWell}
V_j(x) = \left\{
	\begin{array}{ccc}
		C_j & : & 0 < x < r_j, \\
		0 & : & x > r_j,
	\end{array}
\right.
\end{eqnarray}
such that $C_j < E$, $\forall j$, and it is assumed that $S_E^-(j) = \partial \supp V_j = \{ {\bf x} \in  \mathbb{R}^n \, | \, \| {\bf x} \| = r_j \}$.  (These potentials are not governed by Theorem \ref{Ch6_theorem1}.) Evidently, the leading tunnelling trajectories are linear, and moreover, the following equality is valid 
\begin{eqnarray}\label{Ch6_AgmonDistanceForWells}
\rho_E ( S_E^-, S_E^+) = \min_j \left\{ \tilde{\rho}_E \left({\bf q}_j, P{\bf q}_j \right) \right\},
\end{eqnarray}
where ${\bf q}_j = {\bm \min}_{x_n} S_E^-(j)$ and $\tilde{\rho}_E$ was defined in Equation (\ref{Ch6_AgmonDIstanceInConstantField}). Let us estimate the tunnelling rates within Equation (\ref{Ch6_AgmonTunnellingRate}) for the two dimensional system of two equivalent centres of type (\ref{Ch6_PotentialWell}) (see Figure \ref{Ch6_FigTwoCentres}).  A straightforward geometrical derivation, using Equations (\ref{Ch6_AgmonTunnellingRate}), (\ref{Ch6_AgmonDIstanceInConstantField}), and (\ref{Ch6_AgmonDistanceForWells}), shows that
\begin{eqnarray}\label{Ch6_TwoEqualPotentialWellsTunnelRates}
\Gamma \propto \exp\left\{ -\frac 2{3F} \left[ FR(1-| \cos\theta |) - 2E \right] ^{3/2} \right\},
\end{eqnarray} 
where $R$ is the distance between the potential wells (i.e., the bond length of a model molecule) and $\theta$ is the angle between the field and the molecular axis. The obtained angular dependent rates are plotted in Figure \ref{Ch6_FigRatesTwoCentres}. 

According to Equation (\ref{Ch6_AgmonTunnellingRate}), rates obtained within the geometrical approach does not account for an initial molecular orbital. This technique provides solely the contribution  from the shape of the barrier, hence, the name -- the ``geometrical approach.'' An advantage of such a method is that it reduces the calculation of tunnelling rates to a rather simple geometrical exercise.

\begin{figure}
\begin{center}
\includegraphics[scale=0.40]{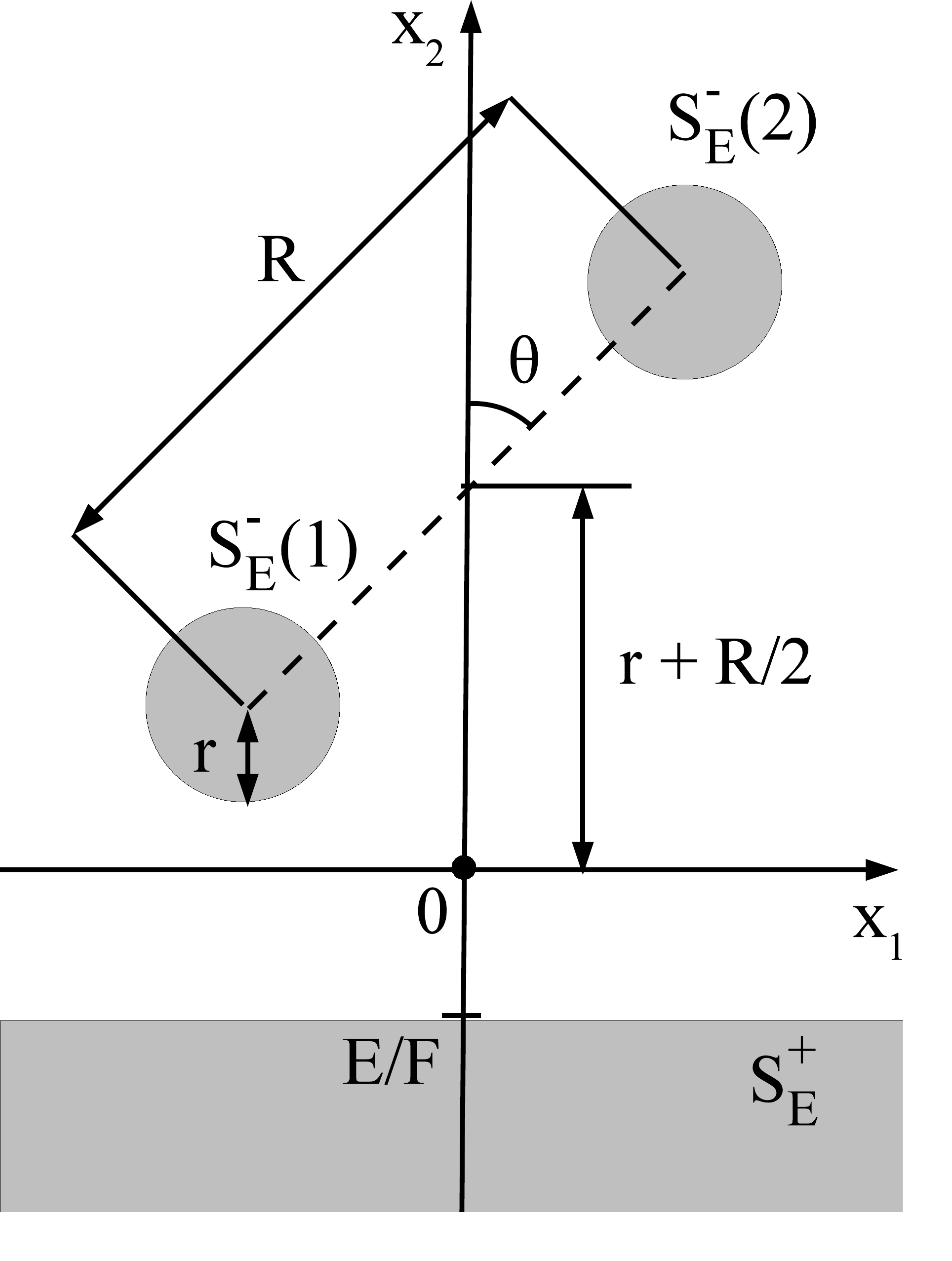}
\caption[The geometry of a two-centre model]{The geometry of a two-centre model employed to obtain Equation (\ref{Ch6_TwoEqualPotentialWellsTunnelRates}). Grey colour denotes the classically allowed region.}\label{Ch6_FigTwoCentres}
\end{center}
\end{figure}

\begin{figure}
\begin{center}
\includegraphics[scale=0.60]{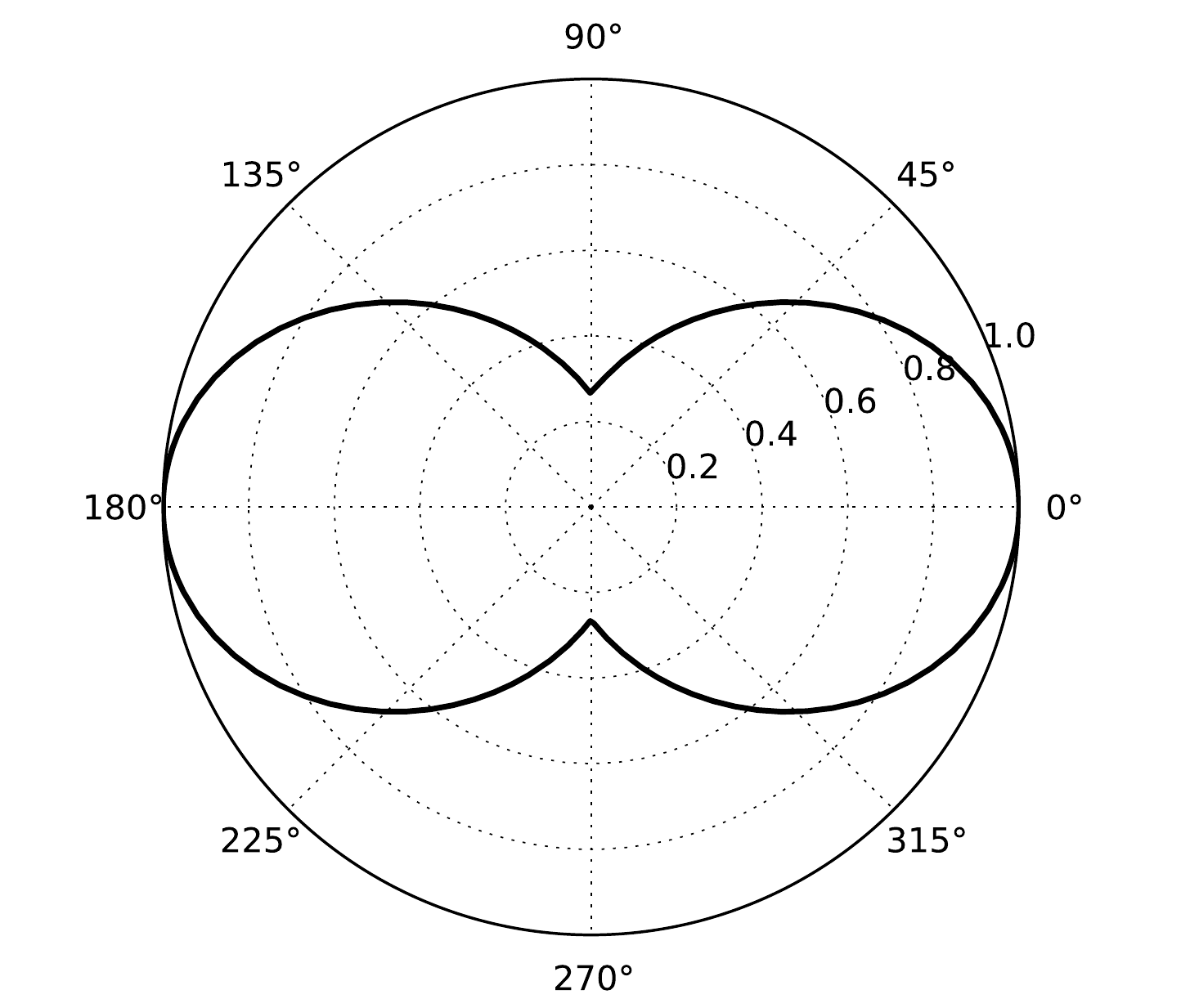}
\caption[The normalized tunnelling rates for the two-centre model]{The polar plot of the normalized tunnelling rates for the two centre model [Equation (\ref{Ch6_TwoEqualPotentialWellsTunnelRates})] as a function of the angle $\theta$. Chosen values of the parameters are $F=0.01$ (a.u.), $R=2$ (a.u.), and $E=-0.5$ (a.u.).}\label{Ch6_FigRatesTwoCentres}
\end{center}
\end{figure}

\section{Conclusions}

Having introduced the leading tunnelling trajectory as an instanton path that gives the highest tunnelling probability, we have proven that leading tunnelling trajectories for multi centre short range potentials are linear (Theorem \ref{Ch6_theorem1}) and ``almost'' linear for multi centre long range potentials (Theorem \ref{Ch6_theorem2}). In a nutshell, these results have been achieved because the multi centre (i.e., molecular) potential is represented as a sum of spherically symmetric potentials, and such conclusions regarding the shape of the trajectories in the single centre (i.e., atomic) case are quite expectable owing to the axial symmetry. An important peculiarity of the  theorems is that assumptions that they involve are satisfied in majority situations of current experimental interest. Nevertheless, the proven statements by no means exhaust all interesting cases; on the contrary, we have barley scratched the surface, and there is plenty of room for further generalizations and expansions. For example, we have not discussed the case when internal classically allowed regions associated with the centres merge.  One may anticipate that the leading tunnelling trajectories still should be linear under some additional assumptions (by the argument of the backward propagation of the leading tunnelling trajectory). The issue of the uniqueness of the trajectories was left nearly undiscussed. 

In any case, one can always employ the fast marching method, discussed in Section \ref{Ch6_Sec2}, to obtain numerically some information on the leading tunnelling trajectories.

The fact that the leading trajectories for long range potentials are not straight lines is of vital importance. As in the atomic case \cite{Perelomov_1966, Perelomov_1967_A, Perelomov_1967_B, Perelomov_1968, Popov2004, Popov2005, Popruzhenko2008a, Popruzhenko2008b, Popruzhenko2009}, this deviation is crucial for a quantitative treatment \cite{Tong2002, Fabrikant2009, Bin2010, Fabrikant2010, Murray2010}, and sometimes even for a qualitative analysis, because it leads to the correct pre-exponential factor of ionization rates that describes the influence of the Coulomb field of nuclei. However, Theorem \ref{Ch6_theorem2} suggests that the deviations can be accounted for by means of the perturbation theory where the zero order approximation being a field free trajectory. This is a part of the celebrated \gls{PPT} approach, widely employed in the literature for analytical calculations of atomic Coulomb corrections. Nevertheless, the \gls{PPT} method requires matching the quasiclassical wave function of an electron in the continuum with the bound (field free) atomic wave function. This step is a stumbling block for generalization of the \gls{PPT} approach to the molecular case (for the suggestion of a solution to such a problem see Reference \cite{Murray2010}). Theorem \ref{Ch6_theorem2} in fact offers a solution to the problem of matching. According to Theorem \ref{Ch6_theorem2}, matching should be done on spherical surfaces of radii $\Delta_j$ [Equation (\ref{Ch6_DeltaDeff})] centred at the nuclei. This is an alternative technique to the method developed in Reference \cite{Murray2010}.

It has been suggested in Reference \cite{Meckel2008a} that molecular photoionization in the tunnelling limit may act as a scanning tunnelling microscope (STM). Since rotating a molecule with respect to a field direction is analogous to moving the tip of an STM, then the observed angular-dependent ionization probability should provide information for a molecule similar to the position dependence of the tunnelling current in the STM. We point out that there is a resemblance between such a descriptive comparison and our results. As it has been shown in Theorem \ref{Ch6_theorem1} (by the backward propagation of the leading trajectory), the leading tunnelling trajectory starts at the atomic centre that is the closest to the barrier exit  (i.e., the outer turning surface); hence, the qualitative similarity of molecular tunnelling with the STM. 

\chapter{Outlook}\label{chapter7}

The adiabatic approximation is a convenient, yet intuitive, tool for studying processes induced by low frequency laser radiation. In the current work we have studied analytically  few electron processes such as single-electron ionization and \gls{NSDI}.  

A future direction, motivated by constantly expanding experimental data, could be a general study of many-electron processes. The complexity of the problem decreases effectiveness and perhaps even usefulness of any analytical method;   {\it ab initio} solutions of the Schr\"{o}dinger equation are becoming more and more prominent. Nevertheless, bearing a positive attitude towards analytical approaches, we shall discuss two possible directions of further applications of the adiabatic approximation. 

First, the adiabatic approximation allows one to reduce the original $N$-body problem to an effective, or quasi, single body problem within the \gls{SFA}. We shall demonstrate such a reduction on the example of the problem of simultaneous $N$ electron ionization \cite{Zon_1999}.  At the beginning all the $N$ electrons are bounded, i.e.,
\begin{equation}\label{Ch7_EiNeIoniz}
E_i^{(Ne)} = -\sum_{i=1}^N I_p^{(i)} \equiv -N\bar{I_p}.
\end{equation}
At the end, they are liberated and move under the influence of a linearly polarized laser field 
\begin{eqnarray}\label{Ch7_EfNIonz}
E_f^{(Ne)} (\varphi) = \frac 12 \sum_{i=1}^N \left[\K_i + \A(\varphi)\right]^2.
\end{eqnarray}
First and foremost, Equation (\ref{Ch7_EfNIonz}) is to be simplified. To accomplish that, let us use the following simple transformation. For a given set of numbers $x_1, x_2, \ldots, x_N$, the standard deviation $\sigma_x$ is defined by
$$
\sigma_x^2 = \frac 1N \sum_{i=1}^N (x_i - \bar{x})^2 = \frac 1N \left( \sum_{i=1}^N x_i^2 - N\bar{x}^2 \right),
$$
where $\bar{x} = \frac 1N \sum_{i=1}^N x_i$ denotes the mean; thus,
\begin{equation}\label{Ch7_Deviation}
\sum_{i=1}^N  x_i^2 = N\left(\sigma_x^2 +\bar{x}^2\right).  
\end{equation}
Applying Equation (\ref{Ch7_Deviation}) to Equation (\ref{Ch7_EfNIonz}), we obtain 
\begin{equation}\label{Ch7_EfN_mean_deviation}
E_f^{(Ne)}(\varphi) = \frac N2 \left\{ \left[ \bar{\kk} + A(\varphi) \right]^2 + \bar{\k}^2 + \sigma_{\kk}^2 + \sigma_{\k}^2 \right\}.
\end{equation}
Note that the initial (\ref{Ch7_EiNeIoniz}) and final (\ref{Ch7_EfN_mean_deviation}) energies resemble the initial and final energies for single-electron ionization given by Equation (\ref{Ch3_Ei_Ef}). Let us name the transition from Equation (\ref{Ch7_EfNIonz}) to Equation (\ref{Ch7_EfN_mean_deviation}) {\it the mean-deviation parametrization.} Having introduced the action for simultaneous $N$ electron ionization
$$
S_{Ne} = \frac 1{\omega} \int^{\varphi_0} \left[ E_f^{(Ne)} (\varphi) - E_i^{(Ne)} \right]d\varphi, 
$$
we can formulate the reduction of such an $N$ electron process to the single-electron ionization as
\begin{eqnarray}
S_{Ne} = N S \left( \bar{\kk}, \bar{\k}, \bar{I_p} + \frac 12 \left[ \sigma^2_{\kk} + \sigma^2_{\k}\right]\right),
\end{eqnarray}
where $S$ is the single-electron action defined in Equation (\ref{Ch3_SingleElectronIonizAction}). Therefore, the ionization rates (\ref{Ch3_GeneralRate}) are generalized to the case of simultaneous $N$ electron ionization.

Furthermore, the mean-deviation parametrization can be performed even in the case of a ``continuous'' ensemble of electrons,
\begin{eqnarray}
E_f^{(Ne)}(\varphi) = \frac 12 \int d^3 \P\, \rho(\P) [\P + \A(\varphi)]^2 =  \frac 12[\overline{\pp} + A(\varphi)]^2 + \frac 12 \left( \overline{\p}^2 + \sigma_{\parallel}^2 + \sigma_{\perp}^2 \right),
\end{eqnarray}
where $\sigma_{\parallel, \perp}$ are the standard deviations of the momenta and $\overline{p_{\parallel, \perp}}$ are the mean momenta given by
$$
\sigma_{\parallel, \perp}^2 = \overline{p_{\parallel, \perp}^2} - \overline{p}_{\parallel, \perp}^2, \qquad
\overline{p_{\parallel, \perp}^n} = \int d^3\P \, \rho(\P) p_{\parallel, \perp}^n, \quad (n=1,2)
$$
and $ \rho(\P)$ is the momentum distribution function of the ensemble.

The mean-deviation reduction can be evidently used in many-electron non-sequential (correlated) processes and other processes in external fields. In fact, we have already applied this trick to the two-electron \gls{SEE} process studied in Chapter \ref{chapter5} [see the interpretation of Equation (\ref{Ch5_SingleElectronIonizationAndSEE})].

However, the role of Coulomb interactions remains the most challenging question in many-electron phenomena. The presented mean-deviation parametrization can be most naturally introduced only in the \gls{SFA}, and currently it is not obvious how to account for the Coulomb interactions within this formalism. Perhaps, one can develop the \rm{SF-EVA}-like approach to the $N$ electron problem;  an alternative path could be to use  the kinetic theory \cite{Liboff1998}. 

The mean-deviation parametrization, besides being just a computational trick, can be used to visualize and process experimental data obtained from many-particle coincidence measurements (recently quadruple coincidence measurements have became feasible \cite{Yamazaki2009}). According to the mean-deviation parametrization, the averages and standard deviations of momenta contain as much information as all the components of the momenta; hence, we have a reduction of a many-dimensional data set, which cannot be plotted directly, to merely a two-dimensional plot. 

As far as molecular single-electron ionization is concerned, the combination of the Hislop-Sigal geometrical approach to tunnelling and the fast marching method (both have been presented in Section \ref{Ch6_Sec2}) forms a powerful toolkit for studying strong filed phenomena in polyatomic molecules within the quasiclassical approximation. Furthermore, the demonstrated simplicity of the shapes of leading tunnelling
trajectories may encourage future developments of analytical
models of molecular ionization. However, we note that the geometrical approach has
a fundamental limitation -- it does not account for effects of molecular
orbitals, and there is no a priori way of including these effects. In
spite of that, one may always attempt to introduce such corrections in a
heuristic manner, e.g., multiplying the geometrical rates by a Dyson orbital.   

In Chapter \ref{chapter6}, we modelled a molecule by a single-electron multi centre potential, hence discarding effects of electron-electron correlations. Nevertheless, the geometrical approach to tunnelling reviewed in Section \ref{Ch6_Sec2} can account for these effects after an appropriate adaptation that is presented in Reference \cite{Sigal1988}. Intuitively speaking, according to such a method, the leading tunnelling trajectory of the system is selected such that the minimum number of electrons are displaced during tunnelling. More importantly, the fast marching method can be also utilized to obtain this leading tunnelling trajectory. Since correlation dynamics of electrons plays an important role in molecular ionization leading to interesting novel effects \cite{Walters2010}, applications of the geometrical ideas developed in Reference \cite{Sigal1988} to molecular ionization should be the aim of subsequent studies.

\appendix
\chapter*{APPENDICES}
\addcontentsline{toc}{chapter}{APPENDICES}
\chapter{The Phase of Ionization of the Recolliding Electron as a Function of the Phase of Recollision}\label{Appendix_1}

Employing the \gls{SFA}, we write the formula corresponding to the digram of \gls{NSDI} (Figure \ref{Ch4_Fig4}) 
\begin{eqnarray}\label{Ch4_SFA_S-matrix}
\ket{\Psi(t)} &\sim& \int_{t_i}^t dt_b \int_{t_b}^t dt_r \int d^3 \K \, \U(t, t_r) \frac 1{r_{12}} \ket{\K g^+}\bra{g^+ \K} \hat{V}_L(t_b) \ket{gg} \times \nonumber\\
&&\exp\left\{ -\frac i2 \int_{t_b}^{t_r} [\K + \A(\tau)]^2 d\tau +i|E_{g^+}|(t_r - t_b) + i|E_{gg}|(t_b-t_i) \right\},
\end{eqnarray}
where $\U(t, t_r)$ is the evolution operator of the studied system, $r_{12}$ is the distance between the electrons, $\hat{V}_L(t_b)$ is the interaction between the ionized electron and the laser field, and $E_{g^+}$ and $E_{gg}$ are energies of the states $\ket{g^+}$ and $\ket{gg}$, respectively.

We use the saddle point approximation to calculate the
integrals over $\K$ and $t_b$ in Equation (\ref{Ch4_SFA_S-matrix}).  The phase
of the integral over $\K$ has the following form:
$$
S_1(\K) = -\frac 12 \int_{t_b}^{t_r}  \left[ \K + \A(\tau) \right]^2 d\tau.
$$
The saddle point of this integral is given by
\begin{equation}\label{Ch4_Ksaddle}
\K^* = \frac{-1}{t_r - t_b} \int_{t_b}^{t_r} \A(\tau)d\tau,
\end{equation}
with the restriction $t_r \neq t_b$.  Note that generally
speaking, $\K^*$ can be complex since $t_b$,  as will be
clarified below, is complex for $\gamma \neq 0$.
 The phase of  the integral over $t_b$ in Equation (\ref{Ch4_SFA_S-matrix}) reads
$$
S_2 (t_b) = -\frac 12 \int_{t_b}^{t_r} [\K^* + \A(\tau)]^2 d\tau  +|E_{g^+}|(t_r - t_b) + |E_{gg}|(t_b-t_i).
$$
Hence, the saddle point $t_b(t_r)$ is a function of $t_r$ and given as a solution of the following equation
\begin{equation}\label{Ch4_EqTbofTr}
\frac{\cos(\omega t_r) - \cos[\omega t_b(t_r)]}{\omega[t_r-t_b(t_r)]} + \sin[\omega t_b(t_r) ] = \pm i\gamma,
\end{equation}
where $\gamma$ is the Keldysh parameter for the first electron,
\begin{equation}\label{Ch4_KeldyshParameter}
\gamma = \frac{\omega\sqrt{2(E_{g^+}-E_{gg})}}F = \frac{\omega}F \sqrt{2I_p^{(2)}}.
\end{equation}
Methods of computing the saddle points [Equations (\ref{Ch4_EqTbofTr}) and
(\ref{Ch4_DeltaE})] have been  widely discussed in the literature (see,
for example, References \cite{deMorissonFaria2004, deMorissonFaria2004B, deMorissonFaria2005} and references therein). We shall use a general and simple approach for identifying
correct saddle points between different solutions of the saddle-point equations in the complex plane (see Sec. \ref{Ch4_s3}).

It is convenient to introduce the following phases: $\phi_b =
\omega t_b$ and $\varphi_r = \omega t_r$. According to Equation
(\ref{Ch4_EqTbofTr}),  the saddle point $\phi_b$ is a complex
double-valued function of $\varphi_r$ which can be given by
$\phi_b(+\gamma; \varphi_r)$ and $\phi_b(-\gamma; \varphi_r)$,
where $\gamma$ is the Keldysh parameter (\ref{Ch4_KeldyshParameter}).
Here, the complex single-valued function $\phi_b(\gamma;
\varphi_r)$ is defined as a solution of the following
transcendental equation:
\begin{equation}\label{Ch4_EqPhiBpfPhiR}
\frac{\cos \varphi_r - \cos \phi_b (\gamma; \varphi_r) }{\varphi_r - \phi_b(\gamma; \varphi_r)} + \sin \phi_b(\gamma; \varphi_r) = i\gamma.
\end{equation}
No analytical solution of such an equation is available.
The special case of the function $\phi_b(\gamma; \varphi_r)$ for $\gamma=0$,
\begin{equation}\label{Ch4_EqPhiBpfPhiRG0}
\frac{\cos \varphi_r - \cos \varphi_b(\varphi_r) }{\varphi_r - \varphi_b(\varphi_r)} + \sin \varphi_b(\varphi_r) = 0,
\end{equation}
is very important because of the following two reasons.

First, the function $\varphi_b(\varphi_r)$ is a real valued
function for real $\varphi_r$ (and single-valued for any complex
$\varphi_r$); this allows one to interpret the motion of the first
electron in terms of classical trajectories. The
function $\varphi_b (\varphi_r)$ is defined on the interval $(\pi/2,
2\pi]$ because only during that interval can the free electron 
recollide with its parent ion. The function
$\varphi_b(\varphi_r)$ is bounded in the interval
$0\leqslant \varphi_b (\varphi_r)<\pi/2$. Second, the function $\varphi_b(\varphi_r)$ can be physically
understood as a tunneling limit ($\gamma \ll 1$) of
$\phi_b(\gamma; \varphi_r)$  (i.e., the low-frequency limit).

In terms of the laser phase, the
vector potential $\A(\varphi)$ is
$$
\A(\varphi) = - (\mathbf{F}/\omega) \sin\varphi.
$$

The difference between Equation (\ref{Ch4_EqPhiBpfPhiR}) and Equation
(\ref{Ch4_EqPhiBpfPhiRG0}),  which connect the phase of recollision
$\varphi_r$ with the phase of ionization $\varphi_b(\varphi_r)$ or
$\phi_b(\gamma; \varphi_r)$, is also important for the last step
 of \gls{NSDI} -- the release of the  two electrons following the recollision at $\varphi_r$.

Calculating the integral over $t_r$ in Equation (\ref{Ch4_SFA_S-matrix}) by the saddle point approximation, we need to obtain the transition point
$\varphi_r^0$ for negligible $\gamma$. It is the solution of the  saddle-point equation
\begin{equation}\label{Ch4_DeltaE}
\Delta E (\varphi_r^0) \equiv
\frac 12 \left[\K_1 + \A(\varphi_r^0)\right]^2 +  \frac 12 \left[\K_2 + \A(\varphi_r^0)\right]^2 - \frac 12 \left[\A(\varphi_r^0) - \A(\varphi_b( \varphi_r^0))\right]^2+I_p^{(2)}=0,
\end{equation}
such that
$$
\pi/2 < \Re \varphi_r^0 \leqslant 2\pi,
$$
where $I_p^{(2)} = |E_{g^+}|$ is the ionization potential of the second electron. For $\gamma \neq 0$, the equation is
\begin{equation}\label{Ch4_DeltaEGamma}
\Delta E (\gamma; \varphi_r^0) \equiv
\frac 12 \left[\K_1 + \A(\varphi_r^0)\right]^2 +
\frac 12 \left[\K_2 + \A(\varphi_r^0)\right]^2 - \frac 12 \left[\A(\varphi_r^0) - \A(\phi_b(\gamma; \varphi_r^0))\right]^2+I_p^{(2)}=0,
\end{equation}
where $\phi_b(\gamma; \varphi_r^0)$ now depends on $\gamma$. Note that Equations
(\ref{Ch4_DeltaE}) and (\ref{Ch4_DeltaEGamma}) are basically Equation
(\ref{Ch2_TransitionPoint_Summary}) written in slightly different notations.

If the solution of Equation (\ref{Ch4_DeltaE}) on the interval $(\pi/2,
2\pi]$ is real, then direct collisional ionization is possible.
However, we are interested in the deep quantum regime when the
following inequality is valid for the second electron:
$$
I_p^{(2)} > 3.17 U_p.
$$
By introducing the Keldysh parameter for the second electron
$
\gamma_2 = \sqrt{ I_p^{(2)}/(2U_p) },
$
we can write the last inequality as
\begin{equation}\label{Ch4_ConditionCMI}
\gamma_2 > 1.26.
\end{equation}

Equation (\ref{Ch4_ConditionCMI}) physically means that the returning
electron does not have enough energy to free the second electron.

When recollision energy is not sufficient for collisional
ionization, transition requires help from the laser field.
Mathematically, the arising integral is similar to those in the
adiabatic approximation (see Chapter \ref{chapter2}). The energy gap $\Delta
E(\varphi_r)$ in Equations (\ref{Ch4_DeltaE}) and (\ref{Ch4_DeltaEGamma})
plays the role of the transition energy for the non-adiabatic
transition [the term $E_i(t)-E_f(t)$ in Equation (\ref{Ch2_PropNonAdiabat_Summary})].
The peculiarity of $\Delta E(\gamma; \varphi_r)$ given by Equation
(\ref{Ch4_DeltaEGamma}) is that it need not be  real even for real
$\varphi_r$, since in the term $\left[
A(\varphi_r)-A(\phi_b(\gamma; \varphi_r))\right]^2$ the phase
$\phi_b$ is complex. This subtle aspect underscores the important
difference between using the solutions of Equation (\ref{Ch4_EqPhiBpfPhiR})
or Equation (\ref{Ch4_EqPhiBpfPhiRG0}) for the phase of ionization
$\phi_b$. For classical trajectories, where $\varphi_b(\varphi_r)$
[Equation (\ref{Ch4_EqPhiBpfPhiRG0})] is real for real $\varphi_r$, $\Delta
E(\varphi_r)$ [Equation (\ref{Ch4_DeltaE})] is also real for real
$\varphi_r$. This is the standard assumption for the adiabatic approximation. When the complex phase of ionization $\phi_b(\gamma;
\varphi_r)$ [Equation (\ref{Ch4_EqPhiBpfPhiR})] is used, i.e., when
``quantum'' trajectories for recollision are used, $\Delta
E(\gamma; \varphi_r)$ [Equation (\ref{Ch4_DeltaEGamma})]  need not be real
for real $\varphi_r$.  In References \cite{Moyer_2001}, the
Dykhne method (i.e., the adiabatic approximation for a system with a discreet spectra) has been generalized for this case, provided
that the complex function $\Delta E(\gamma; \varphi)$ satisfies
the Schwarz reflection principle (recently, this result was
confirmed and further generalized in Reference \cite{Schilling_2006}).

The function $\Delta E(\varphi)$ (\ref{Ch4_DeltaE}) obeys the Schwarz reflection
principle, i.e., $ \Delta E^*(\varphi^*) = \Delta E(\varphi).$
Hence, we can conclude that if $\varphi_r^0$ is the solution that
lies in the lower half-plane,  then $\left(\varphi_r^0\right)^*$ is
the solution that lies in the upper half-plane. Furthermore, it
can be easily proven that the following equation takes place for
any function $\Delta E(\varphi)$ which satisfies the Schwarz
reflection principle and any complex number $\varphi_r^0$:
$$
\Im\int_{\Re\varphi_r^0}^{\left(\varphi_r^0\right)^*} \Delta E(\varphi)d\varphi = -\Im\int_{\Re\varphi_r^0}^{\varphi_r^0}\Delta E(\varphi) d\varphi.
$$
From the previous equation, we can see that the transition points
that lie in the lower half-plane lead  to exponentially large
probabilities, which are unphysical.
Hereafter, let $\varphi_r^0$ denote  the  solution of the equation
$\Delta E(\varphi_r^0)=0$, which is the closest to the real axis
and lies in the upper-half plane.

Before continuing our discussion, let us point out the
following simple equalities, which follow from Equation
(\ref{Ch4_EqPhiBpfPhiR}):  $\Re[ \phi_b(+\gamma;
\varphi_r)]=\Re[\phi_b(-\gamma; \varphi_r)]$ and $\Im[
\phi_b(+\gamma; \varphi_r)]=-\Im[\phi_b(-\gamma; \varphi_r)]$ for
real $\varphi_r$. Furthermore, we obtain
\begin{equation}
 \label{Ch4_GenSchwartPrin}
\phi_b^*(-\gamma; \varphi_r^*) = \phi_b(\gamma; \varphi_r), \quad
\Im\int_{\Re\varphi_r^0}^{\left(\varphi_r^0\right)^*} \Delta
E(\gamma;\varphi)d\varphi =
-\Im\int_{\Re\varphi_r^0}^{\varphi_r^0}\Delta E(-\gamma;\varphi)
d\varphi,
\end{equation}
where $E(\gamma;\varphi)$ is defined in Equation (\ref{Ch4_DeltaEGamma}).

Bearing in mind that formula (\ref{Ch4_SFA_S-matrix}) must give an
exponentially small result (which implies that the transition
point must be located in the upper-half plane) and taking into
account Equation (\ref{Ch4_GenSchwartPrin}), we define
 the phase of ionization in the case of $\gamma\neq 0$ as
\begin{equation}\label{Ch4_QuantPhaseBirth}
\Phi (\gamma; \varphi_r) = \left\{
\begin{array}{ll}
\phi_b(-\gamma; \varphi_r) & \mbox{if} \quad \Im\left( \varphi_r \right)>0, \\
\Re\left[\phi_b(\gamma;\varphi_r)\right]  &
\mbox{if}\quad \Im\left(\varphi_r\right)=0,\\
\phi_b(+\gamma; \varphi_r)  & \mbox{if} \quad \Im\left( \varphi_r \right)<0.
\end{array}
\right.
\end{equation}

Equation (\ref{Ch4_QuantPhaseBirth}) is the most consistent definition of
the {\it quantum-mechanical} phase of ionization
 of the first electron
as a function of the phase of return. Generally speaking, there
was an ambiguity in selecting the value of $\Phi(\gamma;
\varphi_r)$ for real $\varphi_r$. However, we have chosen it in
such a way due to the following reason. The function
$\Im\left[\Phi(\gamma; \varphi_r)\right]$ has a jump discontinuity
on the real axis, but the function $\Re\left[\Phi(\gamma;
\varphi_r)\right]$ has a removable discontinuity that can be
removed by employing the equality
$$
\Re\left[\phi_b(\gamma;\varphi_r)\right] \equiv \frac 12 \left[\phi_b(\gamma; \varphi_r+i0) + \phi_b(-\gamma; \varphi_r-i0)\right]
\quad (\mbox{for real } \varphi_r ).
$$
Furthermore, the function $\Phi(\gamma; \varphi_r)$ obeys the
Schwarz reflection principle [$\Phi^*(\gamma; \varphi_r^*) =\Phi(\gamma; \varphi_r) $], and the following equality takes place:
$$
\Phi(0; \varphi_r) = \varphi_b( \varphi_r).
$$
It is essential that according to definition
(\ref{Ch4_QuantPhaseBirth}), the function $\Phi(\gamma; \varphi_r)$ is
real-valued on the real axis and thus allows the identical
interpretation in terms of the classical trajectories as for
$\varphi_b( \varphi_r)$. Therefore, the
definition of $\varphi_r^0$ and inequality (\ref{Ch4_ConditionCMI})
are unchanged in the case of $\gamma\neq  0$.

\chapter{Coulomb Corrections in \glsentryname{NSDI} within the Strong-Field Eikonal-Volkov Approach}\label{Appendix_2}

The key step in dealing with the singularities of the Coulomb
potentials during the recollision is to partition the
electron-electron and electron-ion interactions in the
two-electron Hamiltonian as follows:
\begin{eqnarray}\label{Ch4_PotentialsIntro}
V_{ee} &\equiv& V_{ee} - V_{ee,lng} + V_{ee,lng} = V_{ee,lng} + \Delta V_{ee,shr}, \nonumber\\
V_{en} &\equiv& V_{en} - V_{en,lng} + V_{en,lng} = V_{en,lng} + \Delta V_{en, shr}.
\end{eqnarray}
The potential $V_{ee,lng}$ has a long-range behavior identical to $V_{ee}$, but no
singularity at the origin, and $\Delta V_{ee, shr}$ is singular but short-range potential.
The same applies to $V_{en,lng}$ and $\Delta V_{en,shr}$. We choose
\begin{equation}\label{Ch4_Potentials}
\Delta V_{en,shr}(r) = V_{en}(r)\exp(-r/r_0), \quad \Delta
V_{ee,shr}(r_{12}) = V_{ee}(r_{12})\exp\left[ -
r_{12}/r_{12}^{(0)}\right],
\end{equation}
where $r_0$ and $r_{12}^{(0)}$ will be defined later.  Note that
the partitioning (\ref{Ch4_PotentialsIntro}) and (\ref{Ch4_Potentials}) has
been employed originally in the \gls{PPT} approach
 for the problem of single-electron ionization.

Now, we can write the Hamiltonian as
$$
\h(t) = \h_s(t) + \Delta V_{shr},
$$
where $\Delta V_{shr} \equiv \Delta V_{ee,shr} + \Delta V_{en, shr}$ and $\h_s$ is the rest,
which includes smoothed Coulomb potentials for
electron-electron and electron-nuclear interactions, $V_{ee,lng}$ and $V_{en,lng}$.

To first order in $\Delta V_{shr}$, the amplitude to find two electrons with momenta $\K_1, \K_2$
at the detector at the time $t$ is
\begin{eqnarray}\label{Ch4_Ampl}
a(\K_1, \K_2) = -i\int_{t_i}^t dt_r \int d^3 \K \bra{ \K_1 \K_2} \U_s (t, t_r)\Delta V_{shr} \ket{g^+ \K}
\bra{\K g^+}\U_s (t_r, t_i) \ket{gg}.
\end{eqnarray}
Approximations in Equation (\ref{Ch4_Ampl}) are first order in $\Delta V_{shr}$ and the assumption that
at the moment of recollision the ion is in its ground state. Both are well justified.

The next step is to approximate the two parts of the evolution:
before $t_r$ and after $t_r$. The key component for correlated
spectra is the second part -- after $t_r$. The main aspect of the
first part of the evolution -- prior to $t_r$ -- is to supply
an active electron with the required energy.

To simplify the amplitude $\bra{ \K_1 \K_2} \U_s (t, t_r)\Delta V_{shr} \ket{g^+ \K}$,
we insert the decomposition of unity,
\begin{eqnarray}\label{Ch4_AmplB}
 b(\K_1, \K_2, \K, t_r) &=& \bra{ \K_1 \K_2} \U_s (t, t_r)\Delta V_{shr} \ket{g^+ \K}  \nonumber\\
&=& \int\int d^3 \R_1 d^3 \R_2 \bra{\K_1 \K_2} \U_s(t,t_r) \ket{\R_1 \R_2}\bra{\R_1 \R_2} \Delta V_{shr} \ket{g^+ \K} \nonumber\\
&\approx&  \int\int d^3 \R_1 d^3 \R_2  \bra{\K_1 + \A(t_r) ,\, \K_2 + \A(t_r)} \R_1 \R_2 \rangle \bra{\R_1 \R_2} \Delta V_{shr} \ket{g^+ \K}
\nonumber\\
&& \times\exp\left[ -i\int_{t_r}^t \left\{ \frac 12 [\K_1 + \A(\tau)]^2 + \frac 12 [\K_2 + \A(\tau)]^2 +\right.\right. \nonumber\\
&& \qquad +V_{ee,lng}(\R_{12}(\tau))+ V_{en,lng}(\R_{1}(\tau))+V_{en,lng}(\R_{2}(\tau)) \Big\}d\tau\Big].
\end{eqnarray}
Here we have applied the \gls{SF-EVA} method \cite{Smirnova_2008}.  The
integral from the nonsingular parts of the electron-electron and
electron-ion interactions are calculated along the trajectories in
the laser field only. The trajectories
\begin{equation}\label{Ch4_Trajectories}
\R_{1,2}(t) = \R_{1,2} + \int_{t_r}^{t} [ \K_{1,2} + \A(\tau)] d\tau
\end{equation}
and $\R_{12}(t) = \R_1(t)-\R_2(t)$ begin at the positions $\R_1,
\R_2$ at instant $t_r$.  The bra-vectors $\bra{\K_{1,2} +
\A(t_r)}$ are plane waves. Their distortion by the
electron-electron and electron-core interactions appears in the
($\R_1, \R_2$)-dependent exponential phase factors in Equation
(\ref{Ch4_AmplB}).

Since $\Delta V_{shr}$ is a short-range potential and $\ket{g^+}$
is limited within a  characteristic ionic radius, the term
$\bra{\R_1 \R_2}\Delta V_{shr}\ket{g^+ \K}$  allows us to fix the
initial values of $\R_1$ and $\R_2$. The characteristic radius for
the partitioning of the Coulomb potential into the short-range and
long-range parts is set as  $r_0=r_{12}^{(0)} = 1/\left|
E_{g^+}\right|$, where $E_{g^+}$ denotes the energy level of the second (bound)
electron. Therefore, we pull the exponential factor out of
the integral in Equation (\ref{Ch4_AmplB}) with $r_0=r_{12}^{(0)} =
1/\left| E_{g^+}\right|$ and $r_1 = r_2 =0$,
\begin{eqnarray}\label{Ch4_AmplB2}
&& b(\K_1, \K_2, \K, t_r) \approx \bra{\K_1 + \A(t_r) \, \K_2 + \A(t_r)} \Delta V_{shr} \ket{g^+ \K}\exp\left[ -i\int_{t_r}^t \left\{ \frac 12 [\K_1 + \A(\tau)]^2 +\right.\right. \nonumber\\
&& \qquad  \left.\left. +\frac 12 [\K_2 + \A(\tau)]^2 +V_{ee,lng}(\R_{12}(\tau))+ V_{en,lng}(\R_{1}(\tau))+V_{en,lng}(\R_{2}(\tau)) \right\}d\tau\right].
\end{eqnarray}
Effects of the long-range tails of $V_{ee}$ and $V_{en}$ appear in
the exponent  while the collisional transition is govered by the
short-range interaction \cite{Smirnova_2008}. The states
$\ket{\K}$, so far, represent any convenient basis set of
continuum sates in the laser field.

To simplify the amplitude
$$
c(\K, t_r) = \bra{\K g^+} \U_s(t_r, t_i)\ket{gg},
$$
we note that the second electron is bound during the whole
evolution,  and hence we can simplify $c(\K, t_r)$ using single
active electron approximation. In this approximation, $\U_s(t_r,
t_i)$ describes one-electron dynamics in the self-consistent
potential of the ionic core,
$$
V_{sc} (\R_1) = \bra{g^+} V_{ee,lng}(\R_{12})+ V_{en,lng}(\R_{1})+V_{en,lng}(\R_{2})\ket{g^+}.
$$
The effective Hamiltonian for evolution between $t_i$ and $t_r$ is
$$
\h_{sc} (\R_1, t) = \hat{K}_1 + V_{sc}(\R_1) + V_L (\R_1, t),
$$
where $\hat{K}_1$ is the kinetic energy operator and $V_L(\R_1,
t)$ is  the interaction with the laser field. Now the amplitude
$c(\K, t_r)$ becomes
\begin{equation}\label{Ch4_Ampl2}
c(\K, t_r) = -i \int_{t_i}^{t_r} dt_b \bra{\K} \U_{sc} (t_r, t_b) V_L(\R_1, t_b) \ket{g_D}\exp\left[i\left|E_{gg}\right| (t_b-t_i)\right],
\end{equation}
where $\ket{g_D} =\langle g^+_2 \ket{gg}$ is proportional to the Dyson orbital between the ground states of the neutral and ion and $E_{gg}$ is the energy of the ground state $\ket{gg}$.

The ionic potential contributes to the propagator in Equation (\ref{Ch4_Ampl2}) twice: when the electron leaves
the atom near $t_b$ and when it returns to the ionic core near $t_r$. The contribution ``on the way out''
introduces standard Coulomb correction
\cite{Perelomov_1966, Perelomov_1967_A, Perelomov_1967_B, Perelomov_1968, Popruzhenko2008b, Popruzhenko2008a}
to the ionization amplitude and hence affects the overall height of the final two-electron distribution.
The contribution of $V_{sc}$  ``on the way in'' affects the spatial structure of the recolliding wave packet.
As shown in References \cite{Smirnova_2007, Smirnova_2008}, for short collision times the Coulomb-laser coupling
is small and $V_{sc}$ ``on the way in'' can be included in the adiabatic approximation,
\begin{eqnarray}\label{Ch4_Ampl3}
c(\K, t_r)\ket{\K_{ev} g^+} \approx -i R_C \int_{t_i}^{t_r} dt_b \ket{\K_{ev} g^+}\bra{\K + \A(t_b) -\A(t_r)} V_L(t_b) \ket{g_D}\times \nonumber\\
 \exp\left[ -\frac i2 \int_{t_b}^{t_r} [\K + \A(\tau)-\A(t_r)]^2d\tau +i\left|E_{g^+}\right|(t_r - t_b) + i\left| E_{gg} \right|(t_b - t_i)\right]
\end{eqnarray}
Here $\ket{\K_{ev}}$ is the field-free continuum wave function in
the  eikonal approximation, which includes distortions of the
incoming plane wave with asymptotic momentum $\K$, $\bra{ \K +
\A(t_b) - \A(t_r)}$ is a plane wave, and $R_C$  is the Coulomb
correction to the ionization amplitude which compensates for
approximating $\bra{ \K + \A(t_b) - \A(t_r)}$ with a plane wave in
the matrix element $\bra{\K + \A(t_b) -\A(t_r)} V_L(t_b)
\ket{g_D}$.

Now, putting together Equations (\ref{Ch4_Ampl3}) and (\ref{Ch4_AmplB2}) and
changing the integration variable $\K \to \K + \A(t_r)$, we arrive
at
\begin{eqnarray}\label{Ch4_GenGamma}
 a(\K_1, \K_2) &\approx& -\int_{t_i}^t dt_r \int_{t_i}^{t_r} dt_b \int d^3 \K \int d^3 \R_1 d^3 \R_2 \,   \nonumber\\
&& \times\exp\left[ -\frac i2 \int_{t_b}^{t_r} [\K + \A(\tau)]^2 d\tau
+i\left|E_{g^+}\right|(t_r - t_b) + i\left| E_{gg} \right|(t_b - t_i)-\right.\nonumber\\
&& -i\int_{t_r}^t \left\{ \frac 12 [\K_1 + \A(\tau)]^2 + \frac 12 [\K_2 + \A(\tau)]^2 \right\}d\tau \nonumber\\
&& -i\int_{t_r}^t \left\{ V_{ee,lng}(\R_{12}(\tau))+ V_{en,lng}(\R_1( \tau))+ V_{en,lng}(\R_2( \tau)) \right\}d\tau \Bigg] \nonumber\\
&&\times \bra{\K_1 + \A(t_r), \, \K_2 + \A(t_r)} \R_1 \R_2 \rangle \nonumber\\
&&\times\bra{\R_1 \R_2}
\Delta V_{shr} \ket{g^+,\, \K_{ev}+\A(t_r)}R_C\bra{\K+\A(t_b)} V_L(t_b) \ket{g_D}.
\end{eqnarray}
Note that if the Coulomb corrections $V_{ee,lng}$ and $V_{en,lng}$ are ignored in the exponent of Equation (\ref{Ch4_GenGamma}), then Equation (\ref{Ch4_GenGamma}) coincides with Equation (\ref{Ch4_SFA_S-matrix})  within exponential accuracy.

\chapter{Upper Bounds for Matrix Elements and Transition Amplitudes}\label{Appendix_3}

We derive a multi-dimensional generalization of the Landau method of calculating quasiclassical matrix elements [Equation (\ref{Ch6_Appendix_A_LandauQuasiClassMatrElem})], and we also estimate perturbation theory transition amplitudes [Equations (\ref{Ch6_AppendixA_UpperBoundA2}) and (\ref{Ch6_AppendixA_UpperBoundA1})] in terms of the Agmon distance.
 
For the sake of simplicity, the argument ${\bf x}$ will be omitted in some equations below. Throughout this Appendix, we assume that the Agmon upper bounds \cite{Agmon1982} for bound states ($\psi_n$) are valid, i.e., $\forall \epsilon > 0$ $\exists c_n \equiv c_n(\epsilon)$, $0 < c_n < \infty$, such that 
\begin{eqnarray}\label{Ch6_AppendixA_UppBoundAssump}
 |\psi_n | \leqslant c_n e^{-(1-\epsilon) \rho_n},
\end{eqnarray}
where $\rho_n = \rho_{E_n}$.

Let us choose an arbitrary $\epsilon > 0$. Employing the Schwartz inequality and assumption (\ref{Ch6_AppendixA_UppBoundAssump}), we obtain
\begin{eqnarray}\label{Ch6_AppendixA_Ineq}
\left| \int \psi_p^* V \psi_q d {\bf x} \right|^2  
&=&  \left| \int e^{(1-\epsilon)(\rho_p + \rho_q)} \psi_p^* \psi_q e^{-(1-\epsilon)(\rho_p + \rho_q)} V d {\bf x}\right|^2  \nonumber\\
&\leqslant& B_{p, q}^2 \int \left| e^{(1-\epsilon)(\rho_p + \rho_q)} \psi_p^* \psi_q \right|^2  d {\bf x} \nonumber\\
&\leqslant& B_{p, q}^2 c_p^2(\epsilon') c_q^2(\epsilon')\int e^{-2(\epsilon -\epsilon')(\rho_p + \rho_q)} d{\bf x},
\end{eqnarray}
where $\epsilon > \epsilon' > 0$ and
\begin{eqnarray}
B_{p, q}^l = \int |V|^l e^{-l(1-\epsilon)(\rho_p + \rho_q)} d {\bf x}.
\end{eqnarray}
The integral $\int \exp[-2(\epsilon -\epsilon')(\rho_p + \rho_q)] d{\bf x}$ converges for all $p$ and $q$. Therefore, we have proven that $\forall \epsilon >0$ $\exists c = c(\epsilon)$, $0 < c < \infty$, such that
\begin{eqnarray}\label{Ch6_Appendix_A_LandauQuasiClassMatrElem}
\left| \int \psi_p^* V \psi_q d {\bf x} \right|^2 \leqslant c B_{p, q}^2,
\end{eqnarray}
which is the same as Equation (\ref{Ch6_Inequality_LandauQuasiclassicalMatrixElem}).

Now let us study the problem of estimating of transition amplitudes defined by means of the time dependent perturbation theory. Hereinafter, we assume that a quantum system under scrutiny has no continuum spectrum, and we shall manipulate with all the series and integrals heuristically --  assuming that they all converge, or alternatively, assuming that they are over finite range. We illustrate our idea by estimating the second order amplitude since generalization to higher orders will be evident. 

The second order transition amplitude within the time dependent perturbation theory reads
\begin{eqnarray}
A^{(2)} &=& - \int_{t_i}^{t_f} dt \int^{t_f}_t dt' \int d{\bf x} d{\bf x}' \psi_{fin}^*({\bf x}') e^{-iE_{fin}(t_f-t')}  \nonumber\\
&& \times V({\bf x}') K({\bf x}' t'|{\bf x} t)V({\bf x})\psi_{in}({\bf x})e^{-iE_{in}(t-t_i)},
\end{eqnarray}
where all the $\psi$'s are eigenstates of the system and $K$ is the propagator, which can be written as
\begin{eqnarray}\nonumber
K({\bf x}' t'|{\bf x} t) = \sum_n \psi_n({\bf x}') \psi_n^* ({\bf x}) e^{-iE_n(t' - t)};
\end{eqnarray}
whence, 
$$
|K({\bf x}' t'|{\bf x} t)| \leqslant \sum_n |\psi_n({\bf x}') \psi_n ({\bf x}) |.
$$
Using such a simple estimate as well as inequality (\ref{Ch6_AppendixA_UppBoundAssump}), we obtain
\begin{eqnarray}\label{Ch6_AppendixA_UpperBoundA2}
\frac{\left| A^{(2)}\right|}{ (t_f - t_i)^2 } &\leqslant& \frac{c_{in}c_{fin}}2 \sum_n c_n^2 B_{fin, n}^1 B_{n, in}^1
\leqslant M \sum_n B_{fin, n}^1 B_{n, in}^1,
\end{eqnarray}
where $M \equiv c_{in}c_{fin} \max_n \left\{ c_n^2 \right\} /2$, $0 < M < \infty$.

However, there is no need to confine ourself to the case when the initial and final states are eigenstates. The same idea applies to the general case of the initial ($\phi_{in}$) and final ($\phi_{fin}$) states being represented as linear expansions in the basis of the bound eigenstates, 
\begin{eqnarray}
\phi_{in} = \sum_n \langle \psi_n \ket{\phi_{in}} \psi_n, \quad
\phi_{fin} = \sum_n \langle \psi_n \ket{\phi_{fin}} \psi_n.
\end{eqnarray}
Let us found an upper bound for the first order transition amplitude, which is as follows
\begin{eqnarray}
A^{(1)} &=& -i\int_{t_i}^{t_f} dt \int d{\bf x} \sum_{n, n'} \langle \phi_{fin} \ket{\psi_n} \psi_n^*({\bf x}) e^{-iE_n(t_f-t)} \nonumber\\
&& \times V({\bf x}) \langle \psi_{n'} \ket{\phi_{in}} \psi_{n'}({\bf x}) e^{-iE_{n'}(t-t_i)}.
\end{eqnarray}
Whence, we readily obtain
\begin{eqnarray}\label{Ch6_AppendixA_UpperBoundA1}
\frac{\left|A^{(1)}\right|}{t_f-t_i} &\leqslant& \sum_{n, n'} c_n c_{n'} |\langle \phi_{fin} \ket{\psi_n} \langle \psi_{n'} \ket{\phi_{in}}| B_{n, n'}^1 
\leqslant M \sum_{n,n'} B_{n, n'}^1,
\end{eqnarray}
where $M \equiv \max_{n, n'} \left\{ c_n c_{n'} |\langle \phi_{fin} \ket{\psi_n} \langle \psi_{n'} \ket{\phi_{in}}| \right\}$, $0 < M < \infty$.



\chapter*{Permission's Page}
\addcontentsline{toc}{chapter}{Permission's Page}

Much of the authorÕs research related to this thesis has previously appeared in publications:
\begin{itemize}
\item {\bf Chapter \ref{chapter2}:} D. I. Bondar, W.-K. Liu, and G. L. Yudin. Adaptation of the modified adiabatic approximation to strong-field ionization.
{\it Phys. Rev. A}, 79:065401, 2009 
\ifthenelse{\boolean{ElectronicVersion}}
	{[Preprint \href{http://arxiv.org/abs/0906.1284}{arXiv:0906.1284}].}
	{[Preprint arXiv:0906.1284].} 
	Copyright (2009) by the American Physical Society.
\item {\bf Chapter \ref{chapter3}:} D. I. Bondar. Instantaneous multiphoton ionization rate and initial distribution of electron momentum. {\it Phys. Rev. A},  78:015405, 2008
\ifthenelse{\boolean{ElectronicVersion}}
	{[Preprint \href{http://arxiv.org/abs/0805.1890}{arXiv:0805.1890}].}
	{[Preprint arXiv:0805.1890].}
	Copyright (2008) by the American Physical Society.
\item {\bf Chapter \ref{chapter4}:} D. I. Bondar, W.-K. Liu, and M. Yu. Ivanov. Two-electron ionization in strong laser fields below intensity threshold: signatures of attosecond timing in correlated spectra. {\it Phys. Rev. A}, 79:023417, 2009
\ifthenelse{\boolean{ElectronicVersion}}
	{[Preprint \href{http://arxiv.org/abs/0809.2630}{arXiv:0809.2630}].}
	{[Preprint arXiv:0809.2630].}
	Copyright (2009) by the American Physical Society.
\item  {\bf Chapter \ref{chapter5}:} D. I. Bondar, G. L. Yudin, W.-K. Liu, M. Yu. Ivanov, and A.~D.~Bandrauk. Non-sequential double ionization below laser-intensity threshold: anticorrelation of electrons without excitation of parent ion. {\it Submitted to Phys. Rev. A}, 2010
\ifthenelse{\boolean{ElectronicVersion}}
	{[Preprint \href{http://arxiv.org/abs/1009.2072}{arXiv:1009.2072}].}
	{[Preprint arXiv:1009.2072].}
\item {\bf Chapter \ref{chapter6}:}  D. I. Bondar and W.-K. Liu. Shapes of leading tunnelling trajectories for single-electron molecular ionization. {\it Submitted to J. Phys. A}, 2010
\ifthenelse{\boolean{ElectronicVersion}}
	{[Preprint \href{http://arxiv.org/abs/1010.2668}{arXiv:1010.2668}].}
	{[Preprint arXiv:1010.2668].}
\end{itemize}
According to the Copyright FAQ of the American Physical Society, the author of this thesis has the right to use the listed above articles or portions of these articles in the thesis without requesting permission from the American Physical Society. 

\cleardoublepage


\renewcommand{\bibname}{References}
\addcontentsline{toc}{chapter}{\textbf{\bibname}}


\bibliographystyle{unsrt}
\bibliography{thesis}


\end{document}